\documentclass[11pt]{article}
\usepackage[margin=1in]{geometry}
\usepackage[utf8]{inputenc}
\usepackage{amsmath,amssymb,amsthm}
\usepackage{xcolor}
\usepackage{dsfont}
\usepackage{hyperref}
\hypersetup{colorlinks,allcolors=[rgb]{0,0.13672,0.95703}}
\usepackage{algorithm}
\usepackage{algpseudocode}
\usepackage{soul}

\definecolor{lightgrey}{cmyk}{0,0,0,0.2}
\definecolor{darkgrey}{cmyk}{0,0,0,0.6}

\usepackage{tikz,pgfplots}
\pgfplotsset{compat=newest}
\usetikzlibrary{patterns}
\usetikzlibrary{matrix,decorations.pathreplacing}
\usetikzlibrary{decorations.pathreplacing,calligraphy}

\let\origtop\top
\renewcommand\top{{\scriptscriptstyle{\origtop}}} 
\def\diag{\mathrm{diag}}
\def\nnz{\mathrm{nnz}}
\def\DPP{\mathrm{DPP}}
\def\PDPP{\mathrm{P\textnormal{-}DPP}}

\def\x{{\mathbf x}}
\def\w{{\mathbf w}}
\def\g{{\mathbf g}}
\def\e{{\mathbf e}}
\def\v{{\mathbf v}}
\def\u{{\mathbf u}}

\def\y{{\mathbf y}}
\def\ss{{\mathbf s}}
\def\Alg{{\mathsf{Alg}}}
\def\MatVec{{\mathsf{MatVec}}}

\def\b{{\mathbf b}}

\def\A{{\mathbf A}}
\def\Ac{{\mathcal A}}
\def\Lc{{\mathcal L}}
\def\B{{\mathbf B}}
\def\Sigmab{{\mathbf \Sigma}}
\def\G{{\mathbf G}}
\def\H{{\mathbf H}}
\def\I{{\mathbf I}}
\def\Q{{\mathbf Q}}
\def\zero{{\mathbf 0}}

\def\M{{\mathbf M}}

\def\S{{\mathbf S}}

\def\X{{\mathbf X}}
\def\Z{{\mathbf Z}}

\def\D{{\mathbf D}}

\def\setR{{\mathcal R}}
\def\setS{{\mathcal S}}

\def\Ec{{\mathcal E}}

\def\V{{\mathbf V}}
\def\P{{\mathbf P}}
\def\E{{\mathds E}}

\def\R{{\mathbf R}}

\def\U{{\mathbf U}}

\def\eps{{\epsilon}}
\def\sym{{\mathrm{sym}}}
\def\spec{{\mathrm{spec}}}
\def\R{{\mathds R}}

\def\Var{\mathrm{Var}}
\def\Cov{\mathrm{Cov}}

\newif\ifisfinal
\isfinaltrue

\newcommand{\dn}[1]{\ifisfinal\else\textcolor{magenta}{#1 -dn}\fi}

\def\one{{\mathds{1}}}

\newtheorem{thm}{Theorem}

\newtheorem{lem}[thm]{Lemma}

\newtheorem{cor}[thm]{Corollary}

\newtheorem{dfn}[thm]{Definition}
\newtheorem{remark}[thm]{Remark}

\usepackage{lastpage}

\title{Fine-grained Analysis and Faster Algorithms\\
for Iteratively Solving Linear Systems\thanks{The work was partially supported by NSF grant CCF-2338655 (MD), ARO grant 2003514594 (DL), NSF grant DMS-2011140 (DN), and NSF grant
DMS-2309685 (ER).}}

\author{
Micha{\l} Derezi\'nski\thanks{University of Michigan (\texttt{derezin@umich.edu})}
\quad
Daniel LeJeune\thanks{Stanford University (\texttt{daniel@dlej.net})}
\quad 
Deanna Needell\thanks{University of California, Los Angeles (\texttt{deanna@math.ucla.edu})}
\quad
Elizaveta Rebrova\thanks{Princeton University (\texttt{elre@princeton.edu})}
}
\begin{document}

\maketitle

\begin{abstract}
Despite being a key bottleneck in many machine learning tasks, the cost of solving large linear systems has proven challenging to quantify due to problem-dependent quantities such as condition numbers.
To tackle this, we
    consider a fine-grained notion of complexity for solving linear systems, which is motivated by applications where the data exhibits low-dimensional structure, including spiked covariance models and kernel machines, and when the linear system is explicitly regularized, such as ridge regression.
    
    Concretely, let $\kappa_\ell$ be the ratio between the $\ell$th largest and the smallest singular value of $n\times n$ matrix $\A$.
     We give a stochastic algorithm based on the Sketch-and-Project paradigm, that solves the linear system $\A\x=\b$, that is, finds $\tilde{\x}$ such that $\|\A\tilde{\x} - \b\| \le \epsilon \|\b\|$, in time $\tilde O(\kappa_\ell\cdot n^2\log 1/\epsilon)$ for any $\ell = O(n^{0.729})$.
This is a direct improvement over 
preconditioned conjugate gradient, and it provides a stronger separation between stochastic linear solvers and algorithms accessing $\A$ only through matrix-vector products.

Our main technical contribution is the new analysis of the first and second moments of the random projection matrix that arises in Sketch-and-Project.
\end{abstract}

\section{Introduction}
\label{s:intro}

In the era of big data, the efficient processing of massive datasets has become critically important across a wide range of areas, from machine learning and statistics to scientific computing and industrial applications. Traditional methods, and especially direct approaches, for handling such data often face significant computational challenges due to their high dimensionality and massive volume. In response to this, iterative refinement methods using randomized sampling and sketching have emerged as powerful tools for effectively solving algorithmic tasks in large-scale  machine learning and data science. Yet, the computational cost of these methods is often significantly affected by problem-dependent quantities such as condition numbers, which make it challenging to characterize how their complexity compares to, and is affected by, recent advances in algorithmic theory.

Perhaps one of the most fundamental tasks impacted by this phenomenon is solving large systems of linear equations, which has numerous applications in machine learning such as least squares \cite{dieuleveut2017harder}, kernel ridge regression \cite{alaoui2015fast}, as well as model training with Newton-type methods on both convex and non-convex objectives \cite{erdogdu2015convergence,xu2020newton}. Other applications include imaging \cite{natterer, hounsfield_CAT}, sensor networks \cite{sensors}, and scientific computing \cite{xia2010superfast, wolters2008numerical}, among others. 
In this problem, our goal is to approximately solve $\A\x=\b$, given a large data matrix $\A$ and a vector $\b$. Traditional direct approaches for solving linear systems, such as Gaussian elimination, require $O(n^3)$ time to find the exact solution when $\A$ is a dense square $n\times n$ matrix. Compared to this, deterministic iterative refinement methods, such as Richardson or Chebyshev iteration \cite{golub1961chebyshev} and Krylov methods including the Lanczos algorithm and Conjugate Gradient \cite{saad1981krylov,liesen2013krylov}, produce a sequence of estimates which gradually converge to the solution, having a much cheaper $O(n^2)$ per-iteration cost that comes typically from computing a matrix-vector product with $\A$. Introducing sub-sampling into the iterations has led to stochastic approaches such as Randomized Kaczmarz \cite{Kac37:Angenaeherte-Aufloesung, SV09:Randomized-Kaczmarz} and Randomized Coordinate Descent \cite{nesterov2012efficiency}, which have an even smaller per-iteration cost but tend to require more steps to converge. 

Another algorithmic approach, which has led to improvements in the time complexity of solving linear systems, is fast matrix multiplication. Initiated by \cite{Strassen1969}, this approach relies on the fact that we can invert an $n\times n$ matrix in the time that it takes to  multiply two such matrices. This has led to algorithms with runtime of $O(n^\omega)$, where $\omega<2.371552$ is the current exponent of matrix multiplication, which is regularly improved with continued advances in the area \cite{pan1984multiply,coppersmith1987matrix,williams2012multiplying,williams2024multiplication}. Unfortunately, except in special cases \cite{peng2021solving}, these algorithmic advances have not led to improvements in the complexity of iterative linear system solvers, due to the fundamentally sequential nature of these methods, as well as their dependence on the condition number $\kappa$.

The traditional analysis of 
iterative methods runs into a fundamental 
complexity barrier of $\tilde O(\kappa \cdot n^2)$, or $\tilde O(\sqrt{\kappa} \cdot n^2)$ in the positive definite setting, which depends on the condition number $\kappa$ of the matrix $\A$, defined as the ratio between its largest and smallest singular value: $\kappa = \sigma_{1}/\sigma_{n}$, where $\sigma_1\geq \sigma_2\geq ...\geq \sigma_n$ are the decreasing singular values of $\A$. Yet, a single  condition number does not accurately characterize the cost of iteratively solving a linear system, particularly in machine learning settings where the singular value profile of the input matrix exhibits low-dimensional structure determined by the underlying data distribution or a kernel function.

In this paper, we study the time complexity of iterative linear system solvers through a more fine-grained notion than a single condition number quantity, in a way that is particularly well-suited for machine learning and statistical settings. Concretely, we consider a \emph{parameterized} condition number $\kappa_\ell$, for $\ell\in\{1,...,n\}$, which allows the top-$\ell$ part of the spectrum to be arbitrarily ill-conditioned, while controlling the condition number of the tail of the spectrum (related notions of condition number have been considered, see Section \ref{s:related-work}). In this model, we obtain improved time complexity guarantees for solving linear systems (see Theorems~\ref{t:informal-main} and~\ref{t:informal-psd}) through a combination of new convergence analysis and efficient algorithms for stochastic iterative solvers based on the Sketch-and-Project paradigm \cite{gower2015randomized,gower2018accelerated}.
Our approach not only provides sharper \dn{I worry about statements like this..do we want to explain what we mean here, i.e. sharper than what, or maybe just say by using information on the full spectrum rather than condition numbers, something like this?} guarantees for wide classes of matrices, but also enables us to tie the complexity of iterative solvers together with the ongoing algorithmic advances in fast matrix multiplication, as well as providing a stronger complexity separation between stochastic iterative methods and classical iterative algorithms such as conjugate gradient (see Theorem \ref{t:informal-lower}).

Our fine-grained analysis is well-motivated by many established statistical models of data matrices. These models are commonly of the form ``signal+noise'', which corresponds to the intuition that most of the information is contained in a low-dimensional component of the data. A classical example of this is the spiked covariance model \cite{johnstone2001distribution,capitaine2009largest,cai2013sparse,perry2018optimality}, which describes the data as a low-rank matrix distorted by noise (i.e., $\A+\epsilon\G$).
Moreover, linear systems are often regularized before solving (i.e., $\A+\lambda\I$), either to achieve better generalization, e.g., for kernel ridge regression \cite{alaoui2015fast}, or to attain improved convergence in an optimization method, e.g., damped or cubic-regularized Newton's method \cite{nesterov2006cubic}. 

Other motivating examples include matrices arising from feature extraction techniques commonly used in machine learning, such as kernel machines \cite{Williams01Nystrom}, Gaussian processes \cite{RasmussenWilliams06}, and random features \cite{rahimi2007random}, which lead to well-understood polynomial or exponential singular value decay profiles \cite{burt2019rates,Santa97gaussianregression,RasmussenWilliams06}. For example, our new results imply improved runtimes for solving linear systems with a polynomial spectral decay $\sigma_i\simeq i^{-\beta}$ (see Corollary \ref{c:polynomial}), which arise when using certain kernel functions (e.g., Mat\'ern, \cite{RasmussenWilliams06}), and also, as a result of the power law.

\subsection{Main results}

Our main result shows that the $\tilde O(\kappa\cdot n^2)$ time complexity of iteratively solving a linear system can be directly improved by replacing the condition number $\kappa = \sigma_1/\sigma_n$ with the \emph{spectral tail condition number}, $\kappa_{\ell}:=\sigma_\ell/\sigma_n$ where $\sigma_\ell$ is the $\ell$th largest singular value of $\A$. Using fast matrix multiplication, we show that this can be achieved with any $\ell=O(n^{0.729})$.

\begin{thm}\label{t:informal-main}
    Consider an $n\times n$ matrix $\A$ with singular values $\sigma_1\geq \sigma_2\geq ...\geq\sigma_{n}>0$ and an $n$-dimensional vector $\b$. We can compute $\tilde\x$ such that $\|\A \tilde\x-\b\|\leq \epsilon\|\b\|$ in time:
    \begin{align*}
        \tilde O\bigg(\frac{\sigma_{\ell}}{\sigma_n}\cdot n^2\log1/\epsilon\bigg)\quad\text{for any}\quad \ell=O\big(n^{\frac1{\omega-1}}\big)=O(n^{0.729}).
    \end{align*}
\end{thm}
Note that the complexity improvement in this result is already significant without resorting to fast matrix multiplication, in which case, we can use any $\ell=O(\sqrt n)$. Moreover, unlike with the usual condition number $\kappa$, any future improvements in the matrix multiplication exponent $\omega$ will directly lead to improvements in the exponent of $\ell$ in $\kappa_\ell$. As $\omega$ approaches $2$, similarly $\ell$ approaches $n$, until the two paradigms (potentially) meet at $\tilde O(n^2)$.

To underline that this fine-grained notion of complexity is natural, we show that the above guarantee holds for a remarkably simple stochastic iterative solver (see Algorithm~\ref{alg:main}), a combination of two standard techniques: Sketch-and-Project \cite{gower2015randomized}, which is a randomized iterative framework for solving linear systems that can be viewed as a simple extension of the classical Kaczmarz algorithm \cite{Kac37:Angenaeherte-Aufloesung,SV09:Randomized-Kaczmarz}; and Nesterov's acceleration technique \cite{nesterov1983method}, variants of which have been used in numerous iterative optimization methods \cite{liu2016accelerated,ye2020nesterov,even2021continuized}. While this combination is not new \cite{gower2018accelerated}, its convergence has proven challenging to quantify in terms of the spectral properties~of~$\A$. 

We address this challenge by bridging ultra-sparse sketching techniques \cite{chenakkod2023optimal} with a new convergence analysis that builds on recent advances in combinatorial sampling \cite{derezinski2020improved} and random matrix theory \cite{brailovskaya2024universality}. Along the way, we establish technical results that are likely of independent interest, including a new characterization of the smallest singular value for the family of random matrices obtained via sparse sketching (Lemma~\ref{l:smin-less}).

\smallskip

\noindent
\textbf{Positive definite systems.}
Our algorithmic framework can be adapted to take advantage of certain types of additional structure present in the matrix $\A$. As an example, we show that, for positive definite matrices, the dependence on the spectral tail condition number $\kappa_\ell$ can be improved to a square root, just as regular condition number $\kappa$ can be improved to a square root for deterministic iterative methods~\cite{axelsson1986rate}.
\begin{thm}\label{t:informal-psd}
    Consider an $n\times n$ positive definite $\A$ with singular values $\sigma_1\geq \sigma_2\geq ...\geq\sigma_{n}$ and an $n$-dimensional vector $\b$. We can compute $\tilde\x$ such that $\|\A \tilde\x-\b\|\leq \epsilon\|\b\|$ in time:
    \begin{align*}
        \tilde O\bigg(\sqrt{\frac{\sigma_{\ell}}{\sigma_n}}\cdot n^2\log1/\epsilon\bigg)\quad\text{for any}\quad \ell=O\big(n^{\frac1{\omega-1}}\big)=O(n^{0.729}).
    \end{align*}
\end{thm}
\smallskip

\noindent
\textbf{Separation from matrix-vector query methods.}
An important implication of our results, and of our reliance on the spectral tail condition number, is a more distinct than previously known separation (in terms of complexity) between stochastic solvers, such as Sketch-and-Project, and classical iterative solvers, which interact with $\A$ only through matrix-vector product queries. Here, representing state-of-the-art matrix-vector query solvers, let us consider Krylov methods, such as Conjugate Gradient (CG, \cite{hestenes1952methods}) and its non-symmetric extensions \cite{saad1981krylov}. A careful convergence analysis of Krylov iterations shows that, given any $\ell$ and assuming exact precision arithmetic, they converge $\epsilon$-close in $O(n^2\ell+\kappa_\ell \cdot n^2\log1/\epsilon)$ time for general linear systems, and in $O(n^2\ell+\sqrt{\kappa_\ell}\cdot n^2\log1/\epsilon)$ time for positive definite systems \cite{axelsson1986rate,sw09}. In both cases, the second term matches our complexity in Theorems \ref{t:informal-main} and \ref{t:informal-psd}, but the Krylov methods suffer an additional cost of $O(n^2\ell)$ to take care of the top $\ell$ singular values of $\A$. 

In fact, we show that this additional cost is a fundamental bottleneck of all methods based on the matrix-vector query model. Specifically, we give a lower bound showing that any algorithm which accesses the matrix $\A$ only through matrix-vector product queries 
must incur a worst-case cost of $\tilde\Omega(n^2\ell + \sqrt{\kappa_\ell}\cdot n^2)$ arithmetic operations for dense positive definite linear systems. See Section \ref{s:lower} for a detailed discussion.

\begin{thm}\label{t:informal-lower} 
    Given $n$, $\ell<n$, and $\kappa_\ell\geq 1$, consider the task of solving an $n\times n$ positive definite linear system $\A\x=\b$ with $\frac{\sigma_{\ell}(\A)}{\sigma_n(\A)}\leq\kappa_\ell$. Any algorithm which interacts with $\A$ via adaptive randomized queries  of the form $\v\rightarrow \A\v$, and solves this task to high precision, has query complexity at least $\tilde\Omega(\ell+\sqrt{\kappa_\ell})$, which for dense systems leads to time complexity at least $\tilde\Omega(n^2\ell + \sqrt{\kappa_\ell}\cdot n^2)$.
\end{thm}
 Note that, in addition to deterministic solvers like CG, this lower bound also includes randomized preconditioning methods \cite{martinsson2020randomized}, which probe matrix $\A$ with Gaussian vectors to construct a low-rank approximation (essentially, estimating the top part of its spectrum), and use that information to speed up an iterative solver like CG. Once we go beyond the matrix-vector query model, this preconditioning approach can be sped up by blocking the queries together and using fast $n\times n$ by $n\times \tilde O(\ell)$ rectangular matrix multiplication \cite{le2012faster} to approximate the top part of the spectrum faster. However, even then, the cost is still $\Omega(n^{2+\theta})$ with some $\theta > 0$ for any $\ell=\Omega(n^{0.33})$. In comparison, Theorems~\ref{t:informal-main} and \ref{t:informal-psd} show that Sketch-and-Project with Nesterov's acceleration avoids this cost entirely for~any~$\ell=O(n^{0.729})$. 

\paragraph{Application: Polynomial spectral decay.}
To illustrate how our new results improve on the time complexity of solving linear systems for real-world matrices, we consider matrices with polynomial spectral decay (power law distribution). Concretely, consider a matrix $\A$ with singular values $\sigma_i = \Theta(i^{-\beta}\sigma_1)$ for some positive constant $\beta$. Such behavior can naturally occur in data \cite{clauset2009power,eikmeier2017revisiting}, or it can arise when $\A$ is the kernel matrix of a dataset, i.e., its $(i,j)$th entry is $k(x_i,x_j)$ where $x_i$ is the $i$th data point and $k(\cdot,\cdot)$ is a kernel function. For example, \cite{RasmussenWilliams06} showed that if we use the Mat\'ern kernel function with parameter $\nu>0$, then matrix $\A$ exhibits polynomial decay with $\beta = 2\nu+1$. Solving linear systems with these matrices is the main computational cost in kernel ridge regression, Gaussian processes, and independent component analysis, among others.
Our results imply the following guarantee for solving linear systems with polynomial decay. 
\begin{cor}\label{c:polynomial}
Consider an $n\times n$ matrix $\A$ with singular values $\sigma_1\geq \sigma_2\geq ...\geq \sigma_n$ and an $n$-dimensional vector $\b$. Suppose that $c_1i^{-\beta}\leq \sigma_i\leq c_2i^{-\beta}$ for some absolute constants $0<c_1<c_2$ and $\beta\geq 0.5$. Then, we can compute $\tilde\x$ such that $\|\A\tilde\x-\b\|\leq\epsilon\|\b\|$ in time:
\begin{align*}
    \tilde O\Big(n^{2+\frac{\omega-2}{\omega-1}(\beta-0.5)}\log1/\epsilon\Big).
\end{align*}
\end{cor}
\begin{remark}
This is better than the best known guarantee in the matrix-vector query model \cite{sw09} for any $\beta\in(0.5, 3.75)$, 
and it improves on the prior best known time complexity in the unrestricted model for any $\beta\in(0.5,1.33)$. For example, if we choose $\beta=1$, then we get $\tilde O(n^{2.135})$ runtime, whereas the best prior result (obtained by preconditioning an SVRG-type stochastic solver, see \cite{gonen2016solving,musco2018spectrum}) achieves $\tilde O(n^{2.178})$, and CG obtains $\tilde O(n^{2.667})$. For a detailed discussion see Appendix \ref{a:polynomial}.
\end{remark}

\paragraph{Extensions to sparse least squares.} While our algorithms are stated for consistent linear systems, they can be easily extended to the general inconsistent setting, i.e., least squares regression. This can be achieved by incorporating a variant of our Sketch-and-Project algorithm as an inner solver within a sketch-to-precondition type algorithm \cite{rokhlin2008fast}, obtaining nearly-input-sparsity time algorithms for a tall least squares task parameterized by the spectral tail condition number $\kappa_\ell$. 
For example, given an $n\times d$ matrix $\A$ and an $n$-dimensional vector $\b$, we can find $\tilde \x$ such that $\|\A\tilde\x-\b\|^2\leq\min_\x\|\A\x-\b\|^2+\epsilon\|\b\|^2$ in $\tilde O((\nnz(\A) + d^2\kappa_\ell)\log1/\epsilon)$ time, where $\nnz(\A)$ is the number of non-zeros in $\A$ and $\kappa_\ell = \sigma_\ell(\A)/\sigma_d(\A)$ for $\ell = O(d^{0.729})$. 
See Section \ref{s:completing} for further discussion.

\vspace{-3mm}
\subsection{Main technical contributions}

Our algorithms are built on the Sketch-and-Project framework \cite{gower2015randomized}, an extension of the Randomized Kaczmarz algorithm \cite{SV09:Randomized-Kaczmarz}. The core idea of Sketch-and-Project is that in each iteration we construct a small randomized sketch of the linear system, and then project the current iterate onto the set of solutions of that sketch. Formally, given an $n\times n$ linear system $\A\x=\b$ and a current iterate $\x_t$, we generate a new random $k\times n$ sketching matrix $\S=\S(t)$, where $k\ll n$ is the sketch size, and then consider a smaller linear system $\S\A\x=\S\b$. For example, $\S$ can be a subsampling matrix which selects $k$ out of $n$ linear equations (rows of $\A$). In the original Sketch-and-Project algorithm, we would now compute $\x_{t+1}$ as the projection of $\x_t$ onto the subspace defined by the smaller system. A natural extension of this is to use Nesterov's acceleration scheme \cite{gower2018accelerated}, which introduces two additional sequences $\y_t$ and $\v_t$ to mix the projection step with the current trajectory of the iterates (see Algorithm \ref{alg:main}).

While various Sketch-and-Project algorithms have been proposed and studied across a long line of works, such as  \cite{hanzely2018sega,necoara2019randomized,lejeune2024asymptotics,tang2023sketch}, their complexity analysis has proven remarkably challenging. Crucially, their convergence rate depends on the spectral properties of a random projection matrix defined by the sketching matrix $\S$, namely $\P := (\S\A)^\dagger\S\A$, which is a projection onto the subspace spanned by the rows of the sketch $\S\A$.  Matrix $\P$ is also central to other applications of sketching, such as low-rank approximation \cite{halko2011finding,cohen2015dimensionality,cohen2017input}, and so understanding its spectral properties is of significant independent interest across the literature on randomized linear algebra \cite{woodruff2014sketching,martinsson2020randomized}. Our main technical contributions are new sharp characterizations of the first and second matrix moments of $\P$ for a class of ultra-sparse sketching matrices $\S$, as explained~below.

Sketch-and-Project with Nesterov's acceleration has been proposed in \cite{richtarik2020stochastic} and previously studied by \cite{gower2018accelerated}. However, these papers characterize the convergence rate of the method only in terms of the properties of the random projection $\P$, without identifying how these properties are determined by the spectrum of the input matrix $\A$, or the choice of the sketching matrix $\S$. Specifically, they showed that, as long as the matrix $\bar\P=\E[\P]$ is invertible, then the iterates converge to the solution $\x^*=\A^{-1}\b$, with the expected convergence rate given by
\begin{align}
\E\big[\|\x_t-\x^*\|^2\big] &\leq 2\bigg(1-\sqrt{\frac\mu\nu}\,\bigg)^t\|\x_0-\x^*\|^2,
\label{eq:convergence-intro}
\\
\text{where}\quad    \mu &= \lambda_{\min}(\E[\P])
    \quad\text{and}\quad
    \nu = \lambda_{\max}\Big(\E\big[ (\bar\P^{-1/2}\P\bar\P^{-1/2})^2\big]\Big).\nonumber
\end{align}
Here, we can think of $\mu$ and $\nu$ as the first and second moment properties of the random projection matrix $\P$, since $\mu$ requires lower-bounding its first matrix moment $\bar\P=\E[\P]$, whereas $\nu$ requires upper-bounding its normalized second matrix moment $\E[(\bar\P^{-1/2}\P\bar\P^{-1/2})^2]$. The core challenge in obtaining these guarantees is that they are fundamentally average-case properties that cannot be recovered by usual high-dimensional concentration techniques (consider that we seek a bound on $\lambda_{\min}(\E[\P])$ when the matrix $\P$  is not even invertible).

While the characterization of the second-order term $\nu$ presented in this paper appears to be the first such guarantee, there has been a number of works analyzing the complexity of Sketch-and-Project \cite{rk20,mdk20,derezinski2024sharp,derezinski2023solving} based on quantifying only the first-order term~$\mu$. These works do not consider Nesterov's acceleration and rely on a weaker convergence guarantee that follows from  \eqref{eq:convergence-intro} by observing that $1-\sqrt{\mu/\nu}\leq 1-\mu$. This allows one to only have to lower bound $\mu$, while avoiding the second matrix moment of $\P$, but it also leads to sub-optimal condition number dependence.

First such guarantees for $\mu$ were obtained only recently: initially for specialized combinatorial sub-sampling matrices $\S$ based on determinantal point processes (DPPs, see \cite{rk20,mdk20}); then,  for dense Gaussian sketching matrices \cite{rebrova2021block}, as well as sub-Gaussian matrices with some degree of sparsity \cite{derezinski2024sharp}, and asymptotically free sketching matrices in the asymptotic regime~\cite{lejeune2024asymptotics}. 
All of these sketching/sub-sampling approaches are still too expensive to yield fast algorithms, however the resulting guarantees exposed a connection between the first moment of $\P$ and the spectral tail condition number $\kappa_\ell$ when the sketch size $k$ is proportional to $\ell$.

Finally, \cite{derezinski2023solving} gave a reduction showing that after preprocessing the matrix $\A$ with a randomized Hadamard transform \cite{ailon2009fast}, a uniform sub-sampling matrix $\S$ of size $k$ yields a guarantee on $\mu$ that is at least as good as a DPP sub-sampling matrix of size $\ell = \Omega(k/\log^3 n)$. 
This led to the first \emph{efficient} sub-sampling scheme with a convergence guarantee, although requiring the sample size to be larger by a factor of $O(\log^3 n)$ than what we might hope to achieve.

\smallskip

\noindent
\textbf{Technical contribution 1:} \emph{First-moment projection analysis via optimal DPP reduction.} 
In our first main technical contribution, we directly improve on the existing guarantees for the first matrix moment of the projection matrix $\P$, both for uniform sub-sampling matrices $\S$, as well as when using ultra-sparse sketching matrices. Our general strategy is similar to that of \cite{derezinski2023solving}, in that we also show a reduction from uniform to DPP sub-sampling, however the actual reduction is entirely different 
(it involves designing a custom DPP sampling algorithm that is then coupled with a uniform sampling oracle; see Section \ref{s:mu} for details). Overall we show that, after preprocessing $\A$ with a randomized Hadamard transform, a uniform sub-sample of size $k$ yields the same bound on the first moment of $\P$ as a DPP sample of size $\ell=\Omega(k/\log k)$, which closes the gap in the result of \cite{derezinski2023solving} from $O(\log^3n)$ to $O(\log k)$. We note that standard lower bounds based on the coupon collector problem \cite{tropp2011improved} suggest that the factor $O(\log k)$ is unavoidable, meaning that our reduction size of $\Omega(k/\log k)$ is likely optimal.

Moreover,  our first-moment analysis provides a lower bound on the entire matrix moment of $\P$, rather than just its smallest eigenvalue, which proves crucial for bounding $\nu$ later on. This result can be informally stated as follows: if $\P$ is produced from an $n\times n$ matrix $\A$ using a sub-sampling or sketching matrix $\S$ of size $k\times n$, then:
\begin{align*}
\text{(Lemma \ref{l:mu-less})}\quad    \E[\P] \ \succeq_\delta \ \A^\top\A(\A^\top\A+\lambda\I)^{-1},\quad 
    \text{where}\quad\lambda = \frac1\ell\sum_{i>\ell}\sigma_i^2\quad\text{for}\quad\ell = \Omega\Big(\frac k{\log k}\Big),
\end{align*}
where $\succeq_\delta$ hides an additive $\delta \I$.
This result in particular implies a lower bound on $\mu$, which roughly states that $\mu = \Omega(\ell/(n\kappa_\ell^2))$. As an immediate corollary of our new first-moment analysis of the random projection matrix~$\P$, we can sharpen the complexity of solving linear systems with $\ell$ large singular values, i.e., those where $\kappa_\ell = O(1)$, the setting considered by \cite{derezinski2023solving}. For comparison, they obtained $O((n^2\log^4n+n\ell^{\omega-1}\log^{3\omega}n)\log1/\epsilon)$ for this problem.
\begin{cor}\label{t:large-sv}
    Consider an $n\times n$ matrix $\A$ with at most $\ell$ singular values larger than $O(1)$ times its smallest one, and vector $\b$. We can compute $\tilde\x$ such that $\|\A\x-\b\|\leq\epsilon\|\b\|$ in:
    \begin{align*}
        O\Big((n^2\log^2n + n\ell^{\omega-1}\log^\omega \ell)\log1/\epsilon\Big)\quad\text{time}.
    \end{align*}
\end{cor}

\smallskip

\noindent
\textbf{Technical contribution 2:} \emph{Second-moment projection analysis via Gaussian universality.} 
In our second main technical contribution, which is arguably the more significant of the two, we give the first non-trivial upper bound on the normalized second matrix moment of the random projection matrix $\P$. Here, our analysis crucially builds on recent breakthroughs in non-asymptotic random matrix theory \cite{brailovskaya2024universality}, which show that the singular values of certain classes of random matrices are close to the singular values of a Gaussian matrix with matching entry-wise mean and covariance (we refer to this property as \emph{Gaussian universality}). In our case, Gaussian universality is required for the random sketch $\S\A$, where $\S$ is a sparse sketching matrix.
Fortunately, we are able to show that even extremely sparse  sketching matrices \cite{derezinski2021newton,chenakkod2023optimal}, for which computing $\S\A$ costs only $\tilde O(nk)$, are sufficient for this to hold. We use this to show that when an $n\times n$ matrix $\A$ is preprocessed with a randomized Hadamard transform, then the resulting projection $\P=(\S\A)^\dagger\S\A$ satisfies:
\begin{align*}
 \text{(Lemma \ref{l:nu-gaussian})}\quad   \E\big[\P\bar\P^{-1}\P\big] \ \precsim\  \frac n\ell\cdot\bar\P,\quad\text{for}\quad\ell = \Omega\Big(\frac k{\log k}\Big),
\end{align*}
where recall that $\bar\P=\E[\P]$ (we refer to Lemma \ref{l:nu-gaussian} for a formal statement).
Note that the result directly implies a bound on the largest eigenvalue of $\E[(\bar\P^{-1/2}\P\bar\P^{-1/2})^2]$, leading to $\nu\lesssim n/\ell$. Putting this together with the bound on $\mu$, with a careful choice of the parameters we obtain $\sqrt{\mu/\nu}\gtrsim \ell/(n\kappa_\ell)$, which can be directly plugged into  the convergence rate of Sketch-and-Project with Nesterov's acceleration \eqref{eq:convergence-intro}.

The key reason we rely on Gaussian universality as part of the proof of Lemma \ref{l:nu-gaussian} is to establish the following 
lower bound on the smallest singular value for $k\times n$ sketches $\S\A$, which should be of independent interest: 
\begin{align*}
 \text{(Lemma \ref{l:smin-less})}\quad   \sigma_{\min}^2(\S\A) = \Omega\Big(2k\sigma_{2k}^2+\sum_{i>2k}\sigma_i^2\Big)\quad\text{with high probability,}
\end{align*}
where recall that $\sigma_1\geq \sigma_2\geq ...$ are the singular values of $\A$.
Remarkably, this bound appears to be new even for dense Gaussian and sub-Gaussian sketching matrices~$\S$, but crucially, we are able to show this also for the above mentioned sparse sketching matrices.

\smallskip

\noindent
\textbf{Completing the complexity analysis.} 
To recover the time complexity in Theorem \ref{t:informal-main}, we must also address the cost of the projection step in each iteration of the algorithm. This step effectively involves solving a very wide $k\times n$ linear system, which can be done inexactly in time $\tilde O( nk + k^\omega)$ by relying on standard sketch-and-precondition techniques \cite{rokhlin2008fast,msm14,woodruff2014sketching}. This forces us to extend the existing convergence analysis of Sketch-and-Project with Nesterov's acceleration \cite{gower2018accelerated} to allow inexact projections, which is done in Section \ref{s:inexact}. Putting everything together and choosing sketch size $k=\tilde O(\ell)$, our convergence analysis shows that the algorithm requires $\tilde O(\frac n\ell \kappa_\ell\cdot\log1/\epsilon)$ iterations, each costing $\tilde O(n\ell + \ell^\omega)$ time. Setting $\ell = \tilde O(n^{\frac1{\omega-1}})$ yields the desired time complexity.

\section{Related work}\label{s:related-work}
Next, we give additional background and related work, focusing on different approaches to iteratively solving linear systems, as well as connections between our techniques and other areas such as matrix sketching, combinatorial sampling and random matrix theory.
\smallskip

\noindent
\textbf{Complexity of iterative linear solvers.} A number of prior works have given sharp characterizations of the convergence of iterative linear system solvers, going beyond the usual condition number $\kappa$ in some way, which makes them relevant to our discussion. In the context of Krylov methods such as CG, early works such as \cite{axelsson1986rate} showed that its convergence rate depends on the ``clusterability'' of the spectrum of $\A$, which has been restated in terms of~$\kappa_\ell$, as discussed earlier (see, e.g., Theorem 2.5 in \cite{sw09}). 

For stochastic iterative methods, such as randomized Kaczmarz and randomized coordinate descent, a number of \emph{averaged} notions of the condition number have been used \cite{SV09:Randomized-Kaczmarz,leventhal2010randomized,liu2016accelerated,bollapragada2024fast}, which depend in a more complex way on the entire spectrum, and may also lead to sharper bounds than $\kappa$. We recover similar averaged condition numbers in our analysis (Theorems \ref{thm:main} and \ref{thm:psd}). This has been used to suggest that stochastic methods can be faster than CG \cite{lee2013efficient}, although without providing fine-grained upper/lower bounds, such as those given here based on the parameter $\ell$.
A few works have also considered versions of spectral tail condition numbers, similar to $\kappa_\ell$, either as part of the analysis \cite{gonen2016solving,musco2018spectrum} or the assumptions \cite{derezinski2023solving}. Some acceleration schemes other than Nesterov's momentum have been proposed for stochastic solvers \cite{lin2015universal,frostig2015regularizing,allen2018katyusha,bollapragada2024fast, alderman2024randomized}, which can give sharper convergence guarantees in some cases, at the expense of additional computational overhead. Remarkably, we are able to show our results for Nesterov's original scheme, which is both simple and practical.

\smallskip

\noindent
\textbf{Matrix sketching and sampling.} Randomized techniques for sketching or subsampling matrices have been developed largely as part of the area known as Randomized Numerical Linear Algebra \cite{woodruff2014sketching,drineas2016randnla,martinsson2020randomized,derezinski2024recent}, with applications in least squares \cite{sarlos2006improved,rokhlin2008fast} and low-rank approximation \cite{cohen2015dimensionality,cohen2017input}, among others. In this context, the sketched projection matrix $\P=(\S\A)^\dagger\S\A$, which we study as part of our analysis, arises naturally in low-rank approximation when bounding error of the form $\|\A - \A\P\|$, where $\A\P$ is the projection of the input matrix $\A$ onto the rank $k$ span of the sketch $\S\A$ \cite{halko2011finding}. 

Some of the most popular and efficient sketching techniques are sparse random matrices, including CountSketch \cite{charikar2004finding}, OSNAP \cite{nelson2013osnap}, and LESS \cite{derezinski2021sparse}. In our setting, we require sketching matrices that are even sparser than typically used, since they are applied repeatedly, so we rely on recently proposed LESS-uniform sketches \cite{derezinski2021newton,chenakkod2023optimal} in conjunction with the randomized Hadamard transform~\cite{ailon2009fast,tropp2011improved}. 

Most of the popular matrix sub-sampling methods are based on i.i.d.\ importance sampling, for example using the so-called leverage scores \cite{drineas2006sampling,drineas2012fast}. More recently, a number of works have used non-i.i.d.\ combinatorial sampling techniques, such as volume sampling and determinantal point processes, for solving matrix problems including least squares and low-rank approximation (for an overview, see \cite{derezinski2021determinantal}). We build on some of these approaches internally as part of our analysis (Section~\ref{s:mu}).

\smallskip

\noindent
\textbf{Sketch-and-Project.} The sketch-and-project method was developed as a unified framework  for iteratively solving linear systems \cite{gower2015randomized}. Varying the distribution of the sketching matrix $\S$ and the projection type determined by a positive definite matrix $\B$, this general update rule reduces to a variety of classical methods such as  randomized Gaussian pursuit, Randomized Kaczmarz, and Randomized Newton methods, among others (e.g., \cite{gower2019rsn, gower2015randomized}). The applicability of the sketch-and-project framework spans beyond linear solvers, including extensions  aimed at solving linear and convex feasibility problems \cite{necoara2019randomized}, efficiently minimizing convex functions \cite{gower2019rsn,hanzely2020stochastic}, solving nonlinear equations \cite{yuan2022sketched} and inverting matrices \cite{gower2017randomized}. 

The performance of generic sketch-and-project type methods is often hindered by the cost of storing and applying the sketching matrix, as well as the complexity of solving the sketched system (projection step), especially when the sketch size is not very small. Special features of the sketch-and-project based algorithm presented in this work (Algorithm~\ref{alg:main})---such as ultra-sparse sketching matrices, momentum, and solving sketched systems iteratively---take care of the mentioned inefficiencies of the original sketch-and-project method and can be employed to refine the sketch-and-project methods beyond linear solvers; partial results in this direction were obtained in \cite{hanzely2018sega,gazagnadou2022ridgesketch,derezinski2021sparse,derezinski2023solving}. 

\smallskip

\noindent
\textbf{Random matrix theory.}
Due to the random nature of the sketching matrix $\S$, the analysis of Sketch-and-Project is greatly facilitated by advances in random matrix theory. In the asymptotic setting, when $\S$ has i.i.d.\ entries, the spectrum of $\S\A$ converges deterministically to a generalized Marchenko--Pastur distribution~\cite{silverstein1995empirical}. This spectral characterization enables a precise understanding of the projection matrix $\P$ and its expectation $\bar \P$ through deterministic equivalents in the asymptotic regime~\cite{hachem2007deterministic,rubio2011spectral}, recently extended beyond the i.i.d.\ setting to asymptotically free sketches by \cite{lejeune2024asymptotics}, who also applied this asymptotic analysis to Sketch-and-Project.

In finite dimensions, there are a variety of known results on matrix concentration \cite{vershynin_2018}. In particular, the largest and smallest singular values of i.i.d.\  random matrices are well characterized, e.g., \cite{rudelson2009smallest,oymak2018universality}. For other types of random matrices, recent universality results~\cite{brailovskaya2024universality} prove that a broad class of sparse and non-i.i.d.\ matrices concentrate around the same spectral profile as Gaussian matrices, enabling the application of established results to new types of extremely sparse sketches~\cite{chenakkod2023optimal}.

\section{Preliminaries}
\textbf{Notation.} We denote matrices with upper case bolded font ($\X$), vectors with lower case bolded font ($\x$), and scalars with basic font ($x$). 
The vector $\e_i$ denotes the $i$th coordinate vector. We denote by $[m]$ the set of integers from $1$ to $m$, $[m] := \{1, 2, \ldots m\}$. We denote expectation by $\E$ and probability by $\Pr$, where the randomness should be clear from context if not specified. The smallest and largest (non-zero) singular values of a matrix $\X$ are denoted by $\sigma_{\min}(\X)$ and $\sigma_{\max}(\X)$, respectively.  The norm $\|\cdot\|$ denotes the Euclidean or spectral norm. The condition number of a matrix $\X$ is denoted by $\kappa(\X) = \sigma_{\max}(\X)\sigma^{-1}_{\min}(\X)$. For a symmetric positive definite matrix $\A$, its induced norm is defined as $\|\x\|_\A = \sqrt{\x^\top \A\x},$ and we write $\A \preceq \B$ to denote the (Loewner) matrix ordering, i.e., to say that $\B - \A$ is positive semi-definite. For an event $\Ec$, we write $\neg\Ec$ to denote the complement of $\Ec$. Lastly, we utilize constants $c, c_1, c_2, \ldots, C, C_1, C_2, \ldots$ to denote absolute constants, which may hold different values from one instance to the next. We reserve the use of big-$O$ notation for complexity remarks, and we use the notation $\tilde O$ to hide logarithmic factors. To avoid difficulty in implicitly adjusting constants, we assume that $m, n, k, \ell, \delta$ are bounded away from 1. 

\smallskip

\noindent
\textbf{Sketching model.} We obtain fast solvers for our main results with Algorithm~\ref{alg:main}, which is based on the Sketch-and-Project framework \cite{gower2015randomized}, with ultra-sparse sketching matrices $\S$ and Nesterov's acceleration. For ultra-sparse sketching, we consider the following sparse sketching construction, following similar constructions in \cite{derezinski2021sparse,chenakkod2023optimal}.

\begin{dfn}[LESS-uniform]\label{d:less}
A sparse sketching matrix $\S\in\R^{k\times m}$ is called a LESS-uniform matrix with $s$ non-zeros per row if it is constructed out of independent and identically distributed (i.i.d.) rows, with the $i$th row of $\S$ given by
  \begin{align*}
    \sqrt{\frac ms}\sum_{j=1}^s r_{ij}\e_{I_{ij}}^\top,
  \end{align*}
  where $I_{i1},...,I_{is}$ are i.i.d.\ random indices sampled uniformly with replacement from $[m]$, and $r_{ij}$ are i.i.d.\
  Rademacher random variables, i.e., $\Pr(r_{ij}=1)=\Pr(r_{ij}=-1)=1/2$.
\end{dfn}
\noindent
This is one of the most natural ways to generate sparse sketching matrices. We can avoid more expensive techniques used in earlier works \cite{derezinski2024sharp} thanks to initial preprocessing with the randomized Hadamard transform~\cite{ailon2009fast}.
\begin{dfn}[Randomized Hadamard Transform]\label{d:rht}
    For any $m$ that is a power of 2, the Hadamard matrix $\H_m\in\R^{m\times m}$ is defined so that $\H_0 = 1$ and:
    \begin{align*}
        \H_m=\begin{bmatrix}
            \H_{m/2}&\H_{m/2}\\
            \H_{m/2}&-\H_{m/2}
        \end{bmatrix}.
    \end{align*}
    An $m\times m$ randomized Hadamard transform 
    (RHT) is a random matrix $\Q=\frac1{\sqrt m}\H_m\D$, where $\H_m$ is the Hadamard matrix and $\D$ is an $m\times m$ diagonal matrix with i.i.d.~Rademacher entries. Applying $\Q$ to a vector takes $O(m\log m)$ time.
\end{dfn}
\noindent
An advantage of the RHT is that it can be used to uniformize the leverage scores of a matrix, which is crucial in our analysis of both the first and second matrix moments of the random~projection~$\P$.
\begin{lem}[Lemma 3.3, \cite{tropp2011improved}]\label{l:rht}
    Consider a matrix $\U\in\R^{m\times d}$ such that $\U^\top\U=\I$ and let $\Q$ be the $m\times m$ Randomized Hadamard Transform. Then, the matrix $\tilde\U=\Q\U$ with probability $1-\delta$ has nearly uniform leverage scores,
    i.e., the norms of its rows are bounded as $\|\tilde\u_i\| \leq \sqrt{d/m}+\sqrt{8\log(m/\delta))/m}$ for all $i = 1, \ldots, m$. 
\end{lem}

As an auxiliary result, we obtain new bounds on the extreme singular values of the random matrix $\S\U$, where $\S$ is a $k\times m$ LESS-uniform embedding and $\U$ is an $m\times d$ isometric embedding matrix, i.e., such that $\U^\top\U=\I$.
Specifically, in Appendix \ref{sec:gauss-univ} we show that $\S\U$ matrices have extreme singular values of the same order as after Gaussian sketching. 
\begin{lem}\label{l:universality}
      Consider an $m\times d$ matrix $\U$ such that $\U^\top\U=\I$ and such that the norm of each row of $\U$ is bounded by $C\sqrt{d/m}$. There are universal constants $c_1,c_2,C_1,C_2>0$ such that if $\S$ is a $k\times m$ LESS-uniform matrix with $1\leq k\leq d/2$, $s\ge C_1\log^4(d/\delta)$ non-zeros per row, and dimension $d\geq C_2\log(1/\delta)$, then with probability $1-\delta$
      \begin{align*}
          c_1\sqrt d\leq \sigma_{\min}(\S\U)\leq \sigma_{\max}(\S\U)\leq c_2 \sqrt d.
      \end{align*}
\end{lem}
\noindent
 We use the upper bound from Lemma \ref{l:universality} in our first-moment projection analysis (Section~\ref{s:mu}), while the lower bound is needed for the second-moment analysis (Section~\ref{s:nu}). Crucially, a similar result does not hold if we replace the LESS-uniform matrix with a uniform sub-sampling matrix (the smallest singular value of a uniformly sub-sampled $\U$ can be arbitrarily close to zero with a non-negligible probability when the sample size $k$ is less than $d$).

\smallskip

\noindent
\textbf{Numerical stability.} Even though, for simplicity, our results are stated in the exact arithmetic (Real-RAM) model, we provide a stability analysis of our algorithm in Section~\ref{s:inexact}, which shows that its convergence rate is stable with respect to the error in the projection steps. 
Extending this analysis to the Word-RAM model with word sizes polylogarithmic in the parameters of the problem is a promising direction, since all other linear algebraic operations we rely on are known to be stable \cite{demmel2007fast}. 
This suggests another potential advantage of Sketch-and-Project type solvers over deterministic Krylov methods, for which ensuring stability is an ongoing area of research \cite{musco2018stability,peng2021solving}.

\section{Main Algorithm and Convergence Guarantees}

Our main results consist of convergence theorems describing the iteration-wise convergence rate of Algorithm~\ref{alg:main}, revealing sharp dependencies on the spectral profile of matrix $\A$. Here, note that the algorithm takes as an additional parameter positive definite matrix $\B$, which controls the type of projection operator used in Sketch-and-Project.

\begin{algorithm}[!ht]
\caption{Sketch-and-Project with Nesterov's acceleration}
\label{alg:main}
\begin{algorithmic}[1]
\State \textbf{Input: }matrix $\A\in\R^{m\times n}$, vector $\b \in \R^m$, p.d.\ matrix $\B\in\R^{n\times n}$, sketch size $k$, iterate $\x_0$, iteration $t_{\max}$, $\beta=1-\sqrt{\frac{\mu}{\nu}}$, $\gamma=\frac1{\sqrt{\mu\nu}}$, $\alpha=\frac1{1+\gamma\nu}$;
\State $\v_0 \leftarrow\x_0$;
\For{$t = 0, 1, \ldots, (t_{\max}-1)$}
\State $\y_t \leftarrow \alpha \v_t + (1-\alpha)\x_t$;
\State Generate sketching matrix $\S$;
\State Solve $(\S\A\B^{-1}\A^\top\S^\top) \u_t = \S\A\y_t-\S\b$ for $\u_t$
\Comment{Possibly inexactly, see Section \ref{s:inexact}}
\label{line-u}
\State $\w_t \leftarrow \B^{-1}\A^\top\S^\top\u_t$; \label{line-w}
\State $\x_{t+1} \leftarrow \y_t - \w_t$; \Comment{$\x_{t+1} = \mathrm{argmin}_\x\|\x-\y_t\|_{\B}\ \text{ s.t. }\S\A\x=\S\b$}
\State $\v_{t+1} \leftarrow \beta\v_t + (1-\beta)\y_t - \gamma\w_t$;
\EndFor \\
\Return $\tilde{\x} = \x_{t_{\max}}$; \Comment{Solves $\A\x=\b$.}
\end{algorithmic}
\end{algorithm}

We first present the more general convergence result for an $m\times n$ matrix $\A$, where for simplicity we assume that $m\geq n$ and $\A$ has full column rank. To fully capture the convergence properties in terms of the spectral profile of $\A$, we rely on the following notion of average condition number, $\bar\kappa_{\ell:q}$, parameterized by two indices $1\leq \ell<q\leq n$ describing the range of singular values on which the condition number is computed:
\begin{align*}
    \bar\kappa_{\ell:q} = \Big(\frac1{q-\ell}\sum_{i=\ell+1}^q\frac{\sigma_i^2}{\sigma_q^2}\Big)^{1/2},
\end{align*}
where $\sigma_1\geq\sigma_2\geq ...\geq \sigma_n>0$ are the singular values of $\A$. We use $\bar\kappa_\ell:=\bar\kappa_{\ell:n}$ as a shorthand.

\begin{thm}\label{thm:main}
    Consider an $m\times n$ matrix $\A$ with full column rank, and a vector $\b$ such that the linear system $\A\x=\b$ is consistent. Suppose that we preprocess the linear system with a randomized Hadamard transform $\Q$ by setting $\A\leftarrow\Q\A$ and $\b\leftarrow \Q\b$, and choose $\B=\I$ in Algorithm~\ref{alg:main}. Also, let the sketching matrices $\S$ be LESS-uniform with sketch size $k\geq C_1\log(m/\delta)$, $2k \leq m$ and $s\geq C_2\log^4(n/\delta)$ non-zeros per row. Then, conditioned on a $1-\delta$ probability event that only depends on the randomness of $\Q$, 
    \begin{align*}
        \E\big[\|\x_t-\x^*\|^2\big]\leq 2\Big(1- \frac{c_1 \ell}{n\bar\kappa_{\ell}\bar\kappa_{\ell:2k}}\Big)^t\|\x_0-\x^*\|^2\qquad\text{for}\qquad \ell= \Big\lceil\frac {c_2k}{\log k}\Big\rceil,
    \end{align*}
    where $\x^*$ is the minimum norm solution.
    Moreover, the algorithm can be implemented to run in time $O(mn\log m + t(nk(s+\log(n\bar\kappa_\ell)) + k^\omega))$.
\end{thm}
\begin{remark}
    While our results are stated for $\A$ with full column rank, they should extend naturally to rank-deficient $\A$ by operating on the subspace defined by the column span of~$\A$, following arguments similar to those of \cite{gower2015stochastic}. We focus here on the full-rank setting for the sake of better clarity of the notation and analysis.
\end{remark}

In the case where $\A$ is positive definite, the additional structure provides us with an improved condition number dependence in the convergence rate, by a square root factor. This is achieved by using a modified projection operator dictated by the choice of $\B$. Here, the average condition number is slightly redefined in terms of the eigenvalues of~$\A$ (instead of its singular values) to suite the positive definite setting:
\begin{align*}
    \tilde\kappa_{\ell:q} =\frac1{q-\ell}\sum_{i=\ell+1}^q\frac{\lambda_i}{\lambda_q},
\end{align*}
where $\lambda_1\geq\lambda_2\geq ...\geq \lambda_n>0$ are the eigenvalues of the $n\times n$ positive definite $\A$. Here, similarly as before, we use $\tilde\kappa_\ell$ as a shorthand for $\tilde\kappa_{\ell:n}$.
\begin{thm}\label{thm:psd}
Consider an $n\times n$ positive definite matrix $\A$ and a vector $\b\in\R^n$. Suppose that we initialize Algorithm \ref{alg:main} by setting both $\A$ and $\B$ to $\Q\A\Q^\top$ and replacing $\b$ with $\Q\b$. Also, let the sketching matrices $\S$ be LESS-uniform with sketch size $k\geq C_1\log(n/\delta)$, $2k \leq n$ and $s\geq C_2\log^4(n/\delta)$ non-zeros per row. Then, conditioned on a $1-\delta$ probability event that only depends on the randomized Hadamard transform $\Q$,
    \begin{align*}
        \E\big[\|\x_t-\x^*\|_{\A}^2\big]\leq 2\Big(1- \frac{c_1 \ell}{n\sqrt{\tilde\kappa_{\ell}\tilde\kappa_{\ell:2k}}}\Big)^t\|\x_0-\x^*\|_{\A}^2\qquad\text{for}\qquad \ell= \Big\lceil\frac {c_2k}{\log k}\Big\rceil,
    \end{align*}
    with $\x^*=\A^{-1}\b$ denoting the solution. Moreover, the algorithm can be implemented to run in time $O(n^2\log n + t(nks + k^\omega))$.
\end{thm}
We next briefly discuss how the time complexities given in Theorem \ref{t:informal-main} and \ref{t:informal-psd} can be recovered from the above convergence results, with the details relegated to Section \ref{s:completing}.

Following the assumptions of Theorem \ref{t:informal-main}, suppose that we are given an $n\times n$ invertible matrix $\A$ and an $n$-dimensional vector $\b$. If the sketch size in Algorithm \ref{alg:main} satisfies $k=\tilde O(n^{\frac1{\omega-1}})$ and $\ell=\Omega(k/\log k)$, then we can write down the time complexity described by Theorem~\ref{thm:main} as follows:
\begin{align*}
    \tilde O\Big(n^2 + \bar\kappa_{\ell}\bar\kappa_{\ell:2k}\frac n\ell (nk + k^\omega)\Big)
= \tilde O\Big(\bar\kappa_{\ell}\bar\kappa_{\ell:2k}(n^2 + nk^{\omega-1})\Big) = \tilde O\big(\bar\kappa_{\ell}\bar\kappa_{\ell:2k}n^2\big),
\end{align*}
and an analogous calculation holds for Theorem \ref{thm:psd}, with $\bar\kappa_\ell\bar\kappa_{\ell:2k}$ replaced by $\sqrt{\tilde\kappa_\ell\tilde\kappa_{\ell:2k}}$. It remains to translate those condition number quantities into our spectral tail condition number $\kappa_\ell=\frac{\sigma_\ell}{\sigma_n}$. Fortunately, in Section \ref{s:completing} we show that one can always find a sketch size $k=\tilde O(n^{\frac1{\omega-1}})$ and the corresponding $\ell=\tilde\Omega(k)$, such that those quantities are bounded by $\tilde O(\kappa_\ell)$ and $\tilde O(\sqrt{\kappa_\ell})$, respectively, which is sufficient to establish Theorems~\ref{t:informal-main}~and~\ref{t:informal-psd}. In the next sections, we present the convergence analysis of Sketch-and-Project with Nesterov's acceleration, needed for Theorems \ref{thm:main} and \ref{thm:psd}, which is our main technical contribution.

\section{First-Moment Projection Analysis via Optimal DPP Reduction}
\label{s:mu}
In this section, we show how to lower bound the smallest eigenvalue of the expected projection matrix, i.e., $\mu = \lambda_{\min}(\E[\P])$, where $\P=(\S\A)^\dagger\S\A$ is the random projection matrix central to the convergence guarantee \eqref{eq:convergence-intro} for Sketch-and-Project. In fact, we give a more general result, lower bounding the entire expected projection in positive semidefinite ordering, which will be necessary later for the analysis of the second-order term $\nu$.

\begin{lem}\label{l:mu-less}
Suppose that a rank $n$ matrix $\A\in\R^{m\times n}$ is transformed by the randomized Hadamard transform (RHT), i.e., $\A\leftarrow \Q\A$. 
If $\S$ is a $k\times m$ LESS-uniform sketching matrix with sketch size $k\in[ C_1\log (m/\delta),n] $ and $s$ non-zeros per row, then, as long as $s=O(1)$ or $s \geq C_2\log^4(n/\delta)$, and $n \geq C_3 \log(1 / \delta)$, conditioned on an RHT property that holds with probability $1-\delta$, there exists $\ell\geq \frac{ck}{\log k}$ such that the matrix $\P=(\S\A)^{\dagger}\S\A$ satisfies:
\begin{align*}
\E[\P] \succeq\A^\top\A(\A^\top\A+\lambda\I)^{-1} - \delta\I,\qquad\text{where}\qquad \lambda = \frac1\ell\sum_{i>\ell}\sigma_i^2.
\end{align*}
\end{lem}
First, we show how the above result can be translated into a lower bound on the smallest eigenvalue of the expected projection, and therefore, on $\mu$.
\begin{cor}\label{c:mu-less}
    Under the assumptions of Lemma \ref{l:mu-less}, and choosing $\delta\leq \frac1{n\bar\kappa_\ell^2}$, we have
    \begin{align*}
        \lambda_{\min}(\E[\P]) \geq \frac{\ell/4}{n\bar\kappa_\ell^2}\quad\text{for}\quad \ell= \Big\lceil\frac{ck}{\log k}\Big\rceil,\quad \bar\kappa_\ell = \Big(\frac1{n-\ell}\sum_{i>\ell}\frac{\sigma_i^2}{\sigma_n^2}\Big)^{1/2}.
    \end{align*}
\end{cor}
\begin{proof}
 Note that we can choose a sufficiently large $C_1$ in the lemma to make sure that $k$ is large enough so that $\ell=\lceil\frac {ck}{\log k}\rceil\in[4,n/2]$. Thus, we have $\lambda = \frac1\ell\sum_{i>\ell}\sigma_i^2\geq \frac 1{n-\ell}\sum_{i>\ell}\sigma_i^2\geq \sigma_n^2$. Note that if the lower bound of Lemma \ref{l:mu-less} holds with some $\ell\geq \frac{ck}{\log k}$, then it holds with $\ell = \lceil\frac{ck}{\log k}\rceil$. Thus, using Lemma \ref{l:mu-less} with $\delta\leq \frac1{n\bar\kappa_\ell^2}$, we have:
 \begin{align*}
     \lambda_{\min}(\E[\P]) 
     &\geq \lambda_{\min}(\A^\top\A(\A^\top\A+\lambda\I)^{-1}) - \delta
     = \frac{\sigma_n^2}{\sigma_n^2 + \lambda} - \delta
     \geq \frac{\sigma_n^2}{2\lambda} - \delta
     =\frac{\ell/2}{(n-\ell)\bar\kappa_l^2}-\delta\geq \frac{\ell/4}{n\bar\kappa_\ell^2},
 \end{align*}
 which concludes the proof of the corollary.
\end{proof}
\begin{remark}
    We note that for the case of $s=1$ non-zeros per row, the LESS-uniform sketch is precisely equivalent to uniform sub-sampling. In this case, Corollary \ref{c:mu-less} is a direct improvement over the result of \cite{derezinski2023solving}, from 
    $\ell=\Omega(\frac {k}{\log^3 m})$ to $\ell=\Omega(\frac {k}{\log k})$ 
     (their Lemma~4.2).   
\end{remark}

Before we proceed with the proof of Lemma \ref{l:mu-less}, we introduce some definitions related to determinantal point processes, a family of non-i.i.d.~sub-sampling distributions commonly studied in machine learning \cite{kulesza2012determinantal}, that are central to our analysis.
\begin{dfn}\label{d:dpp}
    Given an $m\times m$ positive semi-definite matrix $\B$, a determinantal point process $\setS\sim\DPP(\B)$ is a distribution over all sets $\setS\subseteq\{1,...,m\}$ such that $\Pr(\setS)\propto \det(\B_{\setS,\setS})$, where $\B_{\setS,\setS}$ denotes the prinicipal submatrix of $\B$ indexed by $\setS$.
\end{dfn}
The distribution of $\setS\sim\DPP(\B)$ is defined over all $2^m$ subsets of $\{1,...,m\}$, which means that the size of $\setS$ is also a random variable (albeit one can show that it is highly concentrated around its mean), hence we call this a random-sized DPP. If we condition the DPP on a particular set size, $|\setS|=k$, the resulting distribution is often called a (fixed-size) $k$-DPP.

The key benefit of determinantal point processes in the context of Sketch-and-Project is that, thanks to a Cauchy-Binet-type determinantal summation formula, the expectation of the random projection matrix $\P$ has a simple closed form under DPP sampling, which, as we see below, is needed for our analysis.

\begin{lem}[Lemma 5, \cite{derezinski2020improved}]\label{l:dpp-formula}
Given matrix $\A\in\R^{m\times n}$, let $\setS\sim \DPP(\frac1\lambda\A\A^\top)$ for some $\lambda>0$, and define $\S \in \R^{|\setS| \times m}$ to be the sampling matrix  corresponding to $\setS$, i.e., a matrix consisting of rows that are standard basis vectors $\mathbf{e}_i^\top$ for $i\in \setS$, so that $\S\A$ is the submatrix of the rows of $\A$ indexed by $\setS$. Then, we have
\begin{align}\label{eq:proj_lower_bound}
\E[(\S\A)^\dagger\S \A] = \A^\top\A(\A^\top\A+\lambda\I)^{-1}.
\end{align}
\end{lem}
If we could use a DPP sampling matrix instead of our LESS-uniform sketching matrix in the statement of Lemma \ref{l:mu-less}, then the desired lower bound on the expectation of $\P$ would follow directly from this result.
However, unfortunately, we cannot rely on sampling directly from a DPP for two reasons. First, despite considerable work in this direction (for an overview, see \cite{derezinski2021determinantal}), the cost of sampling from these combinatorial distributions far exceeds the time complexity of uniform sub-sampling, or of our LESS-uniform sparse sketching, which is necessary to recover our main results. Second, while DPPs provide a closed form for the first moment of the projection $\P$, the same is not true for the normalized second moment, required to bound $\nu$ in \eqref{eq:convergence-intro}. To achieve this second guarantee, we must rely on Gaussian universality properties which are specific to sketching matrices such as LESS-uniform.

Our proof of Lemma \ref{l:mu-less} starts with the observation that we do not actually need to sample from a DPP to produce a lower bound on $\E[\P]$, but rather, it suffices to produce a sample that contains a DPP sample. This is true because for any two subsets $\setS_1 \subseteq \setS_2 \subseteq [m]$ and the corresponding sampling matrices $\S_1$, $\S_2$, the projection matrices $\P_{\setS_i} := (\S_i\A)^\dagger\S_i\A$ satisfy $\P_{\setS_1} \preceq \P_{\setS_2}$. This leads to the following strategy: construct a hypothetical sampling algorithm which produces a DPP sample $\setS_{\DPP}$, but also as a byproduct, defines another sample $\setS$ such that: (a) the sample $\setS$ contains the DPP sample, i.e., $\setS_{\DPP}\subseteq \setS$, while not being too much larger; and (b) the distribution of $\setS$ is easier to sample from (for example, uniform sampling). This strategy, which can be viewed as coupling together the distributions of $\setS_{\DPP}$ and $\setS$, is then used to show that $\E[\P_{\setS}]\succeq\E[\P_{\setS_{\DPP}}]$. Finally, we rely on the DPP expectation formula from Lemma \ref{l:dpp-formula} to complete the analysis, showing that the easy sampling scheme $\setS$ enjoys the desired guarantee (without needing to implement the DPP algorithm). 

Our coupling argument is based on a combination of two algorithms. The first one, Algorithm \ref{alg:dpp}, is a classical method for sampling from a determinantal point process by reformulating it as a mixture of so-called Projection DPPs \cite{kulesza2012determinantal}, which are a special instance of a fixed-size $k$-DPP where $k$ is the rank of the given matrix.

\begin{dfn}
Given an $m\times m$ positive semi-definite matrix $\B$, we define $\setS\sim $\textnormal{P-DPP}$(\B)$ as a distribution over subsets $\setS\subseteq\{1,...,m\}$ of size $\mathrm{rank}(\B)$ such that $\Pr(\setS)\propto \det(\B_{\setS,\setS})$.
\end{dfn}\vspace{-2mm}
\begin{algorithm}
\caption{Random-sized DPP sampler (Algorithm 1 from \cite{kulesza2012determinantal})}\label{alg:dpp}
  \begin{algorithmic}[1]
    \State \textbf{input:} $\A\in\R^{m\times n}$ with SVD $\A=\U\D\V^\top$ where $\D=\diag(\sigma_1,...,\sigma_n)$
    \State \textbf{output:} $\setS_{\DPP}\sim $DPP$(\A\A^\top)$ 
    \State For each $i=1,...,n$, independently sample $b_i\sim \mathrm{Bernoulli}(\frac{\sigma_i^2}{\sigma_i^2+1})$
    \State Construct matrix $\U_{*,\b}$ consisting of columns of $\U$ indexed by $\b=[b_1,...,b_n]$
    \State \Return $\setS_{\DPP}\sim $P-DPP$(\U_{*,\b}\U_{*,\b}^\top)$ using Algorithm \ref{alg:p-dpp}
  \end{algorithmic}
\end{algorithm}

At a high level, Algorithm \ref{alg:dpp} first samples a subset of the left singular vectors of the input matrix $\A$ via independent Bernoulli coin flips, and then samples from a Projection DPP defined using just those vectors. The fact that this algorithm returns a sample from $\DPP(\A\A^\top)$, while certainly non-trivial, has been known since the work of \cite{dpp-independence}. Note that, in particular, this scheme characterizes the distribution of the size of $\setS_{\DPP}$, since this size is equal to the number of singular vectors that are selected.

Observe that Algorithm \ref{alg:dpp} technically requires an exact singular value decomposition (SVD) of the input matrix $\A$, which means that it is clearly far too expensive for our purposes. However, as explained above, we will only use it in our proof of Lemma \ref{l:mu-less}.

We combine this algorithm with a recently proposed rejection sampling scheme for sampling a Projection DPP \cite{derezinski2019minimax}, given in Algorithm \ref{alg:p-dpp}. This algorithm describes how we can select a P-DPP sample out of a stream of i.i.d.~samples drawn according to a distribution $q$ over $\{1,...,m\}$. What matters for our coupling argument is the sample complexity of this procedure, i.e., how many i.i.d.~samples from $q$ need to be drawn by the algorithm (in line~\ref{line:oracle}) to extract one P-DPP sample $\setS_{\DPP}$ via rejection sampling. This stream of i.i.d.~samples (including the rejected ones), which the algorithm collects into the set $\setS$ (line~\ref{line:setS}), is going to be the larger sample that we couple with the DPP.

\begin{algorithm}
    \caption{Projection DPP sampler}\label{alg:p-dpp}
    \begin{algorithmic}[1]
    \State \textbf{input:} $\U\in\R^{m\times \ell}$ such that $\U^\top\U=\I$, probability distribution $q$ over $\{1,...,m\}$ such that $\|\u_i\|^2\leq R\ell q_i$ for all $1\leq i\leq m$ and some $R\geq 1$, where $\u_i^\top$ is the $i$th row of $\U$
    \State \textbf{output:} $\setS_{\DPP}\sim $P-DPP$(\U\U^\top)$
    \State Initialize $\Z = \I_\ell$, $\setS_{\DPP} = \{\}$, $\setS=\{\}$
    \State \textbf{for }$j=1,...,\ell$ 
    \State \quad\textbf{repeat}
    \State  \quad\quad Sample index $I \sim q$ \label{line:oracle}
    \State \quad\quad $\setS\leftarrow \setS\cup\{I\}$
    \label{line:setS}
    \State \quad\quad Sample \ $\mathrm{Accept}\sim
    \text{Bernoulli}\big(\frac{\u_{I}^\top\Z\u_{I}}{R\ell q_{I}}\big)$
    \State \quad\textbf{until} $\mathrm{Accept}=1$
    \State \quad $\setS_{\DPP}\leftarrow \setS_{\DPP}\cup\{I\}$
    \State \quad$\Z \leftarrow \Z -
    \frac{\Z\u_{I}\u_{I}^\top\Z}{\u_{I}^\top\Z\u_{I}}$
    \State \textbf{end for}
    \State \Return $\setS_{\DPP}$
    \end{algorithmic}
\end{algorithm}

To ensure that the rejection sampling used in Algorithm \ref{alg:p-dpp} terminates efficiently, we rely on the following prior result, which bounds the number of rejection sampling steps required across all iterations of the sampler.
\begin{lem}[\cite{derezinski2019minimax}]\label{l:p-dpp}
Given $\U\in\R^{m\times \ell}$ such that $\U^\top\U=\I$, let $\u_i^\top$ denote its $i$th~row. Consider a probability distribution $q$ over $\{1,...,m\}$ such that $\|\u_i\|^2\leq R\ell q_i$ for all $i$ and some $R\geq 1$. Then, Algorithm \ref{alg:p-dpp} produces a sample $\setS_{\DPP}\sim \PDPP(\U\U^\top)$, and with probability $1-\delta$ it requires at most $|\setS|=O(R\ell\log(\ell/\delta))$ samples drawn from $q$ (line \ref{line:oracle}). 
\end{lem}

The number of samples required by Algorithm \ref{alg:p-dpp} depends on $R$, which specifies how well $q$ approximates the so-called leverage score sampling distribution of $\U$. For a matrix $\U$ such that $\U^\top\U=\I$, its $i$th leverage score is defined as the norm squared of the $i$th row $\u_i$ of $\U$. In particular, if $R=1$ in Lemma \ref{l:p-dpp}, then $q_i$ corresponds to exact leverage score sampling, $q_i=\|\u_i\|^2/\sum_i\|\u_i\|^2 = \|\u_i\|^2/\ell$. This distribution effectively describes the likelihood of an index $i$ being sampled into $\setS_{\DPP}$, via the following formula: $\Pr(i\in\setS_{\DPP})=\|\u_i\|^2$.

Naturally, we cannot simply use $R=1$ in Algorithm \ref{alg:p-dpp}, because that would force $q$ to be the exact leverage score sampling distribution, which is itself as expensive to compute as the SVD.
Instead, in our algorithm, $R$ is controlled by applying the RHT to the data matrix and relying on Lemma \ref{l:rht} to ensure that the row norms (and thereby the leverage scores) are uniformly bounded. This way, even a uniform sampling distribution $q$ provides a decent enough approximation of leverage score sampling so that $R=\tilde O(1)$.

To further refine the coupling argument, and adapt it to LESS-uniform sketching, our proof of Lemma \ref{l:mu-less} also relies on the notion of total variation distance between two probability distributions, which we define here.

\begin{dfn}[Total variation distance]
Let $\mu, \nu$ be probability measures defined on a measurable space $(\Omega, \mathcal{F})$, the total variation distance between $\mu$ and $\nu$ is defined as
\begin{align*}
d_{\mathrm{tv}}(\mu, \nu) = \sup_{A\in\mathcal{F}} |\mu(A) - \nu(A)|.
\end{align*}
\end{dfn}

Bounding the total variation distance allows us to couple two random variables so that they are equal with high probability.

\begin{lem}[Example 4.14, \cite{v14}]
\label{lem:tv_distance}
Let $\mu, \nu$ be probability measures defined on some $(\Omega,\mathcal F)$. Then, we have $\inf_{(X,Y)}\Pr\{X\neq Y\} = d_{\mathrm{tv}}(\mu, \nu)$, where $X\sim\mu$, $Y\sim\nu$, and
the infimum is taken over all couplings of $\mu$ and $\nu$.
\end{lem}

We are now ready to give the proof of Lemma \ref{l:mu-less} by coupling a LESS-uniform sketch with a DPP sample produced by Algorithm \ref{alg:dpp} for a specially expanded version of matrix~$\A$.
\begin{proof}[Proof of Lemma \ref{l:mu-less}]
 As mentioned above, the general approach of our proof is to show that the LESS-uniform sketching matrix $\S$ defined in Lemma \ref{l:mu-less} essentially \emph{contains} a DPP sample, so that we can rely on the simple closed form expression for the expected projection under DPP sampling (Lemma \ref{l:dpp-formula}). To do this, we must first represent a LESS-unifom sketching matrix itself as a sample. We achieve this by defining an expanded matrix $\bar\A$, which contains all possible rows that could be produced via this sketching method.
  
  Define the LESS-uniform \emph{expansion}
  matrix $\bar\S\in\R^{N\times m}$, where $N=(2m)^s$ and $s$ is the number of non-zeros per row in the definition of LESS-uniform (Definition \ref{d:less}), so that $\bar\S$ consists of all
  rows of the form:
  \begin{align*}
\sqrt{\frac m{sN}}\sum_{i=1}^sr_{i}\e_{I_i}^\top, \quad
    r_1,...,r_s\in\{1,-1\},\quad I_1,...,I_s\in[m].
  \end{align*}
  Note that we have $\bar\S^\top\bar\S=\I$, because the LESS-uniform sketching distribution is isotropic. 
Now, let $\U\D\V^\top$ be the SVD of $\A$ and consider matrix $\bar\A
= \bar\S\Q\A$, where recall that $\Q$ is the RHT matrix and it satisfies $\Q^\top\Q=\I$. We can now write the SVD of $\bar\A$ as follows: $\bar\A=(\bar\S\tilde\U)\D\V^\top$ where $\tilde\U = \Q\U$. Observe that we designed matrix $\bar\A$ so that an i.i.d.~uniform sample of rows from $\bar\A$ is distributed identically to a LESS-uniform sketch of matrix $\Q\A$.

Next, we consider running Algorithm \ref{alg:dpp} on the expanded matrix to produce a sample $\setS_{\DPP}\sim\DPP(\frac1\lambda\bar\A\bar\A^\top)$. The size of the sample is random, and it is equal to $\sum_ib_i$, since Algorithm \ref{alg:p-dpp} always selects one element per each $b_i=1$. Thus, we can obtain the expected size of the DPP sample as follows, where $\sigma_i$'s denote the singular values of $\A$:
\begin{align*}
    \E\big[|\setS_{\DPP}|\big] = \sum_{i=1}^n\E[b_i] = \sum_{i=1}^n\frac{\sigma_i^2/\lambda}{\sigma_i^2/\lambda+1} =\sum_{i=1}^n\frac{\sigma_i^2}{\sigma_i^2+\lambda}
\end{align*}
Recall that, due to our choice of $\lambda = \frac1\ell\sum_{i>\ell}\sigma_i^2$, we have:
\begin{align*}
\sum_{i=1}^n\frac{\sigma_i^2}{\sigma_i^2+\lambda}\leq \ell + \frac1\lambda \sum_{i>\ell}\sigma_i^2 = 2\ell.
\end{align*}
From now on, we will define $\tilde\ell:=|\setS_{\DPP}|=\sum_ib_i$ as a shorthand for the random DPP sample size.   
Ultimately, our goal is to show that the sampler with high probability produces its output sample $\setS_{\DPP}$ after drawing only at most $O(\ell\log \ell)$ uniform samples in line~\ref{line:oracle} of Algorithm \ref{alg:p-dpp}, i.e., to bound the size of the set $\setS$ produced in the course of the algorithm. The collection of rows of $\bar\A$ indexed by $\setS$ forms a LESS-uniform sketch $\S\Q\A$, which allows us to argue that a $k=O(\ell\log \ell)$ size LESS-uniform sketch contains a  $\DPP(\frac1\lambda\bar\A\bar\A^\top)$ sample.

We next define the high probability event depending only on the RHT, which will be used to ensure that the parameter $R$ in Lemma \ref{l:p-dpp} and Algorithm \ref{alg:p-dpp} can be bounded. First, consider the random vector $\b$ defined by Algorithm \ref{alg:dpp} when applied to matrix $\frac{1}{\sqrt{{\gamma}}}\bar\A$, and let $\Ec$ be the event that the matrix $\Q\U_{*,\b}$ achieves nearly-uniform marginals; i.e., $\|\e_i^\top\Q\U_{*,\b}\|^2 \leq C(\tilde \ell+\log(m/\delta))/m$ for all $i$, where recall that $\U_{*,\b}$ is an $m\times \tilde\ell$ matrix. For any fixed $\b$, this event holds with probability $1-\delta^2$ due to Lemma~\ref{l:rht}. In other words, we know that $\Pr(\neg\Ec\mid\b)\leq \delta^2$. From this, it follows that:
  \begin{align*}
   \E\big[\Pr(\neg\Ec\mid\Q)\big] = \Pr(\neg\Ec) = \E\big[\Pr(\neg\Ec\mid\b)\big] \leq \delta^2.
  \end{align*}
  Thus, letting $\delta_{\Q}=\Pr(\neg\Ec\mid\Q)$, via Markov's inequality we have that:
    \begin{align*}
    \Pr\big(\delta_\Q \geq \delta\big)\leq \frac{\E[\delta_\Q]}{\delta} = \frac{\delta^2}{\delta} = \delta.
  \end{align*}
  So, for a random $\Q$,
  there is an event $\mathcal{A}$ such that $\Pr(\mathcal{A})\geq 1-\delta$, and
  conditioned on $\mathcal{A}$, we achieve nearly-uniform row norms for a random $\b$
  with probability $1-\delta$. From now on, we will condition on event $\mathcal A$. 
  Next, note that a standard Chernoff bound on the sum of Bernoulli random variables implies that 
   $\tilde \ell=\sum_ib_i \leq C(\ell + \sqrt{\ell \log(1/\delta)})$ 
  with probability $1-\delta$. Thus, for the rest of the analysis, we can assume that $\tilde \ell \leq C\ell$. 
  
  Suppose that for the randomly sampled $\b$ the matrix $\tilde\U_{*,\b}=\Q\U_{*,\b}$ has $\tilde \ell\leq C\ell$ columns and nearly-uniform marginals, i.e., $\|\tilde\u_i\|^2 \leq C(\ell+\log(m/\delta))/m$, where $\tilde\u_i$ is the transposed $i$th row of $\tilde\U_{*,\b}$. 
  Next, we show that even after expanding to matrix $\bar\A$, the corresponding matrix of left singular vectors $\bar\S\tilde\U_{*,\b}$ also has nearly-uniform marginals with high probability. Note that in the case of $s=1$ non-zeros per row, which corresponds to simply sub-sampling from the matrix $\Q\A$, the expansion is trivial, and so this property is immediately implied by $\tilde\U_{*,\b}$ having nearly-uniform rows. However, to work well with the second-moment projection analysis, we need to show the result when $s \geq C\log^4(n/\delta)$. In this case, we can only show the following somewhat weaker guarantee over all rows of $\bar\S\tilde\U_{*,\b}$. Let $\frac1{\sqrt N}\bar\ss$ denote a row of $\bar\S$, with $\bar\ss =\sqrt{\frac m{s}}\sum_{i=1}^sr_{i}\e_{I_i}^\top$. Again using $\tilde\u_i$ to denote the rows of $\tilde\U_{*,\b}$, we have:
  \begin{align*}
  \|\bar\ss^\top\tilde\U_{*,\b}\| \leq
  \sqrt{\frac{m}{s}}\sum_{i=1}^s\|\tilde\u_{I_i}\| \leq
  \sqrt{\frac{m}{s}}\  sC\sqrt{\frac{\ell+\log(m/\delta)}{m}}\leq C\sqrt{s(\ell+\log(m/\delta))}.
\end{align*}
This leads to an upper bound of $C\frac{s\ell}{N}$ 
for the leverage scores of $\bar\S\tilde\U_{*,\b}$, which is sub-optimal by a factor of $s$. However, fortunately we can show a sharper guarantee for a vast majority of the leverage scores. Let $\frac1{\sqrt N}\bar\ss$ now be a \emph{randomly} sampled row of $\bar\S$. Using the Gaussian universality result from Lemma \ref{l:universality} (specifically, the upper-bound on the largest singular value with the sketching matrix of size $1 \times m$)
for $\bar\ss^\top\tilde\U_{*,\b}$, assuming
$\bar\ss$ has at least $C\log^4(n/\delta)$ 
 non-zeros, we have:
\begin{align}
  \Pr\big(\|\bar\ss^\top\tilde\U_{*,\b}\|\geq C\sqrt \ell\big) \leq \delta,\label{eq:prsbar} 
\end{align}
where the probability is taken with respect to the randomness of $\bar\ss$ only. We can use the above guarantees to design a probability distribution $q$ over the row indices $\{1,...,N\}$ of $\bar\S\tilde\U_{*,\b}$, such that it is nearly proportional to the row norms squared, as required by Algorithm~\ref{alg:p-dpp}.

Consider partitioning the row indices of $\bar\A$ into the set $\setR$ of indices
  for which we can establish a sharp bound
  $\|\bar\ss_i^\top\tilde\U_{*,\b}\| \leq C\sqrt \ell$, 
  and its complement
  $\bar \setR$,
  for which we merely have 
  $\|\bar\ss_i^\top\tilde\U_{*,\b}\|\leq C\sqrt{s(\ell+\log(m/\delta))}$, where $\bar\ss_i$ denotes the transposed $i$th row of $\bar\S$. 
Formally, we define:
  \begin{align*}
    \setR = \Big\{i\in\{1,...,N\} \ :\ \|\bar\ss_i^\top\tilde\U_{*,\b}\|^2\leq C\ell\Big\}.
  \end{align*}
  Now, we construct the probability distribution $q$ for Algorithm \ref{alg:p-dpp} as follows:
\begin{align*}
    q_i =
    \begin{cases}
      (1-\epsilon)\frac1N &\text{if }i\in \setR,\\
      \frac1{\ell N}\|\bar\ss_i^\top\tilde\U_{*,\b}\|^2&\text{otherwise,}
    \end{cases}                                            
  \end{align*}
  where $\epsilon$ is chosen so that $\sum_iq_i=1$. Using \eqref{eq:prsbar}, we know that
  $\bar \setR$ includes no more than a $\delta$ fraction of all
  indices, i.e., $|\bar \setR|/N\leq\delta$, so we can solve for $\epsilon$:
  \begin{align*}
1=\sum_i q_i = (1-\epsilon)\sum_{i\in R}\frac1N + \sum_{i\not\in
    R}\frac1{\ell N}\|\bar\ss_i^\top\tilde\U_{*,\b}\|^2
    \leq (1-\epsilon) + Cs\delta,
  \end{align*}
  so we can use $\epsilon \leq Cs\delta$. 
  It is easy to see that as long as $\ell\geq \log(m/\delta)$, distribution $q$ is a constant factor approximation of leverage score sampling for $\bar\S\tilde\U_{*,\b}$, so that it can be used in Algorithm \ref{alg:p-dpp} with $R=O(1)$ in Lemma \ref{l:p-dpp}. However, since we cannot actually compute $q$, we still have to show that $q$ can be coupled with the exact uniform sampling distribution $q' = (\frac1N,....,\frac1N)$ which represents LESS-uniform. We do this by bounding the total variation distance $d_{\textrm{tv}}(q,q')$ between $q$ and $q'$, as follows:
  \begin{align*}
    d_{\mathrm{tv}}(q,q')
    &= \frac12\sum_i|q(i) - q'(i)|
    \\
    &\leq \sum_{i\in \setR}|q(i) - q'(i)| + \sum_{i\in\bar \setR}|q(i) -
      q'(i)|
    \\
    &\leq \sum_{i\in \setR}\frac{\epsilon}{N}
      + \sum_{i\in \bar \setR}\Big(\frac{1}{N} + \frac {Cs\ell}{lN}\Big)
      \\
     &\leq \epsilon + Cs\delta \leq 2Cs\delta.
\end{align*}
So, by Lemma \ref{lem:tv_distance}, we can couple approximate leverage score sampling $q$ with
  uniform sampling $q'$
without observing any difference for $O(\ell\log \ell)$ steps of Algorithm \ref{alg:p-dpp}, with probability $1- O(s\delta \ell\log \ell)$. This is sufficient since we know from Lemma \ref{l:p-dpp} that the algorithm will terminate after $O(\ell\log \ell)$ steps. Note that adjusting the constants in the statement of the result we can replace the failure probability $O(s\delta \ell\log \ell)$ with $\delta$. 

To summarize, we have shown that, conditioned on the event $\mathcal A$ that only depends on $\Q$, running Algorithms \ref{alg:dpp} and \ref{alg:p-dpp} as described above on input $\frac1{\sqrt\lambda}\bar\A=\frac1{\sqrt\lambda}\bar\S\Q\A$ will produce a sample from $\setS_{\DPP}\sim\DPP(\frac1\lambda\bar\A\bar\A^\top)$ by drawing indices at random from $\{1,...,N\}$ according to distribution $q$. Let $\setS$ denote the set of these indices, including the ones rejected in Algorithm~\ref{alg:p-dpp}. We showed that with probability $1-\delta$, this set has size at most $|\setS| = O(\ell\log (\ell/\delta))$, and that the distribution of these indices can be coupled with a uniform sample of indices so that they are indistinguishable with probability $1-\delta$. 

Rephrasing this, we can say that there is a coupling between the run of our DPP sampler (producing $\setS_{\DPP}$) and a uniform sample 
$\setS$ of size $k \leq C\ell\log \ell$ 
such that, conditioned on an event $\mathcal {B}$ that holds with probability $1-\delta$, we have $\setS_{\DPP}\subseteq \setS$. Now, let $\P_{\setS_{\DPP}}$ and $\P_\setS$ denote the corresponding projections onto the span of rows selected by $\setS_{\DPP}$ and $\setS$ respectively. Then, we have:
\begin{align*}
    \E[\P_\setS] &\succeq \E[\P_\setS\one_{\mathcal B}] 
    \\
    &\overset{(1)}{\succeq} \E[\P_{\setS_{\DPP}}\one_{\mathcal B}]
    \\
    &\overset{(2)}{\succeq} \E[\P_{\setS_{\DPP}}] - \delta\I
    \\
    &\overset{(3)}{=} \A^\top\A(\A^\top\A+\lambda\I)^{-1} - \delta\I,
\end{align*}
where $(1)$ follows because $\setS_{\DPP}\subseteq \setS$ when conditioned on $\mathcal B$, $(2)$ follows because $\mathcal B$ holds with $1-\delta$ probability, and $(3)$ follows from Lemma~\ref{l:dpp-formula}. Finally, note that while in the proof we started with $\ell$ and then selected $k \leq C\ell \log \ell$,
this can easily be restated as starting with $k\geq C\log m$ 
and selecting 
$\ell\geq \frac {Ck}{\log k}$, since $\frac k{\log k}\log(\frac k{\log k}) \leq Ck$.
\end{proof}

\section{Second-Moment Projection Analysis via Gaussian Universality}
\label{s:nu}
We now turn to bounding the second-order term $\nu = \lambda_{\max}(\E[(\bar\P^{-1/2}\P\bar\P^{-1/2})^2])$, which appears alongside the first-order term $\mu$ in the convergence bound \eqref{eq:convergence-intro} of Sketch-and-Project with Nesterov's acceleration (Algorithm \ref{alg:main}). We do this by combining the previous bound on the expected projection with a bound on the smallest singular value of the LESS-uniform sketch $\S\A$. Our main result is a more general statement on a positive semidefinite ordering of a certain second matrix moment of $\P$.

\begin{lem}\label{l:nu-gaussian}
Suppose that a rank $n$ matrix $\A\in\R^{m\times n}$ with condition number $\kappa$ is transformed by the randomized Hadamard transform (RHT), i.e., $\A\leftarrow \Q\A$. If $\S$ is a $k\times m$ LESS-uniform sketching matrix with sketch size $k\geq C_1\log(m/\delta)$, $2k \leq m$ and $s\geq C_2\log^4(n/\delta)$ non-zeros per row, $n \geq C_3 \log(1 / \delta)$, and $\delta \leq C_4 / n \kappa^2$, then, conditioned on an RHT property that holds with probability $1-\delta$, matrix $\P=(\S\A)^{\dagger}\S\A$ satisfies:
\begin{align*}
    \E\big[\P\E[\P]^{-1}\P\big] \preceq \frac { c_1 n}{\ell} \bar\kappa_{\ell:2k}^2 \cdot\E[\P] + 2\delta\I 
    \quad\text{for}\quad \ell= \Big\lceil\frac{c_2k}{\log k}\Big\rceil,
\end{align*}
where $\bar\kappa_{\ell:q}:= (\frac1{q-\ell}\sum_{i=\ell+1}^q\frac{\sigma_i^2}{\sigma_q^2})^{1/2}$ with $\sigma_1\geq \sigma_2\geq ...$ denoting the singular values of $\A$. 
\end{lem}
\noindent
From this main result, we can conclude the following bound on $\nu$.
\begin{cor}
    \label{c:nu-bound}
    Under the assumptions of Lemma~\ref{l:nu-gaussian},
\begin{align}
    \nu \leq \frac{c_1 n}{\ell} \bar\kappa_{\ell:2k}^2
    \quad\text{for}\quad
    \ell \geq \frac{c_2 k}{\log k}.
\end{align}
\end{cor}

\noindent
We will prove  Lemma~\ref{l:nu-gaussian} and Corollary~\ref{c:nu-bound} together, using the following lemma which allows us to control the smallest singular value of the LESS-uniform sketch with high probability. 

\begin{lem}\label{l:smin-less}
Suppose that we are given matrix $\A\in\R^{m\times n}$ with full column rank, and let random matrix $\Q\in\R^{m\times m}$ be the randomized Hadamard transform.
If the sketching matrix $\S$ is a $k\times m$ LESS-uniform matrix with $C_1 \log(m/\delta) \leq 2k \leq m$ and $s \geq C_2 \log^4(n/\delta)$ non-zeros per row, then
\begin{align*}
    \sigma_{\min}^2(\S\Q\A) \geq c\Big(2k\sigma_{2k}^2 + \sum_{i>2k}\sigma_i^2\Big)\quad\text{with probability at least $1-\delta$.}
\end{align*}
\end{lem}
\begin{remark}
    Our argument is quite general, and in particular, can be easily adapted to sketching methods where $\S\Q$ is replaced by a Gaussian or a sub-Gaussian matrix. To our knowledge, this is the first such characterization for the smallest singular value of a sketch.
\end{remark}
\begin{proof}
First, we note that $\sigma_{\min}^2(\S\Q\A)=\lambda_{\min}(\S\Q\A\A^\top\Q^\top\S^\top)$. Let $\A=\U\D\V^\top$ be the singular value decomposition (SVD) of $\A$, and let the corresponding SVD of $\Q\A$ be $\tilde\U\D\V^\top$ where $\tilde\U=\Q\U$ and we used the fact that $\Q$ is an orthogonal matrix. The central part of our argument is to use the SVD of $\A$ to decompose the matrix $\S\Q\A\A^\top\Q^\top\S^\top$ into a sum of isotropic random matrices via a careful telescoping argument. This allows us to apply Gaussian universality of the smallest singular value (Lemma~\ref{l:universality}) to each component of the sum. To that end, letting $\tilde\U_{*,t}$ denote the $t$-th column of $\tilde\U$, and $\tilde\U_{*,1:t}$ be the $m\times t$ matrix consisting of the first $t$ columns of $\tilde\U$, we start by lower bounding $\S\Q\A\A^\top\Q^\top\S^\top$ with the following telescoping sum expression:
  \begin{align}
    \S\Q\A\A^\top\Q^\top\S^\top
    &= \S\tilde\U\D\V^\top\V\D\tilde\U^\top\S^\top=\S\tilde\U\D^2\tilde\U^\top\S^\top\nonumber
    \\    &=\sum_{t=1}^n\sigma_t^2\S\tilde\U_{*,t}\tilde\U_{*,t}^\top\S^\top\nonumber
    \\
    &\succeq \sum_{t\geq
      2k}^{n-1}(\sigma_t^2-\sigma_{t+1}^2)\S\tilde\U_{*,1:t}\tilde\U_{*,1:t}^\top\S^\top 
      + \sigma_n^2 \S \tilde\U\tilde\U^\top \S^\top\nonumber
    \\
    &\succeq
      \Big( \sum_{t\geq
      2k}^{n-1}(\sigma_t^2-\sigma_{t+1}^2)\sigma_{\min}^2(\S\tilde\U_{*,1:t})
      + \sigma_n^2 \sigma_{\min}^2(\S\tilde\U)\Big) \cdot\I,\label{eq:smin-less-telescoping}
  \end{align}
where, at a high level, the idea is to combine the contributions of the rank one matrices $\tilde\U_{*,t}\tilde\U_{*,t}^\top$ (corresponding to the $t$-th singular vector) into groups of $2k$ or more, so that all of the matrices $\S\tilde\U_{*,1:t}$ appearing in \eqref{eq:smin-less-telescoping} are sufficiently wide for us to be able to show a lower bound on the smallest singular values of $\sigma_{\min}^2(\S\tilde\U_{*,1:t})$ via Lemma \ref{l:universality}.

\begin{figure}
\begin{tikzpicture}[scale=1]
\tikzset{
    dot/.style={circle, draw, fill=black, inner sep=0pt, minimum
    width=3pt},
    ldot/.style={circle, draw, black, inner sep=0pt, minimum
    width=2pt}, 
    top/.style={anchor=south, inner sep=1pt},
    bottom/.style={anchor=north, inner sep=3pt},
  }

  \draw[fill=lightgrey,lightgrey] (0,1) rectangle (4,1.7);
  \draw[darkgrey] (0,1.7) -- (4,1.7);
  \draw[fill=lightgrey,lightgrey] (0,1) rectangle (5,1.3);
  \draw[darkgrey] (0,1.3) -- (5,1.3);  
  \draw[fill=lightgrey,lightgrey] (0,.7) rectangle (6,1);
  \draw[darkgrey] (0,1) -- (6,1);
  \draw[fill=lightgrey,lightgrey] (0,.5) rectangle (7,.7);
  \draw[darkgrey] (0,.7) -- (7,.7);
  \draw[fill=lightgrey,lightgrey] (0,.3) rectangle (8,.5);
  \draw[darkgrey] (0,.5) -- (8,.5);
  \draw[fill=lightgrey,lightgrey] (0,.15) rectangle (9,.3);
  \draw[darkgrey] (0,.3) -- (9,.3);
  \draw[fill=lightgrey,lightgrey] (0,0) rectangle (10,.15);
  \draw[darkgrey] (0,.15) -- (10,.15);

\node[] (Z0) at (0,0) {};
\node[] (Z) at (1,3.25) {};
\node[top] at (.5,3.25) {\mbox{\scriptsize $\sigma_1^2$}};
\draw [thick] (Z0) rectangle (Z);

\node[] (A0) at (1,0) {};
\node[] (A) at (2,2.75) {};
\draw [thick] (A0) rectangle (A);

\node[] (B0) at (2,0) {};
\node[] (B) at (3,2.25) {};
\draw [thick] (B0) rectangle (B);

\node[] (C0) at (3,0) {};
\node[] (C) at (4,1.7) {};        
\draw [thick] (C0) rectangle (C);
  \node[top] at (3.5,1.7) {\mbox{\scriptsize $~\sigma_{2k}^2$}};
  
\node[] (D0) at (4,0) {};
\node[] (D) at (5,1.3) {};
\draw [thick] (D0) rectangle (D);
\node[top] at (4.5,1.3) {\mbox{\scriptsize $~~\sigma_{2k+1}^2$}};

\node[] (E0) at (5,0) {};
\node[] (E) at (6,1) {}; 
\draw [thick] (E0) rectangle (E);

\node[] (F0) at (6,0) {};
\node[] (F) at (7,0.7) {}; 
\draw [thick] (F0) rectangle (F);

\node[] (G0) at (7,0) {};
\node[] (G) at (8,0.5) {};
\draw [thick] (G0) rectangle (G);

\node[] (H0) at (8,0) {};
\node[] (H) at (9,0.3) {};
\draw [thick] (H0) rectangle (H);

\node[] (I0) at (9,0) {};
\node[] (I) at (10,0.15) {};
\node[top] at (9.5,0.15) {\mbox{\scriptsize $\sigma_n^2$}};
\draw [thick] (I0) rectangle (I);

\draw [decorate, thick,
decoration = {calligraphic brace, mirror, amplitude=1.75pt,raise=1.25pt}] (-.05,1) -- (-.05,.7) node[midway,xshift=-3em]{\scriptsize $t\times (\sigma_t^2\!-\!\sigma_{t+1}^2)$};

\draw [decorate, thick,
decoration = {calligraphic brace, mirror, amplitude=3pt,raise=1pt}] (0.1,-0.1) -- (3.9,-0.1) node[midway,yshift=-0.9em]{\scriptsize $2k\times \sigma_{2k}^2$};
\draw [decorate, thick,
decoration = {calligraphic brace, mirror, amplitude=3pt,raise=1pt}]
(4.1,-0.1) -- (9.9,-0.1) node[midway,yshift=-0.9em]{\scriptsize $\sum_{t>2k}\sigma_t^2$};
\end{tikzpicture}
\caption{Illustration for the telescoping sum argument in the proof of Lemma \ref{l:smin-less}. The expression in \eqref{eq:telescoping-before} can be viewed as computing the shaded area by summing over the row rectangles, whereas the expression in \eqref{eq:telescoping-after} is computing the same shaded are by summing over the columns.}
\label{f:telescoping}
\end{figure}

To that end, note that each matrix $\tilde\U_{*,1:t}$ has orthonormal columns, and since it can be written as $\tilde\U_{*,1:t}=\Q\U_{*,1:t}$, we can use the row uniformizing property of the randomized Hadamard transform $\Q$. Namely, for each value of $t \geq 2k$, we can apply Lemma~\ref{l:rht} to obtain that with probability $1 - \delta/2n$, the row norms of $\tilde\U_{*,1:t}$ are uniformly bounded as
\begin{align*}
  \|\tilde\U_{i,1:t}\| \leq \sqrt{t/m} + \sqrt{8 \log(2mn / \delta) / m}
  \leq C \sqrt{t/m}.
\end{align*}
This means that each matrix $\tilde\U_{*,1:t}$ fits the assumptions of our Gaussian universality result (Lemma~\ref{l:universality}), and we can use it to lower bound the smallest singular value. Specifically, the lemma implies that for any $t\geq 2k$, with probability $1 - \delta / 2n$ the matrix $\S\tilde\U_{*,1:t}\in\R^{k\times t}$ satisfies
\begin{align*}
    \sigma_{\min}(\S\tilde\U_{*,1:t}) \geq c_1 \sqrt{t}.
\end{align*}
Taking a union bound, we can ensure that this singular value lower bound holds for every $t$ with probability $1 - \delta$. Plugging these inequalities into the telescoping sum in \eqref{eq:smin-less-telescoping}, we obtain:
\begin{align}
    \sigma_{\min}^2(\S\Q\A)
    \geq c_1^2\bigg(\sum_{t\geq 2k}^{n-1}t(\sigma_t^2-\sigma_{t+1}^2) + n\sigma_n^2\bigg).\label{eq:telescoping-before}
\end{align}

To reach the final claimed lower bound, it remains to undo the telescoping sum expression. This can be explained most simply via the diagram in Figure \ref{f:telescoping}. Here, the expression in \eqref{eq:telescoping-before} can be interpreted as computing the shaded area in the figure by summing up over the row rectangles of dimension $t \times (\sigma_t^2-\sigma_{t+1}^2)$. Instead, we can compute it by summing over the columns, i.e., the total contributions of each singular value. This breaks the expression down into the term that corresponds to the $2k\times \sigma_{2k}^2$ rectangle, plus the sum over the tail of the remaining singular values. Thus, we obtain:
\begin{align}
\sigma_{\min}^2(\S \Q\A) \geq c_1^2 \Big(2k \sigma_{2k}^2 + \sum_{t = 2k + 1}^n \sigma_t^2 \Big).\label{eq:telescoping-after}
\end{align}
This concludes the proof.
\end{proof}

We are now ready to prove Lemma \ref{l:nu-gaussian}; i.e., provide an upper-bound for the second-order projection expression $\E[\P\E[\P]^{-1}\P]$, along with Corollary~\ref{c:nu-bound}. Recall that the random projection matrix is defined as $\P=(\S\A)^{\dagger}\S\A$, where to simplify the notation we assume that the matrix $\A$ has already been preprocessed with the RHT matrix $\Q$.
\begin{proof}[Proof of Lemma \ref{l:nu-gaussian} and Corollary~\ref{c:nu-bound}]
We start by bounding the term $\E[\P]^{-1}$ inside the expression. For this, we can rely on the first-order analysis from Lemma \ref{l:mu-less}, which states that:
\begin{align*}
    \E[\P]\succeq \A^\top\A(\A^\top\A+\lambda\I)^{-1}-\delta\I,
\end{align*}
where $\lambda =\frac1\ell\sum_{i>\ell}\sigma_i^2$ and $\ell = \lceil\frac {ck}{\log k}\rceil$. Now, since by the assumption in the lemma we have $\delta \leq 1 / 2 n \kappa^2 \leq   \lambda_{\min}(\A^\top\A(\A^\top\A+\lambda\I)^{-1})/2$, 
we have 
\begin{align*}
\E[\P]^{-1}\preceq 2(\A^\top\A)^{-1}(\A^\top\A+\lambda\I) = 2\I + 2\lambda(\A^\top\A)^{-1}.
\end{align*}
Applying this bound to our second-order expression, we obtain:
\begin{align*}
\E\big[\P\E[\P]^{-1}\P\big]
&\preceq 2\E\Big[\P\big(\I + \lambda(\A^\top\A)^{-1}\big)\P\Big]
\\
&= 2\E[\P] + 2\lambda \E[\P(\A^\top\A)^{-1}\P].
\end{align*}
Thus, it remains to bound $\E[\P(\A^\top\A)^{-1}\P]$.
First, we will use Lemma \ref{l:smin-less} to define a high-probability event  that bounds the spectral norm of the matrix $\P(\A^\top\A)^{-1}\P$. To do that, define $\U=\A(\A^\top\A)^{-1/2}$, which is a matrix with orthonormal columns (recall that we assumed full column rank of $\A$). By expanding the Moore-Penrose pseudoinverse we can write $\P=(\S\A)^\dagger\S\A=\A^\top\S^\top(\S\A\A^\top\S^\top)^{-1}\S\A$, which gives us:
 \begin{align*}
     \|\P(\A^\top\A)^{-1}\P\| 
    &= \|(\A^\top\A)^{-1/2}\P(\A^\top\A)^{-1/2}\| \\
    &= \|\U^\top\S^\top(\S\A\A^\top\S^\top)^{-1}\S\U\| \\
    &\leq \|\S\U\|^2\|(\S\A\A^\top\S^\top)^{-1}\| 
    = \frac{\sigma_{\max}^2(\S\U)}{\sigma_{\min}^2(\S\A)}.
 \end{align*}
 This allows us to leverage Lemma \ref{l:smin-less}, ensuring that $\sigma_{\min}^2(\S\A)$ is lower bounded by $c(2k\sigma_{2k}^2+\sum_{i>2k}\sigma_i^2)$ with probability $1-\delta/2$. Moreover, since $\U$ is a matrix with orthonormal columns that was preprocessed by an RHT, we can use Lemma \ref{l:universality} 
to control $\sigma_{\max}^2(\S\U)$, obtaining that with probability $1-\delta/2$:
\begin{align}
    \sigma_{\max}(\S\U) \leq c_2 \sqrt n.
\end{align}
Putting these bounds together, we can define the following event which, for an appropriate constant $C>0$, holds with probability $1-\delta$:
\begin{align*}
    \Ec := \Big\{\|\P(\A^\top\A)^{-1}\P\| \leq \frac{Cn}{2k\sigma_{2k}^2+\sum_{i>2k}\sigma_i^2}\Big\}.
\end{align*}
Unfortunately, the above bound on the spectral norm of $\P(\A^\top\A)^{-1}\P$ is not sharp enough by itself when used to bound the expectation of this matrix.
However, we are able get a sharper bound on $\E[\P(\A^\top\A)^{-1}\P]$ in the positive semidefinite ordering. First, using law of total expectation,  
\begin{align*}
    \E[\P(\A^\top\A)^{-1}\P] 
    &= \E[\P(\A^\top\A)^{-1}\P\mid\Ec] \Pr(\Ec) + \E[\P(\A^\top\A)^{-1}\P\mid\neg\Ec] \Pr(\neg\Ec)
    \\
    &\preceq \E[\P(\A^\top\A)^{-1}\P\cdot 1_{\Ec}] + \delta \sigma_{\min}^{-2}\E[\P\mid\neg\Ec],
\end{align*}
where $1_{\Ec}$ is the characteristic variable of $\Ec$. The second term is small thanks to $\delta$, while for the first term we can use the spectral norm bound, but in a careful way. Recalling that $\Z^\top \A \Z \preceq \Z^\top \B \Z$ if $\A \preceq \B$, and using that $\P=\P^2$ for a projection matrix:
\begin{align*}
    \E[\P(\A^\top\A)^{-1}\P \cdot 1_{\Ec}] 
    &= \E[\P\cdot\P(\A^\top\A)^{-1}\P\cdot\P \cdot 1_{\Ec}]
    \\
    &\preceq \frac{Cn}{2k\sigma_{2k}^2+\sum_{i>2k}\sigma_i^2}\E[\P \cdot 1_{\Ec}].
\end{align*}
Putting those together, since $\E[\P \cdot 1_{\Ec}] \preceq \E[\P]$ and $\P \preceq \I$: 
\begin{align*}
    \E[\P(\A^\top\A)^{-1}\P] \preceq \frac{Cn}{2k\sigma_{2k}^2+\sum_{i>2k}\sigma_i^2}\E[\P] + \delta\sigma_{\min}^{-2}\I.
\end{align*}
Now, we are ready to return to the second-order expression $\E[\P\E[\P]^{-1}\P]$. 
\begin{align*}
\E\big[\P\E[\P]^{-1}\P\big] &\preceq 2\E[\P] + 2\lambda\E[\P(\A^\top\A)^{-1}\P]
\\
&\preceq 2\E[\P] + \frac{2C\lambda n}{2k\sigma_{2k}^2+\sum_{i>2k}\sigma_i^2}\E[\P] + 2\lambda\delta\sigma_{\min}^{-2}\I
\\
&= 2\Big(1 + C\frac n\ell \gamma\Big)\E[\P] + 2\lambda\delta\sigma_{\min}^{-2}\I,\quad\text{where}\quad
\gamma = \frac{\sum_{i>\ell}\sigma_i^2}{2k\sigma_{2k}^2+\sum_{i>2k}\sigma_i^2}.
\end{align*} 
Finally, to conclude the proof, note that $\lambda \geq \sigma_{\min}^2$ and that
\begin{align*}
    \gamma 
    &= \frac{\sum_{i>\ell}\sigma_i^2}{2k\sigma_{2k}^2+\sum_{i>2k}\sigma_i^2} \\
    &= \frac{\sum_{i>\ell}^{2k}\sigma_i^2}{2k\sigma_{2k}^2+\sum_{i>2k}\sigma_i^2} + \frac{\sum_{i>2k}\sigma_i^2}{2k\sigma_{2k}^2+\sum_{i>2k}\sigma_i^2}
    \\
    &\leq \frac1{2k}\sum_{i>\ell}^{2k}\frac{\sigma_i^2}{\sigma_{2k}^2} + 1 \leq \bar\kappa_{\ell:2k}^2+1.
\end{align*}
To bound $\nu$, we multiply $\E[\P\E[\P]^{-1}\P]$ by $\E[\P]^{-1/2}$ from both sides to get $\E[(\E[\P]^{-1/2}\P\E[\P]^{-1/2})^2]$, 
and then take the spectral norm, noting that 
\begin{align*}
    \delta \|\E[\P]^{-1}\| = \frac{\delta}{\mu} \leq 
    \frac{c n}{\ell} \frac{\bar\kappa_{\ell}^2}{n \kappa^2} \leq \frac{cn}{\ell},
\end{align*}
and so the contribution of the $\delta$ can be absorbed into the constant.
\end{proof}

\section{Convergence Analysis with Inexact Projections}\label{s:inexact}
Step \ref{line-u} of Algorithm \ref{alg:main} requires finding $\u_t$ that solves the following linear system:
\begin{align}
(\S\A\B^{-1}\A^\top\S^\top) \u_t = \tilde\b\quad\text{where}\quad\tilde\b=\S\A\y_t-\S\b.\label{eq:solve}
\end{align}
Then, in Step \ref{line-w}, the algorithm computes $\w_t=\B^{-1}\A^\top\S^\top\u_t$. The cost of solving \eqref{eq:solve} is the main cost involved in each iteration of sketch-and-project. Recall that in our main results, we consider two different choices of $\B$ depending on whether $\A$ is positive definite or not.
\smallskip

\noindent
\textbf{Positive definite case.} In the case where matrix $\A$ is positive definite, we choose $\B=\A$, which means that the linear system we must solve becomes $(\S\A\S^\top)\u_t = \tilde\b$. The matrix $\S\A\S^\top$ can be computed in time $\tilde O(nk)$ and inverted in time $O(k^\omega)$, so we can solve the projection step \eqref{eq:solve} exactly in time $\tilde O(nk + k^\omega)$.
\smallskip

\noindent
\textbf{General case.} On the other hand, in the case of solving a general linear system, where we choose $\B=\I$, the resulting linear system is:
\begin{align*}
    (\tilde\A\tilde\A^\top)\u_t=\tilde\b,\quad\text{for}\quad \tilde\A=\S\A.
\end{align*}
Here, the matrix $\tilde\A=\S\A$ can be precomputed in $\tilde O(nk)$ time, however solving the linear system requires computing the matrix $\tilde\A\tilde\A^\top$, which would take $O(nk^{\omega-1})$ to do exactly. To avoid that, we instead solve this projection step approximately using Preconditioned Conjugate Gradient (PCG), which will allow us to recover the same $\tilde O(nk+k^\omega)$ time as in the positive definite case.

Naturally, the quality of the approximation in the projection step influences the convergence rate of Algorithm \ref{alg:main}. Our next Lemma~\ref{lem:approx_inv} shows that as soon as $\w_t$ is estimated with small enough error, Algorithm \ref{alg:main} converges in expectation with convergence rate determined by $\mu$ and $\nu$ that essentially matches the rate achieved with an exact projection step. For the sake of simplicity, we present this result only for the case of $\B=\I$.
\begin{lem}[Algorithm \ref{alg:main} with inexact projections]\label{lem:approx_inv}
Suppose that $\B=\I$, and steps \ref{line-u}, \ref{line-w} of Algorithm~\ref{alg:main} result in an estimate $\hat\w_t \in\text{range}(\A^\top)$ of $\w_t = \tilde\A^\top(\tilde\A\tilde\A^\top)^{-1}\tilde\b$ such that 
$$\|\hat\w_t-\w_t\|\leq \epsilon \|\w_t\| \quad \text{  with }\quad  \epsilon \le \frac{\mu}{\sqrt{8\|\E[\P]\|(\mu \nu + 1)}}.$$ 
Then, with $\beta=1-\frac{3}{4} \sqrt{\frac{\mu}{\nu}}$, $\gamma=\frac1{\sqrt{\mu\nu}}$, $\alpha=\frac1{1+\gamma\nu}$, we have that:
    \begin{align*}
\E\big[\Delta_t\big] \leq \bigg(1-\frac{1}{2}\sqrt{\frac{\mu}{\nu}}\,\bigg)^t\Delta_0,
    \end{align*}   
    where the error function is defined as $\Delta_t: =
\|\v_t-\x^*\|_{\E[\P]^\dagger}^2+
\frac{1}{\mu}
\|\x_t-\x^*\|^2$.
\end{lem}
\begin{remark}
    Exact knowledge of $\mu$ and $\nu$ is not required: Given $\tilde\mu$ and $\tilde\nu$ such that $\mu\geq\tilde\mu$ and $\nu\leq \tilde\nu$, the hyperparameters $\alpha$, $\beta$, and $\gamma$ can be set using $\tilde\mu$ and $\tilde\nu$ in place of $\mu$ and $\nu$, and an identical convergence analysis recovers the rate of $\E[\Delta_t]\leq\big(1-\frac12\sqrt{\tilde\mu/\tilde\nu}\big)^t\Delta_0$ under the  inexact projection condition with $\epsilon =\frac{\tilde\mu}{\sqrt{8(\tilde\mu\tilde\nu+1)}}$, where we used that $\|\E[\P]\|\leq 1$.
\end{remark}
\begin{remark}\label{r:exact-analysis}
    When Algorithm \ref{alg:main} uses exact projections (steps \ref{line-u} and \ref{line-w}), then \cite{gower2018accelerated} showed a convergence bound which can be stated for a general positive definite $\B$. Redefining
    \begin{align*}
        \P &:= \B^{-1/2}\tilde\A^\top(\tilde\A\B^{-1}\tilde\A^\top)^{-1}\tilde\A\B^{-1/2}
        \quad\text{and}
        \\        
        \Delta_t &:= \|\v_t-\x^*\|_{\B^{1/2}\E[\P]^\dagger\B^{1/2}}^2
        +\frac1{\mu}\|\x_k-\x^*\|_{\B}^2,
    \end{align*}    
    they showed the following convergence bound for the exact implementation of Algorithm \ref{alg:main}:
 \begin{align*}
\E\big[\Delta_t\big] \leq \bigg(1-\sqrt{\frac{\mu}{\nu}}\,\bigg)^t\Delta_0,
\end{align*}
    Our Lemma~\ref{lem:approx_inv} gets the same rate up to a constant factor $1/2$ without requiring exact projection steps (for the case of $\B=\I$). We note that the constant $1/2$ can be improved to any constant $c < 1$ by tuning the other parameters of Lemma~\ref{lem:approx_inv}, along the same lines.
\end{remark}
Before we proceed with the proof of Lemma \ref{lem:approx_inv}, we now explain how we  use an iterative solver to compute $\u_t$ and $\w_t$ approximately when $\B=\I$. Taking advantage of the fact that the matrix $\tilde\A$ is very wide, we can use sketching to construct an effective preconditioner for the linear system \eqref{eq:solve}. Following \cite{derezinski2023solving}, our preconditioner is computed as follows:
\begin{enumerate}
    \item Compute $\mathbf{\Phi}\tilde\A^\top$, where $\mathbf{\Phi}$ is a $\phi\times n$ sketching matrix;
    \item Construct $\M=\V\Sigmab^{-1}$, where $\U\Sigmab\V^\top$ is the compact SVD of $\mathbf{\Phi}\tilde\A^\top$.
\end{enumerate}
Here, we refer to a standard sketch-and-precondition guarantee to show that PCG with $\M$ as a preconditioner will solve the above linear system in time $\tilde O(nk+k^\omega)$.
\begin{lem}[Based on Lemma 7.1, \cite{derezinski2023solving}]\label{l:pcg}
Given $\tilde{\A} \in \R^{k \times n}$, $\tilde{\b} \in \R^k$, let $\w^* = \tilde{\A}^\top(\tilde\A\tilde\A)^\top\tilde{\b}$. Given $\delta <1/2$, there is a sparse sketching matrix matrix $\mathbf{\Phi} \in \R^{\phi\times n}$ that satisfies $\phi = O(k+\log(1/\delta))$ such that we can compute  $\M$ in time $O(nk \log(k/\delta) + k^{\omega})$ that with probability $1 - \delta$ satisfies $\kappa(\tilde{\A}^\top\M) = O(1)$, in which case, PCG preconditioned with $\M$ will solve the  system $\tilde\A\tilde\A^\top\u=\tilde\b$ in time $O(nk\log1/\epsilon)$, returning $\hat\u$ such that
\begin{align*}
\|\hat\w-\w^*\|\leq\epsilon\|\w^*\|,
\end{align*}
where $\hat\w = \tilde\A^\top\u$. Also,  with probability $1$ and independently of $\M$, $\|\hat\w-\w^*\|\leq 4\|\w^*\|$.
\end{lem}

We now return to the convergence analysis of accelerated sketch-and-project with inexact projections. Recall that, for the sake of simplicity, we show this with matrix $\B=\I$, which is used in Theorem \ref{thm:main}.

\begin{proof}[Proof of Lemma \ref{lem:approx_inv}]
    Let $\e_t := \hat\w_t-\w_t,$ so that by assumption $\|\e_t\| \leq \epsilon\|\w_t\|$. Note that by definition, 
    \begin{align*}
    \w_t &= \tilde{\A}^\top \u = \A^\top\S^\top(\S\A\A^\top\S^\top)^\dagger\S(\A\y_t - \b) =: \P(\y_t - \x^*).
     \end{align*} 
     So, conditioned on the $t$th iteration, since $\P$ is idempotent,
     \begin{align}\label{eq:boundet1}
      \E\big[\|\e_t\|^2\big] \leq \epsilon^2\E \big[\|\w_t\|^2 \big] \leq \epsilon^2 \E \big[(\y_t-\x^*)^\top \P(\y_t-\x^*)\big] \le \epsilon^2 a \|\y_t-\x^*\|^2, 
      \end{align}
      where $a := \|\E[\P]\|$ and thus we have with $\mu = \lambda_{\min}(\E[\P])$
      \begin{align}\label{eq:boundet}
      \E \big[\|\e_t\|_{\E[\P]^\dagger}^2 \big] = \E \big[\e_t^\top\E[\P]^\dagger\e_t \big] \leq \|\E[\P]^\dagger\|\cdot
          \E \big[\|\e_t\|^2 \big] \leq \frac{a \epsilon^2}{\mu}
          \|\y_t-\x^*\|^2.
      \end{align}
      Then, for $r_t = \|\v_t - \x^*\|_{\E[\P]^\dagger}$, observe that
      \begin{align*}
          \E [r_{t+1}^2] &= \|\v_{t+1} - \x^*\|^2_{\E[\P]^\dagger}\\
          &= \E
          \big[\|\beta\v_t + (1-\beta)\y_t - \x^* - \gamma\P(\y_t-\x^*) - \gamma\e_t\|^2_{\E[\P]^\dagger} \big]\\
          &\leq 
         \E \big[\left(\|\beta\v_t + (1-\beta)\y_t - \x^* - \gamma\P(\y_t-\x^*)\|_{\E[\P]^\dagger} + \|\gamma\e_t\|_{\E[\P]^\dagger}\right)^2 \big]\\
          &\leq (1+\eta) \E \big[ \|\beta\v_t + (1-\beta)\y_t - \x^* - \gamma\P(\y_t-\x^*)\|_{\E[\P]^\dagger}^2 \big] +          
          \left(1+\frac{1}{\eta}\right)\E \big[\|\gamma\e_t\|_{\E[\P]^\dagger}^2 \big],
      \end{align*}
      where in the last inequality we have used the fact that for any $w,z\in\R$ and for any $\eta > 0$ it holds that $2wz\leq \eta w^2 +\frac{1}{\eta}z^2 $ (we will choose $\eta$ later). Then, using \eqref{eq:boundet} and expanding the square, we get
       \begin{align}\label{eq:theIs}
          \E [r_{t+1}^2] \le  (1+\eta)I + (1+\eta)&\gamma^2II - 2(1+\eta)\gamma III + \left(1+\frac{1}{\eta}\right)\gamma^2\frac{a \epsilon^2}{\mu}
          \|\y_t-\x^*\|^2, \\
          \text{ where } \quad &I =  \|\beta\v_t + (1-\beta)\y_t - \x^*\|_{\E[\P]^\dagger}^2, \nonumber\\
          &II = \E \big[\|\P(\y_t-\x^*)\|_{\E[\P]^\dagger}^2 \big], \nonumber\\
          &III = \E \big[\left\langle \beta(\v_t - \x^*) + (1-\beta)(\y_t - \x^*), \E[\P]^\dagger\P(\y_t - \x^*)\right\rangle \big].\nonumber
      \end{align}
The terms $I$, $II$, $III$ were estimated in  \cite{gower2018accelerated} to get
      \begin{align}\label{eq:boundI}
          &I \leq \beta r_t^2 + (1-\beta)\|\y_t - \x^*\|_{\E[\P]^\dagger}^2,\nonumber\\
          &\E[II|\y_t] \leq \nu\|\y_t - \x^*\|_{\E[\P]^\dagger}^2, \\
          &\E[III | \y_t, \v_t, \x_t] = \|\y_t - \x^*\|^2 - \beta\frac{1-\alpha}{2\alpha}(\|\x_t - \x^*\|^2 - \|\y_t-\x_t\|^2 - \|\y_t-\x^*\|^2).\nonumber
      \end{align}
To estimate further, we note
\begin{align}
    \E[\|\x_{t+1} - \x^*\|^2 | \y_t] &= \E[\|(\I-\P)(\y_t - \x^*) - \e_t\|^2  | \y_t]\notag\\
    &\leq (1+\eta)\|\y_t - \x^*\|^2 - (1+\eta)\|\y_t - \x^*\|_{\E[\P]}^2 + \left(1+\frac{1}{\eta}\right) \E \big[\|\e_t\|^2 \big],\notag
\end{align}
      and thus
      \begin{align}\label{eq:newlem8}
          \|\y_t - \x^*\|_{\E[\P]}^2 \leq \|\y_t -\x^*\|^2 - \frac{1}{1+\eta} \E[\|\x_{t+1} - \x^*\|^2 | \y_t] + \frac{1+1/\eta}{1+\eta} \E \big[\|\e_t\|^2 \big].
      \end{align}

      Combining estimates \eqref{eq:boundI} and \eqref{eq:newlem8} back in \eqref{eq:theIs}, and using \eqref{eq:boundet1},  we conclude 
      \begin{align}
          \E[r_{t+1}^2 &| \y_t,\v_t,\x_t] \leq (1+\eta)\beta r_t^2 + (1+\eta)\frac{(1-\beta)}{\mu}\|\y_t - \x^*\|^2 \\
          &+ (1+\eta)\gamma^2 \nu\left(\|\y_t -\x^*\|^2 - \frac{1}{1+\eta} \E[\|\x_{t+1} - \x^*\|^2 | \y_t] + \frac{1+1/\eta}{1+\eta}\epsilon^2a\|\y_t-\x^*\|^2 \right) \notag\\
          &+ 2(1+\eta)\gamma \left(-\|\y_t - \x^*\|^2 + \beta\frac{1-\alpha}{2\alpha}(\|\x_t - \x^*\|^2 - \|\y_t-\x^*\|^2) \right)\notag\\
          &+ \left(\frac{1 + \eta}{\eta}\right)\frac{\gamma^2 a \epsilon^2}{\mu}\|\y_t-\x^*\|^2.\notag
      \end{align}
Or equivalently,
\begin{align}
    \E[r_{t+1}^2 &+ \gamma^2\nu\|\x_{t+1}-\x^*\|^2 | \y_t,\v_t,\x_t] \notag\\
    &\leq (1+\eta)\beta r_t^2 + (1+\eta)\beta\underbrace{\gamma\frac{1-\alpha}{\alpha}}_{P_1}\|\x_t-\x^*\|^2 \\
     &+ (1+\eta)\underbrace{\left(\frac{(1-\beta)}{\mu} +\gamma^2\nu - 2\gamma -\beta\gamma\frac{1-\alpha}{\alpha} + \frac{\gamma^2 a \epsilon^2}{\eta}(\nu   + \frac{1}{\mu})  \right)}_{P_2}\|\y_t - \x^*\|^2.\notag
\end{align}
Note that with $\alpha = (1 + \gamma \nu)^{-1}$ and $\gamma = 1/\sqrt{\nu \mu}$ we have $P_1 = \gamma^2\nu=1/\mu$. 

It remains to choose $\epsilon$ and $\eta$ so that $P_2 \leq 0$ and $(1+\eta)\beta \leq 1- \frac{1}{2}\sqrt{\frac{\mu}{\nu}}$. With $\beta = 1 - \frac{3}{4}\sqrt{\frac\mu\nu}$, the second inequality is achieved by taking $\eta=\frac14\sqrt{\frac\mu\nu}$.
The same parameters achieve $P_2 \le 0$ as soon as $\epsilon \le \frac{\mu}{\sqrt{8a(\mu \nu + 1)}}.$
Then, we have
$$
\E\left[r_{t+1}^2 + \frac1{\mu}\|\x_{t+1}-\x^*\|^2 \;|\; \y_t,\v_t,\x_t\right] \leq \left(1-\frac{1}{2}\sqrt{\frac{\mu}{\nu}}\right)\left(r_t^2 + \frac1{\mu}\|\x_{t}-\x^*\|^2\right),
$$
iterating for $t$, $t-1$, \ldots $0$ we complete the proof.
\end{proof}

\section{Completing the Proofs of Main Results}
\label{s:completing}
In this section, we complete the proof of the main results, including Theorem \ref{thm:main} (solving general linear systems), Theorem \ref{thm:psd} (solving positive definite linear systems) and Corollary~\ref{t:large-sv} (solving linear systems with $\ell$ large singular values). We also discuss how the claims stated in Section \ref{s:intro} follow from these results.

\begin{proof}[Proof of Theorem~\ref{thm:main}]
  Recall that in Theorem \ref{thm:main} we consider an $m\times n$ matrix $\A$ with rank $n$, and a vector $\b$ such that the linear system $\A\x=\b$ is consistent. We then preprocess both the matrix $\A$ and the vector $\b$ with a randomized Hadamard transform $\Q$ so that $\A$ is replaced by $\Q\A$ and $\b$ is replaced by $\Q\b$, which does not change the solution to the linear system. Then, Algorithm \ref{alg:main} is ran with $\B=\I$, and using the inexact PCG solve for step \ref{line-u} as described in Lemma \ref{l:pcg}.

  From Lemma \ref{l:pcg}, conditioned on an event $\Ec_t$ that holds with probability $1-\delta_{\mathrm{PCG}}$, in time $O(nk\log(k/\mu\delta_{\mathrm{PCG}})+k^\omega)$ PCG solves the $t$-th projection step to within sufficient accuracy so that we can apply our convergence analysis of inexact sketch-and-project from Lemma~\ref{lem:approx_inv}. 
  Using Lemma~\ref{lem:approx_inv}, we can recover convergence in terms of  $\|\x_t-\x^*\|^2$ error by converting from the error $\Delta_t=\|\v_t-\x^*\|_{\E[\P]^{\dagger}}^2+\frac1\mu\|\x_t-\x^*\|^2$, where recall that we define the projection matrix $\P=(\S\A)^{\dagger}\S\A$ and denote its smallest eigenvalue as $\mu = \lambda_{\min}(\E[\P]) = \|\E[\P]^\dagger\|^{-1}$. 
  Then, with the choice of $\alpha,\beta,\gamma$ as in Lemma~\ref{lem:approx_inv}, 
  \begin{align*}   
  \E[\Delta_{t+1}] &= \Pr(\Ec_t)\E[\Delta_{t+1}\mid \Ec_t] + \Pr(\neg\Ec_t)\E[\Delta_{t+1}\mid\neg\Ec_t]
      \\
    &\leq \bigg(1-\frac{1}{2}\sqrt{\frac{\mu}{\nu}}\,\bigg)\E[\Delta_t] + \delta_{\mathrm{PCG}}\cdot \E[\Delta_{t+1}\mid\neg\Ec_t],
  \end{align*}
  where recall that $\nu=\lambda_{\max}(\E[(\bar\P^{-1/2}\P\bar\P^{-1/2})^2])$.

  To handle the $\neg\Ec_t$ case, we use the second estimate of Lemma~\ref{l:pcg} which holds unconditionally and states that $\|\hat\w_t-\w_t\|\leq 4\|\w_t\|$ where $\hat\w_t$ is the result of the steps \ref{line-u} and \ref{line-w} of Algorithm~\ref{alg:main},  performed inexactly using PCG. Then, $$\|\hat\w_t\|\leq \|\w_t\|+\|\hat\w_t-\w_t\|\leq5\|\w_t\|\leq 5\|\y_t-\x^*\|, $$ where the last step follows since $\|\w_t\|=\|\P(\y_t-\x^*)\|\leq\|\y_t-\x^*\|$. Using this, direct calculation shows that:
  \begin{align*}
      \Delta_{t+1} 
      &\leq \frac1\mu\|\v_{t+1}-\x^*\|^2 + \frac1\mu\|\x_{t+1}-\x^*\|^2
      \\
      &\leq\frac1{\mu}\Big(\|\v_t-\x^*\|+\|\y_t-\x^*\|+\gamma\|\hat\w_t\|\Big)^2 + \frac1\mu\Big(\|\y_t-\x^*\|+\|\hat\w_t\|\Big)^2
      \\
      &\leq \frac1\mu\Big(\|\v_t-\x^*\|+ (1+5\gamma)\|\y_t-\x^*\|\Big)^2 + \frac{6^2}\mu\|\y_t-\x^*\|^2
      \\
      &\leq\frac2{\mu}\|\v_t-\x^*\|^2 + \frac{2(1+5\gamma)^2+6^2}{\mu}\|\y_t-\x^*\|^2.
  \end{align*}
  Now, using that $\|\y_t-\x^*\|\leq \alpha\|\v_t-\x^*\|+\|\x_t-\x^*\|$, $\alpha\gamma\leq 1$, and $\gamma^2\leq1/\mu$, we obtain 
  $$\Delta_{t+1} = O\Big(\frac1\mu\|\v_t-\x^*\|^2 + \frac1{\mu^2}\|\x_t-\x^*\|^2\Big) \leq \frac C\mu\Delta_t$$ for some constant $C>0$.
  Thus, setting $\delta_{\mathrm{PCG}}= \frac{\mu^{3/2}}{4C\nu^{1/2}}$ and adjusting the runtime accordingly, we obtain $\E[\Delta_{t+1}]\leq (1 - \frac14\sqrt{\mu/\nu})\E[\Delta_t]$. It follows that:
  \begin{align*}
      \E\big[\|\x_t-\x^*\|^2\big]
      &\leq  
      \mu\E\big[\Delta_t\big]\leq{\leq}
\bigg(1-\frac{1}{4}\sqrt{\frac{\mu}{\nu}}\,\bigg)^t \mu\Delta_0 \\
      &= \bigg(1-\frac{1}{4}\sqrt{\frac{\mu}{\nu}}\,\bigg)^t\big(\mu\|\x_0-\x^*\|_{\E[\P]^\dagger}^2+\|\x_0-\x^*\|^2\big)
      \\
      &\leq 2\bigg(1-\frac{1}{4}\sqrt{\frac{\mu}{\nu}}\,\bigg)^t\|\x_0-\x^*\|^2,
  \end{align*}
Now, we can use our first-order and second-order projection analysis to bound $\mu$ and $\nu$ respectively, obtaining:
\begin{align*}
    \text{(Corollary \ref{c:mu-less})}\qquad\mu &\geq \frac{c_1 \ell}{n\bar\kappa_\ell^2},
    \\
    \text{(Corollary \ref{c:nu-bound})}\qquad\nu &\leq \frac{c_2n}{\ell}\bar\kappa_{\ell:2k}^2,
\end{align*}
where recall that $\ell$ is a value that satisfies $\ell\geq \frac{ck}{\log k}$. Putting these together with the above convergence guarantee we obtain:
    $$
     \E\big[\|\x_t-\x^*\|^2\big]\leq 2\bigg(1- \frac{\ell\sqrt{c_1/c_2}}{4n\bar\kappa_{\ell}\bar\kappa_{\ell:2k}}\,\bigg)^t\|\x_0-\x^*\|^2.$$
The overall time complexity of the algorithm includes the $O(mn\log m)$ cost of performing the randomized Hadamard transform plus the per-iteration cost that is dominated by $O(nks)=O(nk\log^4(n/\delta))$ for computing  the sketch $\S\A$ using a $k\times m$ LESS-uniform matrix $\S$ with $s$ non-zeros per row, plus the $O(nk\log(k/\mu\delta_{\mathrm{PCG}})+k^\omega)=O(nk\log(n\bar\kappa_\ell)+k^\omega)$ cost of preconditioning and solving the linear system in step \ref{line-u}. Thus, the overall runtime is:
\begin{align*}
    O(mn\log m + t(nk(\log^4(n/\delta)+\log(n\bar\kappa_\ell))+k^\omega)),
\end{align*}
concluding the proof.
\end{proof}

\smallskip

\noindent
Next, we discuss how the above proof needs to be adapted for solving positive definite linear systems, in order to obtain the improved dependence on the condition numbers, thus proving Theorem \ref{thm:psd}.

\begin{proof}[Proof of Theorem \ref{thm:psd}]
Recall that there are several key differences in how we implement Algorithm \ref{alg:main} when $\A$ is positive definite. First, the randomized Hadamard transform preprocessing step has to be done in a way that preserves the positive definite structure of the problem. To that end, we must apply the matrix $\Q$ symmetrically on both sides of $\A$, thus replacing the original system $\A\x=\b$ with:
\begin{align*}
    \Q\A\Q^\top\x = \Q\b.
\end{align*}    
We note that, while the resulting linear system is not strictly equivalent to the original one, we can easily recover the original solution at the end by returning $\Q^\top\x$ instead of $\x$, which can be computed in additional $O(n\log n)$ time. For clarity of exposition, let us now denote the preprocessed matrix $\Q\A\Q^\top$ as $\bar\A$, and the preprocessed vector $\Q\b$ as $\bar\b$.

 The next difference in the positive definite case is in the choice of matrix $\B$. While in Theorem \ref{thm:main} we effectively dropped the matrix $\B$ by using $\B=\I$, here we set $\B=\bar\A$. The main difference in the implementation of the algorithm resulting from this change is, as we discussed in Section \ref{s:inexact}, the fact that the linear system in step \ref{line-u} is now $(\S\bar\A\S^\top)\u_t=\S\bar\A\y_t-\S\b$, which can be solved more easily than in the general case. Specifically, we first precompute $\S\bar\A$ in time $O(nks)=O(nk\log^4(n/\delta))$, then we compute $\S\bar\A\S^\top$ in time $O(k^2s)$ and $\S\bar\A\y_t-\S\b$ in time $O(nk)$. Finally, to solve the system, we just need to invert the $k\times k$ matrix $\S\bar\A\S^\top$ in time $O(k^\omega)$. 

Finally, we also need to adapt our convergence analysis to the $\B=\bar\A$ case. Here, we rely on the original analysis of exact sketch-and-project from \cite{gower2018accelerated} (see Remark~\ref{r:exact-analysis}). Using the notation from the remark, similarly as in the proof of Theorem \ref{thm:main} we can show:
\begin{align*}
    \E\big[\|\x_t-\x^*\|_{\B}^2\big]\leq \bigg(1-\sqrt{\frac\mu\nu}\bigg)^t\mu\Delta_0\leq 2\bigg(1-\sqrt{\frac\mu\nu}\bigg)^t\|\x_0-\x^*\|_{\B}^2,
\end{align*}
where $\mu$ and $\nu$ are computed with respect to the sketched projection matrix defined for the matrix $\bar\A\B^{-1/2}$. Since $\B=\bar\A=\Q\A\Q^\top$, then $\bar\A\B^{-1/2} = \bar\A^{1/2} = \Q\A^{1/2}\Q^\top$, so we can write the projection matrix as:
\begin{align*}
    \P &= \bar\A^{1/2}\S^\top(\S\bar\A\S^\top)^{-1}\S\bar\A^{1/2}
    \\
    &=\Q\A^{1/2}\Q\S^\top(\S\Q\A\Q\S^\top)^{-1}\S\Q\A^{1/2}\Q^\top = \Q\P'\Q^\top,
\end{align*}
where $\P':=\A^{1/2}\Q\S^\top(\S\Q\A\Q\S^\top)^{-1}\S\Q\A^{1/2}$ is also a projection matrix. Note that applying the orthogonal matrix $\Q$ on both sides of $\P'$ does not affect the calculation of $\mu$ and $\nu$, so we can equivalently compute them from the matrix $\P'$. Fortunately, this matrix is precisely the sketched projection we analyze in Corollaries \ref{c:mu-less} and \ref{c:nu-bound} if we replace matrix $\A$ with matrix $\A^{1/2}$. Thus, in our calculations, the squared singular values become the (non-squared) eigenvalues, and we obtain the corresponding bounds:
\begin{align*}
    \text{(adapted Corollary \ref{c:mu-less})}\qquad\mu &\geq \frac{c_1 \ell}{n\bar\kappa_\ell^2(\A^{1/2})} = \frac{c_1 \ell}{n\tilde\kappa_\ell},
    \\
    \text{(adapted Corollary \ref{c:nu-bound})}\qquad\nu &\leq \frac{c_2n}{\ell}\bar\kappa_{\ell:2k}^2(\A^{1/2}) = \frac{c_2n}{\ell}\tilde\kappa_{\ell:2k}.
\end{align*}
The rest of the analysis proceeds same as in the general case, except with these better condition numbers.
\end{proof}
\smallskip

\noindent
We next show how to obtain Theorem \ref{t:informal-main} from Theorem \ref{thm:main}. The main issue here is to translate the somewhat complex condition numbers appearing in our convergence rates into the simple condition number $\kappa_\ell=\sigma_\ell/\sigma_n$, and then perform the time complexity analysis of the resulting algorithm. Almost identical arguments (omitted here) can be applied to recover Theorem \ref{t:informal-psd} from Theorem \ref{thm:psd}.

\begin{proof}[Proof of Theorem \ref{t:informal-main}] 
Suppose that our sketch size satisfies $k=\tilde O(n^{\frac1{\omega-1}})$ and $\ell=\Omega(k/\log k)$. Then, recall that Theorem \ref{thm:main} implies:
\begin{align*}
    \tilde O\Big(n^2 + \bar\kappa_{\ell}\bar\kappa_{\ell:2k}\frac n\ell (nk + k^\omega)\Big)
= \tilde O\Big(\bar\kappa_{\ell}\bar\kappa_{\ell:2k}(n^2 + nk^{\omega-1})\Big) = \tilde O\big(\bar\kappa_{\ell}\bar\kappa_{\ell:2k}n^2\big).
\end{align*}
We will now show that, with the right choice of $k$, we can bound the condition numbers $\bar\kappa_{\ell}\bar\kappa_{\ell:2k}$ by $\tilde O(\kappa_\ell)$, where $\kappa_\ell=\frac{\sigma_\ell}{\sigma_n}$. We will also use the shorthand $\kappa_{\ell:q}=\frac{\sigma_\ell}{\sigma_q}$. 
We use the following bound:
\begin{align}
    \bar\kappa_{\ell}\bar\kappa_{\ell:2k} 
    &= \sqrt{\frac{\sum_{i>\ell}\sigma_i^2}{(n-\ell)\sigma_n^2}\bar\kappa^2_{\ell,2k}}
     \leq 
        \sqrt{\frac{(2k-\ell)\sigma_{\ell}^2+(n-2k)\sigma_{2k}^2}{(n-\ell)\sigma_{n}^2}\,
        \frac{\sigma_{\ell}^2}{\sigma_{2k}^2}}
        \nonumber\\
& = \frac{\sigma_\ell}{\sigma_n}\sqrt{\frac{2k-\ell}{n-\ell}\frac{\sigma_\ell^2}{\sigma_{2k}^2} + \frac{n-2k}{n-\ell}}\leq \kappa_\ell\sqrt{\frac{2k}n\kappa_{\ell:2k}^2+1},\label{eq:doubling-trick}
\end{align}
where we used the fact that $\bar\kappa_{\ell,2k}\leq\kappa_{\ell,2k}$. Thus, as long as the term under the square root is $\tilde O(1)$, then we have obtained the desired time complexity $\tilde O(\kappa_\ell n^2)$ with $\ell=\Omega(n^{\frac1{\omega-1}})$. Naturally, this may not be true for some choices of $k$, namely if there is a sharp drop in the singular values in the range from $\sigma_\ell$ to $\sigma_{2k}$. To address this, we devise an iterative procedure that always finds a value of $k$ that avoids this phenomenon.

Recall that Theorem \ref{thm:main} uses $\ell=\lceil\frac{c_2 k}{\log k}\rceil$, which means that the ratio of $\frac {2k}{\ell}$ can be bounded by $b := C\log k$. Let us start with an initial choice of $k$ being some $k_0 = \lceil n^{\frac1{\omega-1}}\log n\rceil$ and the corresponding initial value of $\ell$ equal to $\ell_0= \lceil 2k_0/b\rceil =\Omega(n^{\frac1{\omega-1}})$. First, note that we can assume without loss of generality that $\kappa_{\ell_0} \leq n^{\omega-2}$, because otherwise the desired complexity $\tilde O(\kappa_\ell n^2)$ can be obtained by using a direct $O(n^\omega)$ time solver. Now, consider the following procedure, that iterates over $q=0,1,...$:
\begin{enumerate}
\item If $\kappa_{\ell_q:2k_q}\leq \sqrt{n/k_0}$ or $\kappa_{\ell_q}^2\leq \kappa_{\ell_0}$, then stop and choose $k=k_q$.
\item Otherwise, set $k_{q+1} = \lceil bk_q\rceil$, with $\ell_{q+1}=\lceil 2k_{q+1}/b\rceil$, and repeat the procedure with $q=q+1$.
\end{enumerate}
Now, observe that after each iteration in which we do not stop, we have $\kappa_{\ell_q:2k_q}>\sqrt{n/k_0}$, which means that: 
\begin{align*}
    \kappa_{\ell_{q+1}} = \frac{\sigma_{\ell_{q+1}}}{\sigma_n} = \frac{\sigma_{\ell_q}}{\sigma_n}\,
    \frac{\sigma_{\ell_{q+1}}}{\sigma_{\ell_q}}
    \leq \kappa_{\ell_q}\frac{\sigma_{2k_q}}{\sigma_{\ell_q}}=\frac{\kappa_{\ell_q}}{\kappa_{\ell_q:2k_q}}\leq \sqrt{k_0/n}\cdot \kappa_{\ell_q},
\end{align*}
where we used that fact that $\ell_{q+1}\geq 2k_q$, and therefore $\sigma_{\ell_{q+1}}\leq \sigma_{2k_q}$. So, going from $q$ to $q+1$ we decrease the condition number $\kappa_{\ell_q}$ by at least a factor of:
\begin{align*}
    \sqrt{k_0/n}= \tilde O( n^{\frac12(\frac1{\omega-1}-1)}\log^{1/2} n) = \tilde O(n^{-\frac{\omega-2}{2(\omega-1)}}\log^{1/2}n).
\end{align*}
This means that after $q$ iterations of this procedure, we get:
\begin{align*}
    \kappa_{\ell_q}^2 \leq \kappa_{\ell_0}^2n^{-q\frac{\omega-2}{\omega-1}}\log^q n
\leq\kappa_{\ell_0} n^{(\omega-2)(1-\frac q{\omega-1})}\log^q n,
\end{align*}
where we used the assumption that $\kappa_{\ell_0} = O(n^{\omega-2})$. 
Now, note that since $\omega-1< 2$, we have $n^{(\omega-2)(1-\frac q{\omega-1})}<n^{-\theta}$ for some $\theta > 0$ for any $q\geq 2$. So already after $q\leq 2$ iterations, the stopping condition $\kappa_{\ell_q}^2\leq\kappa_{\ell_0}$ must occur. If it does, then we get $\bar\kappa_{\ell_q}\bar\kappa_{\ell_q:2k_q}\leq\kappa_{\ell_q}^2\leq \kappa_{\ell_0}$, which results in the desired time complexity $\tilde O(\kappa_{\ell_0}n^2)$. On the other hand, if the other stopping condition occurs, i.e., $\kappa_{\ell_q:2k_q}\leq \sqrt{n/k_0}$, then we can use the bound in \eqref{eq:doubling-trick} to get
\begin{align*}
\bar\kappa_{\ell_q}\bar\kappa_{\ell_q:2k_q}\leq \kappa_{\ell_q}\sqrt{\frac{2k_q}{n}\kappa_{\ell_q:2k_q}^2+1}
\leq \kappa_{\ell_q}\sqrt{\frac{2k_q}{n}\frac{n}{k_0}+1} = O\big(\kappa_{\ell_0}\cdot \log^{q/2}k\big).
\end{align*}
Since $q\leq 2$, we again recover the time complexity bound $\tilde O(\kappa_{\ell_0}n^2)$. In other words, this procedure shows that in the worst case we only need to increase the sketch size $k$ by a factor of $O(\log^2 n)$ to find the value that will lead to the desired time complexity.
\end{proof}
\smallskip

\noindent
Next, we describe how our improved DPP coupling argument used in Lemma \ref{l:mu-less} allows us to give a direct improvement in the time complexity of solving linear systems with $\ell$ large singular values, compared to the previous result from \cite{derezinski2023solving}.

\begin{proof}[Proof of Corollary \ref{t:large-sv}]
    Consider an $n\times n$ linear system $\A\x=\b$ such that all of the singular values of matrix $\A$ except for the top $\ell$ are within a constant factor of each other, i.e., $\sigma_{\ell+1}(\A) = O(\sigma_n(\A))$, which is the setting from \cite{derezinski2023solving}. In particular, this implies that for this matrix $\bar\kappa_\ell \leq C=O(1)$. Let us consider a slightly different implementation of Algorithm~\ref{alg:main}, where instead of a LESS-uniform sketching matrices $\S$ with $s\geq C\log^4(n/\delta)$ non-zeros per row, we are going to use a block sub-sampling matrix (equivalent to taking LESS-uniform with $s=1$). Crucially, our first-order projection analysis (Lemma \ref{l:mu-less}) still applies in this case, showing that:
    \begin{align*}
        \text{(Corollary \ref{c:mu-less})}\qquad \mu\geq \frac{c_1\ell}{n\bar\kappa_\ell^2}\geq \frac{c_1 \ell}{C^2n}
        \quad\text{for}\quad k=O(\ell\log \ell). 
    \end{align*}
    While our second-order projection analysis does not apply for a block sub-sampling matrix, here we can use a simpler upper bound on $\nu$, which was given in Lemma 2 of \cite{gower2018accelerated}, stating that $\nu\leq 1/\mu$. This gives a convergence guarantee of the form $\E[\|\x_t-\x^*\|^2]\leq 2(1-\mu)^t\|\x_0-\x^*\|^2$, and the resulting iteration complexity is $t=O(\mu^{-1}\log(1/\epsilon)) = O(\frac n\ell\log(1/\epsilon))$. Note that, for this variant of the algorithm, the per iteration time complexity is $O(nk)$ to construct the sketch $\S\A$, and $O(nk\log(n)+k^\omega)$ to solve the projection step \ref{line-u}. Thus, we get the following overall time complexity:
    \begin{align*}
        O\Big(n^2\log n + \frac n\ell(nk\log(n) + k^\omega)\log1/\epsilon\Big)
        = O\Big((n^2\log^2 n+ n\ell^{\omega-1}\log^\omega \ell\big)\log1/\epsilon\Big).
    \end{align*}
    We note that the result previously obtained for this problem by \cite{derezinski2023solving} relied on a different DPP coupling argument, which requires $k=O(\ell\log^3n)$. We improve this to $k=O(\ell\log \ell)$, leading to better logarithmic factors in the time complexity.
\end{proof}

\paragraph{Extension to sparse least squares.}
Finally, we discuss how our algorithms can be adapted and applied in the tall and sparse least squares setting. Recall that in the least squares task, we are given a tall $n\times d$ matrix $\A$ with $\nnz(\A)$ non-zero entries, and an $n$-dimensional vector $\b$, and our goal is to minimize $f(\x)=\frac12\|\A\x-\b\|^2$. This can be achieved by running preconditioned gradient descent on $f$, which takes the form of $\x_{t+1} = \x_t - \eta\hat\H^{-1}\g_t$, where $\g_t=\nabla f(\x_t)=\A^\top(\A\x_t-\b)$ is the gradient, which can be computed in $O(\nnz(\A))$ time, whereas $\hat\H$ is a preconditioner which approximates the Hessian matrix $\nabla^2f(\x_t) = \A^\top\A$. This preconditioner can be constructed by computing an $\tilde O(d)\times d$ sketch $\tilde\A=\mathbf{\Phi}\A$ where $\mathbf{\Phi}$ is a sparse oblivious subspace embedding (similarly as in Lemma~\ref{l:pcg}; e.g., see \cite{chenakkod2023optimal}), and defining $\hat\H=\tilde\A^\top\tilde\A$. Now, instead of computing $\hat\H^{-1}$ outright at the cost of $\tilde O(d^\omega)$, we use Algorithm \ref{alg:main} to solve the linear system $\hat\H\x = \g_t$ at every iteration. Since $\hat\H$ is positive definite, we can use the version considered in Theorem \ref{thm:psd}, however this has to be further modified to account for the fact that we can never actually form this matrix, but rather must operate on its decomposition $\hat\H=\tilde\A^\top\tilde\A$. Fortunately, this \emph{implicit} positive definite system solver is a simple adaptation of Sketch-and-Project (see Section 9 of \cite{derezinski2023solving}), which is not affected by Nesterov's acceleration. In the end, the overall time complexity of this procedure includes $\tilde O(\nnz(\A))$ for computing $\tilde\A$, and $\tilde O(\nnz(\A) + d^2\kappa_\ell)$ in each iteration of the preconditioned gradient descent, for computing the gradient and then running the solver, respectively, where $\kappa_\ell=\sigma_\ell(\A)/\sigma_d(\A)$. Since the resulting algorithm only requires $O(\log1/\epsilon)$ steps to converge, the final time complexity is $\tilde O((\nnz(\A)+d^2\kappa_\ell)\log1/\epsilon)$.

\section{Lower bound for Matrix-vector Query Algorithms}
\label{s:lower}

In this section, we give lower bounds for a class of linear system solvers, which include conjugate gradient as well as certain randomized preconditioned solvers, in our fine-grained setting where the linear system has a bounded spectral tail condition number $\kappa_{\ell}$. We build on existing lower bounds for solving positive definite systems, given by \cite{braverman2020gradient}, who obtained them with respect to the classical condition number $\kappa$. The model of computation we consider includes all algorithms which access the matrix $\A$ through matrix-vector queries of the form $\A\v$, where the vector $\v$ can be randomized and chosen adaptively.

\begin{dfn}[Matrix-vector query model,  \cite{braverman2020gradient}]\label{d:matvec}
We say that $\Alg$ is a $\MatVec$ algorithm for solving positive definite linear systems if, given initial point $\x_0\in\R^n$ and $\b\in\R^n$, it interacts with a positive definite matrix $\A\in\R^{n\times n}$ via $T$ adaptive randomized queries, $\w_t=\A\v_t$, and returns an estimate $\tilde \x\in\R^n$ of $\A^{-1}\b$. We call $T$ the query complexity of $\Alg$.
\end{dfn}

This computation model includes many standard deterministic iterative algorithms, most notably conjugate gradient (CG), and it also allows for certain preconditioning techniques, for example using Randomized SVD to build a preconditioner that approximates the top-$\ell$ part of the spectrum of $\A$. The central question of this section is:
\begin{quote}
\emph{Given $n$, $\ell < n$ and $\kappa_\ell\geq 1$, what is the $\MatVec$ query complexity of solving an $n\times n$ positive definite linear system $\A\x=\b$ with $\frac{\sigma_{\ell}(\A)}{\sigma_n(\A)}\leq \kappa_\ell$?}
\end{quote}
As discussed in Section \ref{s:intro}, CG can solve this problem using $O(\ell+\sqrt{\kappa_\ell}\log1/\epsilon)$ $\MatVec$ queries. An alternative strategy is to use a preconditioned solver. To do this, we can probe the matrix $\A$ using $\tilde O(\ell)$ Gaussian random vector queries to construct a rank $\ell$ approximation via the Randomized SVD algorithm with power iteration \cite{halko2011finding}. Augmented with a preconditioner based on such an approximation, CG can solve the system using only $O(\sqrt{\kappa_\ell}\log1/\epsilon)$ queries. 
Even though the latter strategy uses randomization, while the former is fully deterministic, the overall query complexity is still no better than CG (although preconditioning is often preferred in practice due to its improved stability properties). Can we achieve a better query complexity in the $\MatVec$ model for our problem?

Next, we show that the guarantee attained by CG is in fact essentially optimal among all $\MatVec$ algorithms (even allowing randomization), up to logarithmic factors. In particular, this means that for any $\ell=\Omega(n^\theta)$, where $\theta>0$, no $\MatVec$ algorithm has time complexity $\tilde O(\sqrt{\kappa_\ell}\cdot n^2\log1/\epsilon)$ for dense positive definite linear systems, yet we show this for Sketch-and-Project with Nesterov's acceleration (Theorem \ref{t:informal-psd}) given any $\ell = O(n^{0.729})$.
\begin{thm}\label{t:lower}
    Any $\MatVec$ algorithm that, given $\x_0$ and an $n\times n$ positive definite linear system $\A\x=\b$ with $\frac{\sigma_{\ell}(\A)}{\sigma_{n}(\A)}\leq \kappa_\ell$ returns $\tilde\x$ such that:
    \begin{align*}
    \Pr\Big(\|\tilde\x - \x^*\|^2_{\A}\leq \epsilon\|\A\|\|\x_0-\x^*\|^2\Big)\geq 1- \frac1e,\quad\text{where }\x^*=\A^{-1}\b,
    \end{align*}
for $\epsilon\leq \frac1e\min\{\ell^{-2},\kappa_\ell^{-1}\}$, must have query complexity at least $\tilde\Omega(\ell + \sqrt{\kappa_\ell})$.
\end{thm}
\begin{remark}\label{r:lower}
    If we additionally assume that the $\MatVec$ algorithm is deterministic, then we can obtain a stronger query complexity lower bound of $\Omega(\ell + \sqrt{\kappa_\ell}\log1/\epsilon)$, which matches the CG upper bound down to constant factors.
\end{remark}

To establish the above result, we actually provide a more general reduction which shows how to obtain a lower bound for our fine-grained linear system task by combining two types of existing lower bounds: one expressed in terms of the dimension of the problem (without restricting the spectrum), and one expressed in terms of the overall condition number of the input matrix.
\begin{lem}\label{l:lower}
  Let $\Ac$ denote some family of $\MatVec$ algorithms, and let $\Lc$ denote the family of square positive definite linear system tasks $(\A,\b,\x_0)$, where $\x_0$ is the starting point. Define $T_{\Ac,\Lc}(n,\kappa,\epsilon,\delta)$ as the minimum query complexity among algorithms in $\Ac$ that solve all $n\times n$ linear systems in $\Lc$ such that $\kappa(\A)\leq \kappa$, so that
    \begin{align}
    \Pr\Big(\|\tilde\x-\x^*\|_{\A}^2\leq\epsilon\|\A\|\|\x_0-\x^*\|^2\Big)\geq 1-\delta,\quad \text{where }\x^*=\A^{-1}\b.\label{eq:guarantee}
    \end{align}
    Now, let $T$ denote the best query complexity among algorithms in $\Ac$ that solve all $n\times n$ linear systems in $\Lc$ with spectral tail condition number $\frac{\sigma_{\ell}(\A)}{\sigma_n(\A)}\leq \kappa_\ell$, in the sense of \eqref{eq:guarantee}. Then, 
    \begin{align*}
        T\ge \max\Big\{T_{\Ac,\Lc}(\ell,\infty,\epsilon,\delta),T_{\Ac,\Lc}(n-\ell,\kappa_\ell,\epsilon,\delta)\Big\}.
    \end{align*}
  \end{lem}
  \begin{remark}
    An analogous reduction can be obtained along the same lines if we replace $\Lc$ with the family of all square linear systems (dropping the positive definiteness).
  \end{remark}
\begin{proof}
First, let us use $\Lc_\ell(n,\kappa_\ell)$ to denote the above defined family of $n\times n$ linear systems from $\Lc$ with bounded spectral tail condition number. We break the proof down into two cases, depending on which of the two terms dominates the value of the max in the lower~bound.
\smallskip

\noindent    
\textbf{Case 1:}  $T_{\Ac,\Lc}(\ell,\infty,\epsilon,\delta) \geq T_{\Ac,\Lc}(n-\ell,\kappa_\ell,\epsilon,\delta)$. We  proceed via a proof by contradiction. Suppose that some $\Alg\in\Ac$ has query complexity less than $T_{\Ac,\Lc}(\ell,\infty,\epsilon,\delta)$ for solving $\Lc_\ell(n,\kappa_\ell)$. We are going to show how to use $\Alg$ to solve \emph{all} $\ell\times \ell$ linear systems in $\Lc$, thereby obtaining a contradiction. Suppose that $(\A,\b,\x_0)$ is an $\ell\times \ell$ linear system in $\Lc$, and define:
\begin{align*}
(\bar\A,\bar\b,\bar\x_0)=\bigg(\begin{bmatrix}
    \A&\zero\\
    \zero& \sigma_{\ell}(\A)\I_{n-\ell}
\end{bmatrix}, 
\begin{bmatrix}
    \b\\ 
    \zero
\end{bmatrix},
\begin{bmatrix}
    \x_0\\ 
    \zero
\end{bmatrix}
\bigg).
\end{align*}
Note that $(\bar\A,\bar\b,\bar\x_0)\in\Lc_\ell(n,\kappa_\ell)$, since $\bar\A$ is positive definite and $\frac{\sigma_{\ell}(\bar\A)}{\sigma_n(\bar\A)}=1\leq\kappa_\ell$. Moreover, if $\Alg$ runs on this system and returns $\bar\x$ such that $\|\bar\x-\bar\x^*\|_{\bar\A}^2\leq\epsilon\|\bar\A\|\|\bar\x_0-\bar\x^*\|^2$, where $\bar\x^*=\bar\A^{-1}\bar\b$, then the vector $\tilde\x$ consisting of the first $\ell$ coordinates of $\bar\x$ satisfies $\|\tilde\x-\x^*\|_{\A}^2\leq\epsilon\|\A\|\|\x_0-\x^*\|^2$, where $\x^*=\A^{-1}\b$. This gives us the contradiction.
\smallskip

\noindent
\textbf{Case 2:} $T_{\Ac,\Lc}(\ell,\infty,\epsilon,\delta) > T_{\Ac,\Lc}(n-\ell,\kappa_\ell,\epsilon,\delta)$. This case follows similarly. Suppose that some $\Alg\in\Ac$ has query complexity less than $T_{\Ac,\Lc}(n-\ell,\kappa_\ell,\epsilon,\delta)$ for solving $\Lc_\ell(n,\kappa_\ell)$. Now, consider an $(n-\ell)\times  (n-\ell)$ linear system task $(\A,\b,\x_0)\in\Lc$ with condition number $\kappa(\A)\leq \kappa_\ell$, and define:
\begin{align*}
(\bar\A,\bar\b,\bar\x_0)=\bigg(\begin{bmatrix}
    \sigma_1(\A)\I_{\ell}&\zero\\
    \zero& \A
\end{bmatrix}, 
\begin{bmatrix}
    \zero\\ 
    \b
\end{bmatrix},
\begin{bmatrix}
    \zero\\ 
    \x_0
\end{bmatrix}
\bigg).
\end{align*}
Note that, as in Case 1, we have $(\bar\A,\bar\b,\bar\x_0)\in\Lc_\ell(n,\kappa_\ell)$, since $\bar\A$ is positive definite and $\frac{\sigma_{\ell}(\bar\A)}{\sigma_n(\bar\A)} = \frac{\sigma_1(\A)}{\sigma_{n-\ell}(\A)}= \kappa(\A)\leq \kappa_\ell$. Now, solving this system to $\epsilon$ accuracy allows us to recover an $\epsilon$ approximate solution for $(\A,\b,\x_0)$, which is a contradiction with the definition of $T_{\Ac,\Lc}(n-\ell,\kappa_\ell,\epsilon,\delta)$. This concludes the proof of the lemma.
\end{proof}
To complete the proof of Theorem \ref{t:lower}, we will rely on the following lower bound for $\MatVec$ algorithms, given by \cite{braverman2020gradient}, which is shown via a reduction from the task of finding the largest eigenvector of a matrix.
\begin{lem}[Theorem 2.3, \cite{braverman2020gradient}]\label{l:braverman2020}
  For all $d\geq c$ and $s\in[c,d]$, where $c$ is a sufficiently large absolute constant, any $\MatVec$ algorithm which satisfies the guarantee:
  \begin{align*}
    \Pr\bigg(\|\tilde\x - \x^*\|_{\A}^2\leq \frac1e \frac{\|\A\|\|\x_0-\x^*\|^2}{s^2}\bigg)\geq 1-\frac1e,
    \quad\text{where }\x^* = \A^{-1}\b,
  \end{align*}
  for all $(d+s^2)$-sparse $d\times d$ positive definite matrices $\A$ with $\kappa(\A)\leq s^2$, and all $\x_0,\b\in\R^d$, must have query complexity at least $\tilde\Omega(s)$.
\end{lem}
\noindent
We are now ready to complete the proof of Theorem \ref{t:lower}.

\begin{proof}[Proof of Theorem \ref{t:lower}]
Thanks to Lemma \ref{l:lower}, it suffices to lower bound $T_{\Ac,\Lc}(\ell,\infty,\epsilon,\delta)$ and $T_{\Ac,\Lc}(n-\ell,\kappa_\ell,\epsilon,\delta)$, with $\epsilon = \frac1e\min\{\ell^{-2},\kappa_\ell^{-1}\}$ and $\delta=1/e$. For the former, we use Lemma~\ref{l:braverman2020} with $d=\ell$ and $s=\ell-1$, obtaining that $T_{\Ac,\Lc}(\ell,\infty,\epsilon,1/e)=\tilde\Omega(\ell)$. On the other hand, for the latter, we use Lemma \ref{l:braverman2020} with $d=n-\ell$ and $s=\sqrt{\kappa_\ell}$ (without loss of generality, we can assume that $\sqrt{\kappa_\ell}\leq n-\ell$), obtaining that $T_{\Ac,\Lc}(n-\ell,\kappa_\ell,\epsilon,\delta)=\tilde\Omega(\sqrt{\kappa_\ell})$. Putting these two lower bounds together, we obtain a complexity lower bound of the form $\max\{\tilde\Omega(\ell),\tilde\Omega(\sqrt{\kappa_\ell})\}=\tilde\Omega(\ell+\sqrt{\kappa_\ell})$.
\end{proof}
\smallskip

\noindent
Finally, to recover the claim in Remark \ref{r:lower}, we observe that for deterministic $\MatVec$ algorithms one can rely on classical lower bounds developed by \cite{nemirovsky1983wiley} (e.g., see Theorem~7.2.6), which show that $T_{\Ac,\Lc}(n,\kappa,\epsilon,\delta) = \Omega(\min\{n, \sqrt\kappa\log1/\epsilon\})$ (here, since the algorithm is deterministic, $\delta$ can be any positive probability).

\section{Conclusions}
In this paper, we developed a nuanced framework for analyzing iterative linear system solvers that allows us to obtain sharper convergence guarantees than classical perspectives. By introducing a more flexible condition number $\kappa_\ell$ that depends on the tail of the spectrum of the matrix, we provided a finer-grained analysis than traditional approaches that rely on a single condition number. Our stochastic algorithm, based on the Sketch-and-Project paradigm, provides improved convergence guarantees for solving linear systems in many machine learning settings, particularly for large matrices where the top portion of the spectrum is ill-conditioned, while the tail is controlled.
 We highlighted the significance of these improvements for models such as spiked covariance and kernel ridge regression, where low-rank structure is likely. In addition, we demonstrated a clear separation between the performance of stochastic solvers and traditional matrix-vector product-based methods like the well-known preconditioned conjugate gradient.

\appendix

\section{Gaussian Universality of Sketched Isometric Embeddings}\label{sec:gauss-univ}
In this section, we give the proof of Lemma \ref{l:universality}, which bounds the extreme singular values of the random matrix $\S\U$, where $\S$ is a $k\times m$ LESS-uniform embedding (see Definition \ref{d:less}) and $\U$ is an $m\times d$ matrix such that $\U^\top\U=\I$ (an isometric embedding) with row norms of $\U$ bounded by $C\sqrt{d/m}$ for some absolute constant $C$. Essentially, we show that as long as $\S$ has $O(\log^4(d/\delta))$ non-zeros per row, then the extreme singular values of $\S\U$ behave just like the extreme singular values of the corresponding Gaussian matrix.

Our Gaussian universality analysis of the random matrix $\S\U$ follows similarly to the analysis of \cite{chenakkod2023optimal}, who showed that such matrices are subspace embeddings when $k\geq 2d$. Here, the key difference is that we consider the setting where $\S\U$ is a wide matrix (because $k<d$), which means that it cannot be a subspace embedding. Yet the Gaussian universality analysis can still be applied to bound the extreme singular~values. 

This analysis is based on a universality result of \cite{brailovskaya2024universality}, which compares the spectrum of a sum of independent random matrices to that of a Gaussian random matrix with the same mean and covariance structure. In the following, we use $\spec(\X)$ to denote the spectrum of matrix $\X$, and if $\X$ is a random $d\times d$ matrix, we let $\Cov(\X)$ be the $d^2\times d^2$ covariance matrix of the entries of $\X$.
Finally, we define the Hausdorff distance between
two subsets $\mathcal{A}, \mathcal{B} \subseteq \R$ as
$$d_H(\mathcal{A}, \mathcal{B}) := \inf\{\eps > 0 : \mathcal{A} \subseteq \mathcal{B} + [-\eps, \eps] \text{ and } \mathcal{B} \subseteq \mathcal{A} + [-\eps, \eps]\}.
$$
\begin{lem}[Theorem 2.4,  \cite{brailovskaya2024universality}] \label{thm:brailovskayavanhandel}
Given a random matrix model $\X:= \Z_0 + \sum_{i=1}^{n} \Z_i$, where $\Z_0$ is a symmetric deterministic $d \times d$ matrix and $\Z_1, \ldots, \Z_n$ are symmetric independent random matrices with $\E[\Z_i]=0$ and $\|\Z_i\|\leq R$, define the following:
\begin{align*}
	\sigma(\X) & = \Big\|\E[(\X-\E \X)^2]\Big\|^{\frac{1}{2}}\quad\text{and}\quad
	\sigma_*(\X) = 
	\sup \limits _{\|\v\|=\|\w\|=1} \E\Big[
	|\v^\top(\X-\E \X)\w|^2\Big]^{\frac{1}{2}}.
\end{align*}
Let $\G$ be a $d\times d$ symmetric random matrix with jointly Gaussian entries, such that $\E[\G]=\E[\X]$ and $\Cov(\G)=\Cov(\X)$. There is a universal constant $C>0$ such that for any $t \geq 0$,
  \begin{align}\label{eq:haus-dist}
\Pr\Big(d_H\big(\spec(\X),\spec(\G)\big)>C\epsilon(t)\Big)\leq de^{-t},\nonumber\\
\quad\quad \text{where}\quad \epsilon(t) =
	\sigma_*(\X)  t^{\frac{1}{2}} +
	R^{\frac{1}{3}}\sigma(\X)^{\frac{2}{3}} t^{\frac{2}{3}} +
	R t.
\end{align}
\end{lem}
  A result from \cite{chenakkod2023optimal} allows us to bound the variance parameters $\sigma_*$ and $\sigma$ appearing in Lemma~\ref{thm:brailovskayavanhandel}. We will apply Lemma~\ref{thm:brailovskayavanhandel} to a symmetrized version of the $\S\U$ matrix. For any rectangular matrix~$\A$, define
  \begin{align*}\sym(\A) :=
    \begin{bmatrix}
      0 &\A\\
      \A^\top& 0
    \end{bmatrix}.
  \end{align*}
  \begin{lem}[Lemma 3.5, \cite{chenakkod2023optimal}, Covariance Parameters]\label{l:cov-parameters}
Let $\S=\{s_{ij}\}_{i \in [k], j \in [m]}$ be a $k \times m$ random
matrix such that $\E(s_{ij})=0$ and $\Var(s_{ij})=p$ for all $i \in
[k], j \in [m]$, and $\Cov(s_{ij},s_{k\ell})=0$ for any $\{i,t\} \subset
[k], \{j,\ell\} \subset [m]$ and $ (i,j)\neq (t,\ell) $.  Let $\U$ be an arbitrary deterministic matrix such that $\U^\top\U = \I$. We then have 
\begin{align*}
\sigma_*(\sym(\S\U)) \leq 2\sqrt{p} 
\quad \text{and} \quad
\sigma(\sym(\S\U)) \leq \sqrt{pk}.
\end{align*}
\end{lem}
We are now ready to give a proof of Lemma~\ref{l:universality}.

\begin{proof}[Proof of Lemma \ref{l:universality}] 
We are going to apply the universality result to the following matrix, where we use $\u_i^\top$ to denote the $i$th row of $\U$, and $\e_i$ to denote the $i$th standard basis vector:
\begin{align*}
  \X = \sym(\S\U) = \sum_{i \in [k], j \in [s]}\underbrace{\sym\Big(\sqrt{\frac ms}r_{ij}\e_{i} \u_{I_{ij}}^\top \Big)}_{\Z_{ij}},
\end{align*}
It is easy to verify that the sketching matrix $\S=\{s_{ij}\}_{i \in [k], j \in [m]}$ from Definition \ref{d:less} satisfies $\E[s_{ij}]=0$, $\Var[s_{ij}] = 1$ and $\Cov(s_{ij},s_{t\ell})=0$ for any $\{i,t\} \subset [k], \{j,\ell\} \subset [m]$ and $ (i,j)\neq (t,\ell) $. 
Applying Lemma~\ref{l:cov-parameters} with $p=1$, we have $\sigma_*(\sym(\S\U))\leq 2$ and $\sigma(\sym(\S\U))\leq\sqrt k$.  Moreover, using that $\|\u_i\|\leq C\sqrt{d/m}$ where $\u_i$ is the $i$th row of matrix $\U$, we also have:
  \begin{align*}
    \|\Z_{ij}\|\leq \sqrt{\frac{m}{s}}\max_{i}\|\u_{i}\|\leq C\sqrt{\frac{d}{s}}=:R.
  \end{align*}
Now, we apply Lemma \ref{thm:brailovskayavanhandel} with $t = \log(d/\delta)$. For this, observe that the function $\epsilon(t)$ from \eqref{eq:haus-dist} can be bounded as follows:
\begin{align*}
\epsilon(t) 
&= \sigma_*(\X)
    t^{1/2}+R^{1/3}\sigma^{2/3}(\X)t^{2/3}+Rt 
    \\
&=
O\bigg(\frac{d^{1/6}k^{1/3}}{s^{1/6}}\log^{2/3}(d/\delta)
+\sqrt{\frac ds}\log(d/\delta)\bigg),
\end{align*}
so for any $k\leq d/2$, with a sufficiently large absolute constant $C'$, if $s\geq C'\log^4(d/\delta)$ then 
Lemma \ref{thm:brailovskayavanhandel} shows:
\begin{align*}
    \Pr\big[d_H(\spec(\sym(\S\U)),\spec(\sym(\G)))>\sqrt d/6)\big]\leq\delta/2,
\end{align*}
where $\G$ is a $k\times d$ matrix with i.i.d. standard Gaussian entries.
Note that the spectrum of $\sym(\G)$ consists of the $k$ singular values of $\G$, their $k$ negatives, and $d-k$ zeros, and the same is true for the spectrum of $\sym(\S\U)$. This follows because for any matrix $\A$, if $\A=\sum_i\sigma_i\u_i\v_i^\top$ is its singular value decomposition, then we can construct the eigenvectors of $\sym(\A)$ associated with its positive and negative eigenvalues as follows:
 \begin{align*}
   \begin{bmatrix}
     0&\A\\
     \A^\top&0
   \end{bmatrix}
\begin{bmatrix}
\u_i\\ \v_i\end{bmatrix}
=\sigma_i
\begin{bmatrix}
\u_i\\ \v_i\end{bmatrix}
   ,\qquad
   \begin{bmatrix}
     0&\A\\
     \A^\top&0
   \end{bmatrix}
\begin{bmatrix}
-\u_i\\ \v_i\end{bmatrix}
=-\sigma_i
\begin{bmatrix}
-\u_i\\ \v_i\end{bmatrix}
\end{align*}
A standard bound on the extreme singular values of a Gaussian matrix \cite{rudelson2009smallest} implies that:
\begin{align}
    \Pr(\sqrt{d}-\sqrt{k} - t\leq\sigma_{\min}(\G)\leq \sigma_{\max}(\G) \leq \sqrt{d} + \sqrt{k} + t) \leq e^{-t^2/2},
\end{align}
so when $d\geq \log(2/\delta)$ and $k\leq d/2$, we have with probability $1-\delta$ that $\sqrt d/2\leq\sigma_{\min}(\G)\leq\sigma_{\max}(\G)\leq 2\sqrt d$. Thus, combining the Gaussian guarantee with the universality result, with probability $1-2\delta$ we have:
\begin{equation*}
    \sqrt d/3\leq\sigma_{\min}(\S\U)\leq\sigma_{\max}(\S\U)\leq 3\sqrt d.
\end{equation*}
\end{proof}

\section{Linear Systems with Polynomial Spectral Decay}
\label{a:polynomial}
In this section, we discuss how our main results can be applied to derive improved time complexity of solving linear systems with polynomial spectral decay, proving Corollary~\ref{c:polynomial}.  We also discuss and compare how algorithms and analysis from prior works apply to this setting, showing that our approach leads to a new best time complexity for a range of spectral decay profiles.

\begin{proof}[Proof of Corollary \ref{c:polynomial}.]
Recall that we are given an $n\times n$ matrix $\A$ with singular values $\sigma_i=\Theta(i^{-\beta}\sigma_1)$ for $\beta> 1/2$, and our goal is to solve $\A\x=\b$. Our assumption of $\beta>1/2$ stems only from the fact that if $0<\beta\leq 1/2$, then such linear system can be solved in $\tilde O(n^2)$  using a simple stochastic gradient method, and thus the problem is trivial. According to Theorem~\ref{thm:main}, setting sketch size $k=O(n^{\frac1{\omega-1}})$ and $\ell = \Omega(k/\log k)$, we can solve this linear system using Algorithm \ref{alg:main} in time $\tilde O(\bar\kappa_\ell\bar\kappa_{\ell:2k}n^2)$. Thus, it suffices to bound the two condition number quantities. Observe that $\bar\kappa_{\ell:2k} = \tilde O(1)$, since $\sigma_\ell$ and $\sigma_{2k}$ only differ by a factor of $O(\log^\beta(k))$, and moreover:
\begin{align*}
    \bar\kappa_\ell^2 = \frac1{n-\ell}\sum_{i>\ell}\frac{\sigma_i^2}{\sigma_n^2}=\Theta \bigg( n^{2\beta-1}\sum_{i>\ell}i^{-2\beta}\bigg)=\Theta\bigg(\Big(\frac n\ell\Big)^{2\beta-1}\bigg).
\end{align*}
Observing further that $n/\ell = \tilde O\big(n^{\frac{\omega-2}{\omega-1}}\big)$ concludes the proof. 
\end{proof}

\paragraph{Comparison to prior work.}
Next, we discuss algorithms from the most significant prior works, and how they compare in solving linear systems with polynomial spectral decay. First, consider the randomized block Kaczmarz-type solver proposed by \cite{derezinski2023solving}. By choosing sketch size $k=O(n^{\frac1{\omega-1}})$ and $\ell = \Omega(k/\log^3 n)$, they achieve times complexity $\tilde O(\bar\kappa_\ell^2 n^2)$. Following the same derivation as in our proof of Corollary \ref{c:polynomial}, this become $\tilde O(n^{2+\frac{\omega-2}{\omega-1}(2\beta-1)})$, which is worse than our result for any $\beta>0.5$.

Next, we consider a preconditioning-based approach. Here, the strategy is to first construct an $\ell$-rank approximation of $\A$ via block power iteration \cite{halko2011finding}. This can be done by starting with a random Gaussian $n\times \tilde O(\ell)$ matrix $\mathbf{\Omega}$, then repeatedly multiplying it with $\A$ and $\A^\top$, and finally, orthogonalizing the resulting matrix $(\A\A^\top)^q\A\mathbf{\Omega}$ to obtain matrix $\Q$. Then, the matrix $\Q\Q^\top\A$ contains accurate approximation of the top-$\ell$ part of $\A$'s spectrum, which can be used to construct a preconditioner of the linear system. For example, \cite{gonen2016solving,musco2018spectrum} considered an SVRG-type solver that, after being preconditioned in such a way, can solve a linear system in $\tilde O(\bar\kappa_\ell n^2)$ time. The cost of the preconditioning is dominated by the cost of the matrix multiplication, which is $\tilde O(n^2\ell)$ using the classical algorithm. This can be accelerated via fast rectangular matrix multiplication algorithms \cite{le2012faster}, obtaining $\tilde O(n^{2+\max\{0,(\omega-2)\frac{\gamma-\alpha}{1-\alpha}\}})$, where $\alpha\approx 0.32$ is the current value of the parameter of fast rectangular matrix multiplication, while $\gamma = \log_n(\ell)$. Thus, the overall cost of the procedure, combining preconditioning and solving, is $\tilde O(n^{2+\max\{(\omega-2)\frac{\gamma-\alpha}{1-\alpha}, (1-\gamma)(\beta-0.5)\}})$, where $\gamma$ can be chosen from the range $[\alpha,1]$. After optimizing over $\gamma$, we can compare this to our Corollary \ref{c:polynomial}, concluding that our guarantee is better for any $\beta\in(0.5,1.33)$. 

\paragraph{Positive definite matrices.}
We note that Corollary \ref{c:polynomial}, as well as the above derivations for prior works, can  be easily adapted to the setting when $\A$ is known to be positive definite, e.g., if it is a kernel matrix \cite{RasmussenWilliams06}. In this case, our runtime guarantee from Corollary \ref{c:polynomial} can be improved to $\tilde O(n^{2+\frac{\omega-2}{\omega-1}(\beta-1)/2})$ for any $\beta>1$, and the task becomes easy for $\beta\leq 1$. Here, one can show analogously as above that our results improve on the best time complexity for any $\beta\in(1,2.66)$.

\paragraph{Kernel ridge regression.} In the setting of kernel ridge regression, we typically introduce an explicit regularization term, resulting in the linear system $(\A+\lambda\I)\x=\b$ \cite{erdogdu2015convergence}. Depending on the value of $\lambda$, this may further reduce the computational cost of our algorithm (as well as the prior works), because it generally allows choosing a smaller value of $\ell$ in the algorithms. This will favor our method over the state-of-the-art, because our guarantee is the best among all approaches that use $\ell=O(n^{\frac1{\omega-1}})$.

\bibliographystyle{alpha}
\bibliography{bib.bib}

\newcommand{\etalchar}[1]{$^{#1}$}
\begin{thebibliography}{DMIMW12}

\bibitem[AC09]{ailon2009fast}
Nir Ailon and Bernard Chazelle.
\newblock The fast {J}ohnson--{L}indenstrauss transform and approximate nearest
  neighbors.
\newblock {\em SIAM Journal on Computing}, 39(1):302--322, 2009.

\bibitem[AL86]{axelsson1986rate}
Owe Axelsson and Gunhild Lindskog.
\newblock On the rate of convergence of the preconditioned conjugate gradient
  method.
\newblock {\em Numerische Mathematik}, 48:499--523, 1986.

\bibitem[ALM24]{alderman2024randomized}
Seth~J Alderman, Roan~W Luikart, and Nicholas~F Marshall.
\newblock Randomized {K}aczmarz with geometrically smoothed momentum.
\newblock {\em arXiv preprint arXiv:2401.09415}, 2024.

\bibitem[AM15]{alaoui2015fast}
Ahmed Alaoui and Michael~W Mahoney.
\newblock Fast randomized kernel ridge regression with statistical guarantees.
\newblock {\em Advances in Neural Information Processing Systems}, 28, 2015.

\bibitem[AZ18]{allen2018katyusha}
Zeyuan Allen-Zhu.
\newblock Katyusha: The first direct acceleration of stochastic gradient
  methods.
\newblock {\em Journal of Machine Learning Research}, 18(221):1--51, 2018.

\bibitem[BCW24]{bollapragada2024fast}
Raghu Bollapragada, Tyler Chen, and Rachel Ward.
\newblock On the fast convergence of minibatch heavy ball momentum.
\newblock {\em {IMA} {J}ournal of {N}umerical Analysis}, page drae033, 2024.

\bibitem[BHSW20]{braverman2020gradient}
Mark Braverman, Elad Hazan, Max Simchowitz, and Blake Woodworth.
\newblock The gradient complexity of linear regression.
\newblock In {\em Conference on Learning Theory}, pages 627--647. PMLR, 2020.

\bibitem[BRVDW19]{burt2019rates}
David Burt, Carl~Edward Rasmussen, and Mark Van Der~Wilk.
\newblock Rates of convergence for sparse variational {G}aussian process
  regression.
\newblock In {\em International Conference on Machine Learning}, pages
  862--871. PMLR, 2019.

\bibitem[BvH24]{brailovskaya2024universality}
Tatiana Brailovskaya and Ramon van Handel.
\newblock Universality and sharp matrix concentration inequalities.
\newblock {\em Geometric and Functional Analysis}, pages 1--105, 2024.

\bibitem[CCFC04]{charikar2004finding}
Moses Charikar, Kevin Chen, and Martin Farach-Colton.
\newblock Finding frequent items in data streams.
\newblock {\em Theoretical Computer Science}, 312(1):3--15, 2004.

\bibitem[CDDR24]{chenakkod2023optimal}
Shabarish Chenakkod, Micha{\l} Derezi{\'n}ski, Xiaoyu Dong, and Mark Rudelson.
\newblock Optimal embedding dimension for sparse subspace embeddings.
\newblock In {\em 56th Annual ACM Symposium on Theory of Computing}, 2024.

\bibitem[CDMF09]{capitaine2009largest}
Mireille Capitaine, Catherine Donati-Martin, and Delphine F{\'e}ral.
\newblock The largest eigenvalues of finite rank deformation of large wigner
  matrices: Convergence and nonuniversality of the fluctuations.
\newblock {\em The {A}nnals of {P}robability}, pages 1--47, 2009.

\bibitem[CEM{\etalchar{+}}15]{cohen2015dimensionality}
Michael~B Cohen, Sam Elder, Cameron Musco, Christopher Musco, and Madalina
  Persu.
\newblock Dimensionality reduction for k-means clustering and low rank
  approximation.
\newblock In {\em Proceedings of the forty-seventh annual ACM symposium on
  Theory of computing}, pages 163--172, 2015.

\bibitem[CMM17]{cohen2017input}
Michael~B Cohen, Cameron Musco, and Christopher Musco.
\newblock Input sparsity time low-rank approximation via ridge leverage score
  sampling.
\newblock In {\em Proceedings of the Twenty-Eighth Annual ACM-SIAM Symposium on
  Discrete Algorithms}, pages 1758--1777. SIAM, 2017.

\bibitem[CMW13]{cai2013sparse}
T~Tony Cai, Zongming Ma, and Yihong Wu.
\newblock Sparse {PCA}: Optimal rates and adaptive estimation.
\newblock {\em The {A}nnals of {S}tatistics}, 41(6):3074--3110, 2013.

\bibitem[CSN09]{clauset2009power}
Aaron Clauset, Cosma~Rohilla Shalizi, and Mark~EJ Newman.
\newblock Power-law distributions in empirical data.
\newblock {\em SIAM review}, 51(4):661--703, 2009.

\bibitem[CW87]{coppersmith1987matrix}
Don Coppersmith and Shmuel Winograd.
\newblock Matrix multiplication via arithmetic progressions.
\newblock In {\em Proceedings of the nineteenth annual ACM symposium on Theory
  of computing}, pages 1--6, 1987.

\bibitem[DCMW19]{derezinski2019minimax}
Micha{\l} Derezi{\'n}ski, Kenneth~L Clarkson, Michael~W Mahoney, and Manfred~K
  Warmuth.
\newblock Minimax experimental design: Bridging the gap between statistical and
  worst-case approaches to least squares regression.
\newblock In {\em Conference on Learning Theory}, pages 1050--1069. PMLR, 2019.

\bibitem[DDH07]{demmel2007fast}
James Demmel, Ioana Dumitriu, and Olga Holtz.
\newblock Fast linear algebra is stable.
\newblock {\em Numerische Mathematik}, 108(1):59--91, 2007.

\bibitem[DFB17]{dieuleveut2017harder}
Aymeric Dieuleveut, Nicolas Flammarion, and Francis Bach.
\newblock Harder, better, faster, stronger convergence rates for least-squares
  regression.
\newblock {\em Journal of Machine Learning Research}, 18(101):1--51, 2017.

\bibitem[DKM20]{derezinski2020improved}
Micha{\l} Derezi\'nski, Rajiv Khanna, and Michael~W Mahoney.
\newblock Improved guarantees and a multiple-descent curve for column subset
  selection and the {N}ystr\"{o}m method.
\newblock {\em Advances in Neural Information Processing Systems},
  33:4953--4964, 2020.

\bibitem[DLDM21]{derezinski2021sparse}
Micha{\l} Derezi\'nski, Zhenyu Liao, Edgar Dobriban, and Michael Mahoney.
\newblock Sparse sketches with small inversion bias.
\newblock In {\em Conference on Learning Theory}, pages 1467--1510. PMLR, 2021.

\bibitem[DLPM21]{derezinski2021newton}
Micha{\l} Derezi{\'n}ski, Jonathan Lacotte, Mert Pilanci, and Michael~W
  Mahoney.
\newblock Newton-{LESS}: Sparsification without trade-offs for the sketched
  {N}ewton update.
\newblock In {\em Advances in Neural Information Processing Systems},
  volume~34, pages 2835--2847. Curran Associates, Inc., 2021.

\bibitem[DM16]{drineas2016randnla}
Petros Drineas and Michael~W Mahoney.
\newblock Rand{NLA}: randomized numerical linear algebra.
\newblock {\em Communications of the ACM}, 59(6):80--90, 2016.

\bibitem[DM21]{derezinski2021determinantal}
Micha{\l} Derezi\'nski and Michael~W Mahoney.
\newblock Determinantal point processes in randomized numerical linear algebra.
\newblock {\em Notices of the American Mathematical Society}, 68(1):34--45,
  2021.

\bibitem[DM24]{derezinski2024recent}
Micha{\l} Derezi{\'n}ski and Michael~W Mahoney.
\newblock Recent and upcoming developments in randomized numerical linear
  algebra for machine learning.
\newblock In {\em Proceedings of the 30th ACM SIGKDD Conference on Knowledge
  Discovery and Data Mining}, pages 6470--6479, 2024.

\bibitem[DMIMW12]{drineas2012fast}
Petros Drineas, Malik Magdon-Ismail, Michael~W Mahoney, and David~P Woodruff.
\newblock Fast approximation of matrix coherence and statistical leverage.
\newblock {\em The Journal of Machine Learning Research}, 13(1):3475--3506,
  2012.

\bibitem[DMM06]{drineas2006sampling}
Petros Drineas, Michael~W Mahoney, and S~Muthukrishnan.
\newblock Sampling algorithms for $\ell_2$ regression and applications.
\newblock In {\em Proceedings of the seventeenth annual ACM-SIAM symposium on
  Discrete algorithm}, pages 1127--1136, 2006.

\bibitem[DR24]{derezinski2024sharp}
Micha{\l} Derezi{\'n}ski and Elizaveta Rebrova.
\newblock Sharp analysis of sketch-and-project methods via a connection to
  randomized singular value decomposition.
\newblock {\em SIAM Journal on Mathematics of Data Science}, 6(1):127--153,
  2024.

\bibitem[DY24]{derezinski2023solving}
Micha{\l} Derezi{\'n}ski and Jiaming Yang.
\newblock Solving dense linear systems faster than via preconditioning.
\newblock In {\em 56th Annual ACM Symposium on Theory of Computing}, 2024.

\bibitem[EBB{\etalchar{+}}21]{even2021continuized}
Mathieu Even, Rapha{\"e}l Berthier, Francis Bach, Nicolas Flammarion, Pierre
  Gaillard, Hadrien Hendrikx, Laurent Massouli{\'e}, and Adrien Taylor.
\newblock A continuized view on {N}esterov acceleration for stochastic gradient
  descent and randomized gossip.
\newblock {\em arXiv preprint arXiv:2106.07644}, 2021.

\bibitem[EG17]{eikmeier2017revisiting}
Nicole Eikmeier and David~F Gleich.
\newblock Revisiting power-law distributions in spectra of real world networks.
\newblock In {\em Proceedings of the 23rd ACM SIGKDD international conference
  on knowledge discovery and data mining}, pages 817--826, 2017.

\bibitem[EM15]{erdogdu2015convergence}
Murat~A Erdogdu and Andrea Montanari.
\newblock Convergence rates of sub-sampled newton methods.
\newblock {\em Advances in Neural Information Processing Systems}, 28, 2015.

\bibitem[FGKS15]{frostig2015regularizing}
Roy Frostig, Rong Ge, Sham Kakade, and Aaron Sidford.
\newblock Un-regularizing: approximate proximal point and faster stochastic
  algorithms for empirical risk minimization.
\newblock In {\em International Conference on Machine Learning}, pages
  2540--2548. PMLR, 2015.

\bibitem[GHRS18]{gower2018accelerated}
Robert Gower, Filip Hanzely, Peter Richt{\'a}rik, and Sebastian~U Stich.
\newblock Accelerated stochastic matrix inversion: general theory and speeding
  up {BFGS} rules for faster second-order optimization.
\newblock {\em Advances in Neural Information Processing Systems}, 31, 2018.

\bibitem[GIG22]{gazagnadou2022ridgesketch}
Nidham Gazagnadou, Mark Ibrahim, and Robert~M Gower.
\newblock Ridgesketch: A fast sketching based solver for large scale ridge
  regression.
\newblock {\em SIAM Journal on Matrix Analysis and Applications},
  43(3):1440--1468, 2022.

\bibitem[GKLR19]{gower2019rsn}
Robert Gower, Dmitry Kovalev, Felix Lieder, and Peter Richt{\'a}rik.
\newblock {RSN}: randomized subspace newton.
\newblock {\em Advances in Neural Information Processing Systems}, 32, 2019.

\bibitem[GOSS16]{gonen2016solving}
Alon Gonen, Francesco Orabona, and Shai Shalev-Shwartz.
\newblock Solving ridge regression using sketched preconditioned svrg.
\newblock In {\em International conference on machine learning}, pages
  1397--1405. PMLR, 2016.

\bibitem[GR15a]{gower2015randomized}
Robert~M Gower and Peter Richt{\'a}rik.
\newblock Randomized iterative methods for linear systems.
\newblock {\em SIAM Journal on Matrix Analysis and Applications},
  36(4):1660--1690, 2015.

\bibitem[GR15b]{gower2015stochastic}
Robert~Mansel Gower and Peter Richt{\'a}rik.
\newblock Stochastic dual ascent for solving linear systems.
\newblock {\em arXiv preprint arXiv:1512.06890}, 2015.

\bibitem[GR17]{gower2017randomized}
Robert~M Gower and Peter Richt{\'a}rik.
\newblock Randomized quasi-{N}ewton updates are linearly convergent matrix
  inversion algorithms.
\newblock {\em SIAM Journal on Matrix Analysis and Applications},
  38(4):1380--1409, 2017.

\bibitem[GV61]{golub1961chebyshev}
Gene~H Golub and Richard~S Varga.
\newblock Chebyshev semi-iterative methods, successive overrelaxation iterative
  methods, and second order {R}ichardson iterative methods.
\newblock {\em Numerische Mathematik}, 3(1):157--168, 1961.

\bibitem[HDNR20]{hanzely2020stochastic}
Filip Hanzely, Nikita Doikov, Yurii Nesterov, and Peter Richtarik.
\newblock Stochastic subspace cubic {N}ewton method.
\newblock In {\em International Conference on Machine Learning}, pages
  4027--4038. PMLR, 2020.

\bibitem[HKPV06]{dpp-independence}
J~Ben Hough, Manjunath Krishnapur, Yuval Peres, and B{\'a}lint Vir{\'a}g.
\newblock Determinantal processes and independence.
\newblock {\em Probability surveys}, 3:206--229, 2006.

\bibitem[HLN07]{hachem2007deterministic}
Walid Hachem, Philippe Loubaton, and Jamal Najim.
\newblock {Deterministic equivalents for certain functionals of large random
  matrices}.
\newblock {\em The Annals of Applied Probability}, 17(3):875 -- 930, 2007.

\bibitem[HMR18]{hanzely2018sega}
Filip Hanzely, Konstantin Mishchenko, and Peter Richt{\'a}rik.
\newblock {SEGA}: Variance reduction via gradient sketching.
\newblock {\em Advances in Neural Information Processing Systems}, 31, 2018.

\bibitem[HMT11]{halko2011finding}
Nathan Halko, Per-Gunnar Martinsson, and Joel~A Tropp.
\newblock Finding structure with randomness: Probabilistic algorithms for
  constructing approximate matrix decompositions.
\newblock {\em SIAM review}, 53(2):217--288, 2011.

\bibitem[Hou73]{hounsfield_CAT}
G.N. Hounsfield.
\newblock Computerized transverse axial scanning (tomography): Part {I}.
  {D}escription of the system.
\newblock {\em British J. Radiol.}, 46:1016–1022, 1973.

\bibitem[HS{\etalchar{+}}52]{hestenes1952methods}
Magnus~Rudolph Hestenes, Eduard Stiefel, et~al.
\newblock {\em Methods of conjugate gradients for solving linear systems},
  volume~49.
\newblock NBS Washington, DC, 1952.

\bibitem[Joh01]{johnstone2001distribution}
Iain~M Johnstone.
\newblock On the distribution of the largest eigenvalue in principal components
  analysis.
\newblock {\em The Annals of {S}tatistics}, 29(2):295--327, 2001.

\bibitem[Kac37]{Kac37:Angenaeherte-Aufloesung}
S.~Kaczmarz.
\newblock Angen\"aherte aufl\"osung von systemen linearer gleichungen.
\newblock {\em Bull. Int. Acad. Polon. Sci. Lett. Ser. A}, pages 335--357,
  1937.

\bibitem[KT12]{kulesza2012determinantal}
Alex Kulesza and Ben Taskar.
\newblock Determinantal point processes for machine learning.
\newblock {\em Foundations and Trends{\textregistered} in Machine Learning},
  5(2--3):123--286, 2012.

\bibitem[LG12]{le2012faster}
Francois Le~Gall.
\newblock Faster algorithms for rectangular matrix multiplication.
\newblock In {\em 2012 IEEE 53rd annual symposium on foundations of computer
  science}, pages 514--523. IEEE, 2012.

\bibitem[LL10]{leventhal2010randomized}
Dennis Leventhal and Adrian~S Lewis.
\newblock Randomized methods for linear constraints: convergence rates and
  conditioning.
\newblock {\em Mathematics of Operations Research}, 35(3):641--654, 2010.

\bibitem[LMH15]{lin2015universal}
Hongzhou Lin, Julien Mairal, and Zaid Harchaoui.
\newblock A universal catalyst for first-order optimization.
\newblock {\em Advances in neural information processing systems}, 28, 2015.

\bibitem[LPJ{\etalchar{+}}24]{lejeune2024asymptotics}
Daniel LeJeune, Pratik Patil, Hamid Javadi, Richard~G Baraniuk, and Ryan~J
  Tibshirani.
\newblock Asymptotics of the sketched pseudoinverse.
\newblock {\em SIAM Journal on Mathematics of Data Science}, 6(1):199--225,
  2024.

\bibitem[LS13a]{lee2013efficient}
Yin~Tat Lee and Aaron Sidford.
\newblock Efficient accelerated coordinate descent methods and faster
  algorithms for solving linear systems.
\newblock In {\em 2013 ieee 54th annual symposium on foundations of computer
  science}, pages 147--156. IEEE, 2013.

\bibitem[LS13b]{liesen2013krylov}
J{\"o}rg Liesen and Zdenek Strakos.
\newblock {\em Krylov subspace methods: principles and analysis}.
\newblock Numerical Mathematics and Scie, 2013.

\bibitem[LW16]{liu2016accelerated}
Ji~Liu and Stephen Wright.
\newblock An accelerated randomized {K}aczmarz algorithm.
\newblock {\em Mathematics of Computation}, 85(297):153--178, 2016.

\bibitem[MDK20]{mdk20}
Mojmir Mutny, Micha{\l} Derezi\'nski, and Andreas Krause.
\newblock Convergence analysis of block coordinate algorithms with
  determinantal sampling.
\newblock In {\em International Conference on Artificial Intelligence and
  Statistics}, pages 3110--3120. PMLR, 2020.

\bibitem[MMS18]{musco2018stability}
Cameron Musco, Christopher Musco, and Aaron Sidford.
\newblock Stability of the {L}anczos method for matrix function approximation.
\newblock In {\em Proceedings of the Twenty-Ninth Annual ACM-SIAM Symposium on
  Discrete Algorithms}, pages 1605--1624. SIAM, 2018.

\bibitem[MNS{\etalchar{+}}18]{musco2018spectrum}
Cameron Musco, Praneeth Netrapalli, Aaron Sidford, Shashanka Ubaru, and David~P
  Woodruff.
\newblock Spectrum approximation beyond fast matrix multiplication: Algorithms
  and hardness.
\newblock In {\em 9th Innovations in Theoretical Computer Science Conference
  (ITCS 2018)}. Schloss-Dagstuhl-Leibniz Zentrum f{\"u}r Informatik, 2018.

\bibitem[MSM14]{msm14}
Xiangrui Meng, Michael~A Saunders, and Michael~W Mahoney.
\newblock {LSRN}: A parallel iterative solver for strongly over-or
  underdetermined systems.
\newblock {\em SIAM Journal on Scientific Computing}, 36(2):C95--C118, 2014.

\bibitem[MT20]{martinsson2020randomized}
Per-Gunnar Martinsson and Joel~A Tropp.
\newblock Randomized numerical linear algebra: Foundations and algorithms.
\newblock {\em Acta Numerica}, 29:403--572, 2020.

\bibitem[Nat01]{natterer}
Frank Natterer.
\newblock {\em The mathematics of computerized tomography}.
\newblock SIAM, 2001.

\bibitem[Nes83]{nesterov1983method}
Y.~E. Nesterov.
\newblock A method of solving a convex programming problem with convergence
  rate $o(1/k^2)$.
\newblock {\em Doklady Akademii Nauk SSSR}, 269(3):543, 1983.

\bibitem[Nes12]{nesterov2012efficiency}
Yurii Nesterov.
\newblock Efficiency of coordinate descent methods on huge-scale optimization
  problems.
\newblock {\em SIAM Journal on Optimization}, 22(2):341--362, 2012.

\bibitem[NN13]{nelson2013osnap}
Jelani Nelson and Huy~L Nguy{\^e}n.
\newblock {OSNAP}: Faster numerical linear algebra algorithms via sparser
  subspace embeddings.
\newblock In {\em 2013 IEEE 54th annual Symposium on Foundations of Computer
  Science}, pages 117--126. IEEE, 2013.

\bibitem[NP06]{nesterov2006cubic}
Yurii Nesterov and Boris~T Polyak.
\newblock Cubic regularization of newton method and its global performance.
\newblock {\em Mathematical Programming}, 108(1):177--205, 2006.

\bibitem[NRP19]{necoara2019randomized}
Ion Necoara, Peter Richt{\'a}rik, and Andrei Patrascu.
\newblock Randomized projection methods for convex feasibility: Conditioning
  and convergence rates.
\newblock {\em SIAM Journal on Optimization}, 29(4):2814--2852, 2019.

\bibitem[NY83]{nemirovsky1983wiley}
AS~Nemirovsky and DB~Yudin.
\newblock {\em Problem complexity and method efficiency in optimization}.
\newblock John Wiley \& Sons, Inc New York, 1983.

\bibitem[OT18]{oymak2018universality}
Samet Oymak and Joel~A Tropp.
\newblock Universality laws for randomized dimension reduction, with
  applications.
\newblock {\em Information and Inference: A Journal of the IMA}, 7(3):337--446,
  2018.

\bibitem[Pan84]{pan1984multiply}
Victor Pan.
\newblock {\em How to multiply matrices faster}.
\newblock Springer-Verlag, 1984.

\bibitem[PV21]{peng2021solving}
Richard Peng and Santosh Vempala.
\newblock Solving sparse linear systems faster than matrix multiplication.
\newblock In {\em Proceedings of the 2021 ACM-SIAM symposium on discrete
  algorithms (SODA)}, pages 504--521. SIAM, 2021.

\bibitem[PWBM18]{perry2018optimality}
Amelia Perry, Alexander~S Wein, Afonso~S Bandeira, and Ankur Moitra.
\newblock Optimality and sub-optimality of pca i: Spiked random matrix models.
\newblock {\em The Annals of Statistics}, 46(5):2416--2451, 2018.

\bibitem[RK20]{rk20}
Anton Rodomanov and Dmitry Kropotov.
\newblock A randomized coordinate descent method with volume sampling.
\newblock {\em SIAM Journal on Optimization}, 30(3):1878--1904, 2020.

\bibitem[RM11]{rubio2011spectral}
Francisco Rubio and Xavier Mestre.
\newblock Spectral convergence for a general class of random matrices.
\newblock {\em Statistics and Probability Letters}, 81(5):592--602, 2011.

\bibitem[RN21]{rebrova2021block}
Elizaveta Rebrova and Deanna Needell.
\newblock On block {G}aussian sketching for the {K}aczmarz method.
\newblock {\em Numerical Algorithms}, 86:443--473, 2021.

\bibitem[RR07]{rahimi2007random}
Ali Rahimi and Benjamin Recht.
\newblock Random features for large-scale kernel machines.
\newblock {\em Advances in neural information processing systems}, 20, 2007.

\bibitem[RT08]{rokhlin2008fast}
Vladimir Rokhlin and Mark Tygert.
\newblock A fast randomized algorithm for overdetermined linear least-squares
  regression.
\newblock {\em Proceedings of the National Academy of Sciences},
  105(36):13212--13217, 2008.

\bibitem[RT20]{richtarik2020stochastic}
Peter Richt{\'a}rik and Martin Tak{\'a}c.
\newblock Stochastic reformulations of linear systems: algorithms and
  convergence theory.
\newblock {\em SIAM Journal on Matrix Analysis and Applications},
  41(2):487--524, 2020.

\bibitem[RV09]{rudelson2009smallest}
Mark Rudelson and Roman Vershynin.
\newblock Smallest singular value of a random rectangular matrix.
\newblock {\em Communications on Pure and Applied Mathematics},
  62(12):1707--1739, 2009.

\bibitem[RW06]{RasmussenWilliams06}
C.~E. Rasmussen and C.~K.~I. Williams.
\newblock {\em Gaussian Processes for Machine Learning}.
\newblock MIT Press, 2006.

\bibitem[Saa81]{saad1981krylov}
Yousef Saad.
\newblock Krylov subspace methods for solving large unsymmetric linear systems.
\newblock {\em Mathematics of computation}, 37(155):105--126, 1981.

\bibitem[Sar06]{sarlos2006improved}
Tamas Sarlos.
\newblock Improved approximation algorithms for large matrices via random
  projections.
\newblock In {\em 2006 47th annual IEEE symposium on foundations of computer
  science (FOCS'06)}, pages 143--152. IEEE, 2006.

\bibitem[SB95]{silverstein1995empirical}
J.W. Silverstein and Z.D. Bai.
\newblock On the empirical distribution of eigenvalues of a class of large
  dimensional random matrices.
\newblock {\em Journal of Multivariate Analysis}, 54(2):175--192, 1995.

\bibitem[SHS01]{sensors}
Andreas Savvides, Chih-Chieh Han, and Mani~B. Strivastava.
\newblock Dynamic fine-grained localization in ad-hoc networks of sensors.
\newblock In {\em Proceedings of the 7th Annual International Conference on
  Mobile Computing and Networking}, MobiCom '01, page 166–179, New York, NY,
  USA, 2001. Association for Computing Machinery.

\bibitem[Str69]{Strassen1969}
Volker Strassen.
\newblock Gaussian elimination is not optimal.
\newblock {\em Numerische Mathematik}, 13(4):354--356, Aug 1969.

\bibitem[SV09]{SV09:Randomized-Kaczmarz}
T.~Strohmer and R.~Vershynin.
\newblock A randomized {K}aczmarz algorithm with exponential convergence.
\newblock {\em J. Fourier Anal. Appl.}, 15(2):262--278, 2009.

\bibitem[SW09]{sw09}
Daniel~A Spielman and Jaeoh Woo.
\newblock A note on preconditioning by low-stretch spanning trees.
\newblock {\em arXiv preprint arXiv:0903.2816}, 2009.

\bibitem[SZW{\etalchar{+}}97]{Santa97gaussianregression}
Huaiyu~Zhu Santa, Huaiyu Zhu, Christopher K.~I. Williams, Richard Rohwer, and
  Michal Morciniec.
\newblock Gaussian regression and optimal finite dimensional linear models.
\newblock In {\em Neural Networks and Machine Learning}, pages 167--184.
  Springer-Verlag, 1997.

\bibitem[Tro11]{tropp2011improved}
Joel~A Tropp.
\newblock Improved analysis of the subsampled randomized {H}adamard transform.
\newblock {\em Advances in Adaptive Data Analysis}, 3(01n02):115--126, 2011.

\bibitem[TYZL23]{tang2023sketch}
Ling Tang, Yajie Yu, Yanjun Zhang, and Hanyu Li.
\newblock Sketch-and-project methods for tensor linear systems.
\newblock {\em Numerical Linear Algebra with Applications}, 30(2):e2470, 2023.

\bibitem[Ver18]{vershynin_2018}
Roman Vershynin.
\newblock {\em High-Dimensional Probability: An Introduction with Applications
  in Data Science}.
\newblock Cambridge Series in Statistical and Probabilistic Mathematics.
  Cambridge University Press, 2018.

\bibitem[VH14]{v14}
Ramon Van~Handel.
\newblock Probability in high dimension.
\newblock {\em Lecture Notes (Princeton University)}, 2014.

\bibitem[Wil12]{williams2012multiplying}
Virginia~Vassilevska Williams.
\newblock Multiplying matrices faster than coppersmith-winograd.
\newblock In {\em Proceedings of the forty-fourth annual ACM symposium on
  Theory of computing}, pages 887--898, 2012.

\bibitem[WKM{\etalchar{+}}08]{wolters2008numerical}
Carsten~H Wolters, Harald K{\"o}stler, Christian M{\"o}ller, Jochen
  H{\"a}rdtlein, Lars Grasedyck, and Wolfgang Hackbusch.
\newblock Numerical mathematics of the subtraction method for the modeling of a
  current dipole in {EEG} source reconstruction using finite element head
  models.
\newblock {\em SIAM Journal on Scientific Computing}, 30(1):24--45, 2008.

\bibitem[Woo14]{woodruff2014sketching}
David~P Woodruff.
\newblock Sketching as a tool for numerical linear algebra.
\newblock {\em Foundations and Trends{\textregistered} in Theoretical Computer
  Science}, 10(1--2):1--157, 2014.

\bibitem[WS01]{Williams01Nystrom}
Christopher K.~I. Williams and Matthias Seeger.
\newblock Using the {N}ystr\"{o}m method to speed up kernel machines.
\newblock In T.~K. Leen, T.~G. Dietterich, and V.~Tresp, editors, {\em Advances
  in Neural Information Processing Systems 13}, pages 682--688. MIT Press,
  2001.

\bibitem[WXXZ24]{williams2024multiplication}
Virginia~Vassilevska Williams, Yinzhan Xu, Zixuan Xu, and Renfei Zhou.
\newblock New bounds for matrix multiplication: from alpha to omega.
\newblock In {\em Proceedings of the 2024 Annual ACM-SIAM Symposium on Discrete
  Algorithms (SODA)}, pages 3792--3835, 2024.

\bibitem[XCGL10]{xia2010superfast}
Jianlin Xia, Shivkumar Chandrasekaran, Ming Gu, and Xiaoye~S Li.
\newblock Superfast multifrontal method for large structured linear systems of
  equations.
\newblock {\em SIAM Journal on Matrix Analysis and Applications},
  31(3):1382--1411, 2010.

\bibitem[XRM20]{xu2020newton}
Peng Xu, Fred Roosta, and Michael~W Mahoney.
\newblock Newton-type methods for non-convex optimization under inexact hessian
  information.
\newblock {\em Mathematical Programming}, 184(1):35--70, 2020.

\bibitem[YLG22]{yuan2022sketched}
Rui Yuan, Alessandro Lazaric, and Robert~M Gower.
\newblock Sketched {N}ewton--{R}aphson.
\newblock {\em SIAM Journal on Optimization}, 32(3):1555--1583, 2022.

\bibitem[YLZ20]{ye2020nesterov}
Haishan Ye, Luo Luo, and Zhihua Zhang.
\newblock Nesterov's acceleration for approximate newton.
\newblock {\em Journal of Machine Learning Research}, 21(142):1--37, 2020.

\end{thebibliography}
\end{document}